\documentclass[12pt]{article}
\usepackage{graphicx}
\def\hybrid{\topmargin 0pt      \oddsidemargin 0pt
        \headheight 0pt \headsep 0pt
       \voffset-1cm
        \textwidth 6.25in       % A4 paper
       \textheight 9.5in       % A4 paper
        \marginparwidth 0.0in
        \parskip 5pt plus 1pt   \jot = 1.5ex}
\catcode`\@=11
\def\marginnote#1{}

\newcount\hour
\newcount\minute
\newtoks\amorpm
\hour=\time\divide\hour by60
\minute=\time{\multiply\hour by60 \global\advance\minute by-\hour}
\edef\standardtime{{\ifnum\hour<12 \global\amorpm={am}%
        \else\global\amorpm={pm}\advance\hour by-12 \fi
        \ifnum\hour=0 \hour=12 \fi
        \number\hour:\ifnum\minute<10 0\fi\number\minute\the\amorpm}}
\edef\militarytime{\number\hour:\ifnum\minute<10 0\fi\number\minute}

\def\draftlabel#1{{\@bsphack\if@filesw {\let\thepage\relax
   \xdef\@gtempa{\write\@auxout{\string
      \newlabel{#1}{{\@currentlabel}{\thepage}}}}}\@gtempa
   \if@nobreak \ifvmode\nobreak\fi\fi\fi\@esphack}
        \gdef\@eqnlabel{#1}}
\def\@eqnlabel{}
\def\@vacuum{}
\def\draftmarginnote#1{\marginpar{\raggedright\scriptsize\tt#1}}

\def\draftlabel#1{{\@bsphack\if@filesw {\let\thepage\relax
   \xdef\@gtempa{\write\@auxout{\string
      \newlabel{#1}{{\@currentlabel}{\thepage}}}}}\@gtempa
   \if@nobreak \ifvmode\nobreak\fi\fi\fi\@esphack}
        \gdef\@eqnlabel{#1}}
\def\@eqnlabel{}
\def\@vacuum{}
\def\draftmarginnote#1{\marginpar{\raggedright\scriptsize\tt#1}}

\def\draft{\oddsidemargin -.5truein
        \def\@oddfoot{\sl preliminary draft \hfil
        \rm\thepage\hfil\sl\today\quad\militarytime}
        \let\@evenfoot\@oddfoot \overfullrule 3pt
        \let\label=\draftlabel
        \let\marginnote=\draftmarginnote
   \def\@eqnnum{(\theequation)\rlap{\kern\marginparsep\tt\@eqnlabel}%
\global\let\@eqnlabel\@vacuum}  }

\def\numberbysection{\@addtoreset{equation}{section}
        \def\theequation{\thesection.\arabic{equation}}}

\def\underline#1{\relax\ifmmode\@@underline#1\else
        $\@@underline{\hbox{#1}}$\relax\fi}

\def\titlepage{\@restonecolfalse\if@twocolumn\@restonecoltrue\onecolumn
     \else \newpage \fi \thispagestyle{empty}\c@page\z@
        \def\thefootnote{\fnsymbol{footnote}} }

\def\endtitlepage{\if@restonecol\twocolumn \else  \fi
        \def\thefootnote{\arabic{footnote}}
        \setcounter{footnote}{0}}  %\c@footnote\z@ }
\relax

\numberbysection
\hybrid

\newfont{\Bbb}{msbm10 scaled 1\@ptsize00}
\newfont{\Bbbb}{msbm7 scaled 1\@ptsize00}
\newcommand{\CC}{\mbox{\Bbb C}}
\newcommand{\CCC}{\mbox{\Bbbb C}}

\newcommand{\DDD}{\raise-1pt\hbox{$\mbox{\Bbbb D}$}}

\newcommand{\SSS}{\mbox{\Bbb S}}

\newcommand{\UUU}{\raise-1pt\hbox{$\mbox{\Bbbb U}$}}

\newcommand{\ZZ}{\mbox{\Bbb Z}}
\newcommand{\z}{\raise-1pt\hbox{$\mbox{\Bbbb Z}$}}
\newcommand{\s}{\raise-1pt\hbox{$\mbox{\Bbbb S}$}}

\def\beq{\begin{equation}}
\def\eeq{\end{equation}}
\def\p{\partial}

\newtheorem{theorem}{Theorem}[section]
\newtheorem{lemma}{Lemma}[section]
\newtheorem{lemma-definition}{Lemma-Definition}[section]
\newtheorem{corollary}{Corollary}[section]

\newtheorem{remark}{Remark}[section]

\newtheorem{definition}{Definition}[section]
\newtheorem{proposition}{Proposition}[section]

\def\normord{ {\scriptstyle {{\bullet}\atop{\bullet}}} }

\def\square{\hfill
{\vrule height6pt width6pt depth1pt} \break \vspace{.01cm}}

\begin{document}

\begin{titlepage}

\title{Bilinear formalism for 
Schwarzian KP \\ and Harry Dym hierarchies}

\author{V. Prokofev\thanks{
Skolkovo Institute of Science and Technology, 143026, Moscow, Russia
and National Research University Higher School of Economics,
20 Myasnitskaya Ulitsa,
Moscow 101000, Russia,
e-mail: vadprokofev@gmail.com}
\and
A.~Zabrodin\thanks{
National Research University Higher School of Economics,
20 Myasnitskaya Ulitsa,
Moscow 101000, Russia; 
%and NRC ``Kurchatov institute'', Moscow, Russia;
e-mail: zabrodin@itep.ru}}

\date{April 2026}
\maketitle

%\vspace{-7cm} \centerline{ \hfill ITEP-TH-15/26}\vspace{7cm}

\begin{abstract}

We consider the Schwarzian KP and
Harry Dym hierarchies in the framework of the bilinear formalism
which is well known for such integrable hierarchies as KP, modified KP,
BKP, Toda lattice and other. We show that, similarly to the
bilinear formulation of the modified KP hierarchy, the 
Schwarzian KP can be reformulated as an integral bilinear
equation for a pair of KP tau-functions
with the property that any linear combination 
of them is again a tau function of the KP hierarchy. 
The Harry Dym hierarchy is then obtained as the Lax-Sato 
formulation of the SchKP one. The close connection
with B\"acklund-Darboux transformations for integrable hierarchies
is also discussed. Besides, it is shown that the SchKP 
hierarchy has a natural embedding
into the multi-component KP hierarchy.

\end{abstract}

\end{titlepage}

\vspace{5mm}

%

%\newpage
\tableofcontents

\vspace{5mm}

\section{Introduction}

In this article, we revisit the  Schwarzian KP (SchKP) and
Harry Dym (HD) integrable hierarchies and
suggest a formulation of them by means of the bilinear
formalism. The main hero is the tau-function satisfying certain
bilinear differential equations (or functional relations).
In a nutshell, we deal with 
the three basic versions of the integrable KP hierarchy
(KP, modified KP (mKP), SchKP) and trace the 
connections between them:
\beq\label{int0}
\mbox{KP}\longleftrightarrow \mbox{modified KP} \longleftrightarrow 
\mbox{Schwarzian KP (Harry Dym)}.
\eeq
For the KP and mKP hierarchies both bilinear and Lax-Sato 
formulations, as well as their equivalence, are well known.
The bilinear formulation of the SchKP hierarchy seems to be new.
In this approach, the HD hierarchy arises as the Lax-Sato
reformulation of the SchKP one. 

The Kadomtsev-Petviashvili (KP) equation was suggested in 1970
\cite{KP,book1}. It has the form
\beq\label{int-kp}
3u_{yy}=\Bigl (4u_t -6uu_x -u_{xxx}\Bigr )_x.
\eeq
This equation is the first member of an infinite
hierarchy of integrable equations known as the KP hierarchy. 
The modern understanding of this hierarchy is based on the works
of the Kyoto school \cite{DJKM83,JM83} (see
also books \cite{Hirota-book,HarnadBalogh}), where the bilinear approach
was developed and its equivalence to the Lax-Sato formulation
was proved. The universal dependent variable is the tau-function
$\tau ({\bf t})$
depending on an infinite number of independent variables (``times''
${\bf t}=\{t_1, t_2, t_3, \ldots \}$ with $t_1=x$, $t_2=y$, $t_3=t$) 
and satisfying an infinite
number of bilinear differential equations. All of them 
are encoded in the
following bilinear integral functional relation:
\begin{equation}
 \label{MEq1}
     \oint_{C_{\infty}}e^{\xi({\bf t}-{\bf t'}, z)}\,
      {\tau}({\bf t}-
   [z^{-1}]) \tau ({\bf t'}+[z^{-1}]) dz=0,
 \end{equation}
which holds for all ${\bf t}$, ${\bf t}'$.
Here the standard notations are used:
\beq\label{not}
\begin{array}{l}
\displaystyle{
\xi ({\bf t}, z)=\sum_{k\geq 1}t_k z^k,} 
\\ \\
{\bf t}\pm [z^{-1}]=\Bigl \{ t_1\pm z^{-1},
\, \frac{1}{2}z^{-2}, \,  \frac{1}{3}z^{-3}, \, 
\ldots \Bigr \}.
\end{array}
\eeq
The contour $C_{\infty}$ is a big circle of radius $R\to \infty$.
The function $u=u(x,y,t)$ in (\ref{int-kp}) is expressed through
the tau-function as $u=2\p_x^2 \log \tau$.

Later some other appearances of the KP equation and the KP hierarchy
were suggested. The most widely known one is
the modified KP (mKP) hierarchy. Its first member
has the form
\beq\label{int-mkp}
v_{xxxx}-4v_{xt}+3v_{yy}+6v_{xx}v_y -6v_x^2v_{xx}=0.
\eeq
In the literature 
there exist several different versions of the mKP hierarchy (see 
\cite{K85,K89,Kiso,KO93,Dickey99,TakTeo06}).
We will mostly deal with the Kupershmidt and Kiso version
\cite{K85,K89,Kiso}. Its bilinear formulation requires  
two tau-functions $\tau ({\bf t})$, $\hat \tau ({\bf t})$
rather than one. Each of them is a KP tau-function
(i.e., satisfies the integral bilinear equation (\ref{MEq1})),
and, in addition, there is a set of bilinear 
equations which link them together. The dependent variable
$v=v(x,y,t)$ in (\ref{int-mkp}) is $v=\log (\hat \tau /\tau )$.
The transition from
$\tau$ to $\hat \tau$ is an example of what is called 
(auto) B\"acklund-Darboux transformation for the KP hierarchy,
i.e., a transformation that sends any solution to another 
solution. A more general point of view is based on the embedding
of the mKP hierarchy into the more general 2D Toda lattice 
hierarchy \cite{UT84}. In this approach, the tau-function 
$\tau (n, {\bf t})$ depends also on
a discrete variable $n\in \ZZ$, and one identifies 
$\tau ({\bf t})=\tau (n, {\bf t})$,  
$\hat \tau ({\bf t})=\tau (n+1, {\bf t})$ for some fixed 
value of $n$.

The Schwarzian KP equation first appeared 
in \cite{Weiss83} in the context of Painlev\'e analysis
as the 
singularity manifold equation associated with the KP equation.
It has the form
\begin{equation}\label{int-schkp}
    \left(4\frac{S_{t}}{S_x}-\{S,x\}\right)_x
    =3\left(\frac{S_{y}}{S_x}\right)_y+\frac{3}{2}
    \left(\frac{S_{y}^2}{S_x^2}\right)_x.
\end{equation}
Here
\beq\label{Sder1}
\{S,x\}=\frac{S_{xxx}}{S_x}-\frac{3}{2}\left(\frac{S_{xx}}{S_x}\right)^2
\eeq
is the Schwarzian derivative.
Equation (\ref{int-schkp}) is invariant
under fractional-linear transformations of the variable $S=S(x,y,t)$.
The corresponding SchKP hierarchy was discussed in
\cite{BK98,BK99,KS02,Schief03}.
 
The original $(1+1)$-dimensional version of the HD equation 
was suggested by Harry Dym in 1973–1974 in 
an unpublished paper.
It is the following nonlinear 
partial differential equation for a function $U(s,t)$ 
of two variables $s, t$:
 \begin{equation}
 \label{HD1}
    4U_{t}=U^3U_{sss}.
    \end{equation}
In the literature, this equation
first appeared in 1975 in Kruskal's paper \cite{Kr75}, 
and received the name of its discoverer. A slightly different version
of the HD equation arises in the the 
Saffman–Taylor problem with surface tension \cite{Kad90}.
Later, in 1984, Konopelchenko and Dubrovskii \cite{KD84} 
suggested a (2+1)-dimensional generalization of this equation, 
\begin{equation}
\label{HD2}
    4U_{t}=U^3U_{sss}-\frac{3}{U}
    \left(U^2\p_s^{-1}\left(\frac{1}{U}\right)_y\right)_y,
    \end{equation}
where $\p^{-1}_s$ is the formal anti-derivative with respect 
    to $s$. If $U$ does not depend on
    $y$, the second term in the right-hand side 
    vanishes and this equation converts into (\ref{HD1}).
One may say that the original HD equation (\ref{HD1}) is related to its
(2+1)-dimensional generalization (\ref{HD2}) in the same way as the KdV equation
is related to the KP equation.
Like in the KP and mKP cases, the HD equation is the first member of an
infinite hierarchy.

It is known that the three hierarchies have commutation representations
of the Lax-Sato type. This formulation uses the well known notion of 
pseudo-differential operators.
Consider
the following pseudo-differential operators:
\begin{equation}
\label{LLL}
\begin{array}{l}
L^{\rm KP}=\partial+u_1\partial^{-1}+u_{2}\partial^{-2}+ \ldots
\\ \\
L^{\rm mKP}=\partial+v_0+v_1\partial^{-1}+v_{2}\partial^{-2}+ \ldots
\\ \\
L^{\rm HD}=U\partial+w_0+w_1\partial^{-1}+w_{2}\partial^{-2} +\ldots ,
\end{array}
\end{equation}
where $\p$ is derivative with respect to some space variable
(which is $x$ for KP and mKP and $s$ in the HD case).
The functions $u_i, v_i$ in the first two lines are functions of $x$
while the functions $U$ and $w_i$ in the 
third line are functions of $s$.
The first and the second ones are the Lax operators for 
 the KP and mKP hierarchies respectively. The third one
 is the Lax operator for the HD hierarchy.
In all the three cases,
 the commuting flows of the hierarchy 
are defined by the following Lax equations:
\begin{equation}\label{Lax3}
     \partial_{t_n}L=[(L^n)_{\geq k},\, L], \quad n\geq 2,
 \end{equation}
where $k=0,\, 1,\, 2$ for the first, second and third cases respectively. 
(Hereafter we use the notation $(\sum\limits_{n}a_n\partial^{n})_{\geq k}=
\sum\limits_{n\geq k}a_n\partial^{n}$.) The Lax equations define 
a dynamics in the space of the pseudo-differential operators
of the form (\ref{LLL}) and the functions $u_i, v_i, w_i$ and $U$ become
functions not only of the space variable ($x$ or $s$) but
also of all the hierarchical times $t_n$.

The three hierarchies are known to be closely related. 
For example, in \cite{OR} B{\"a}cklund-Darboux (BD) transformations 
were constructed between the three hierarchies, in addition to 
some auto-BD transformations.  
 However, despite the fact that the first two hierarchies have been well studied from various points of view (for example, in terms
 of tau-functions and free fermions, 
for the review see, e.g., \cite{AZ12}), 
 the same can not be said about the SchKP and HD hierarchies.
 
 The main goal of this work is to fill some gaps 
 in the theory of the SchKP and HD hierarchies and, in particular, to
 give their formulation in terms of tau-functions 
 and bilinear equations for them. More precisely, 
 we provide a new formulation of the SchKP hierarchy through a
pair of KP tau-functions $(\hat{\tau}({\bf t}),\bar{\tau}({\bf t}))$ 
(recall that where $t_1$ is
identified with the space variable $x$) linked together by an
additional bilinear equation. The main results can be summarized 
as the following theorem.

\begin{theorem}
Let $\hat \tau({\bf t})$ and  $\bar \tau({\bf t})$ be two KP
tau-functions (i.e. solutions of equation (\ref{MEq1})) such that
they satisfy the bilinear equation
\begin{equation}
 \label{MEq}
     \oint_{C_{\infty}}e^{\xi({\bf t}-{\bf t'},z)}\,
      \Bigl ( \hat \tau({\bf t}-[z^{-1}])
      \bar \tau({\bf t'}+ [z^{-1}])+\bar \tau ({\bf t}-[z^{-1}])
      \hat \tau({\bf t'}+ [z^{-1}]) \Bigr ) dz=0
 \end{equation}
for all ${\bf t}$, ${\bf t}'$. Then, for all $a_1, a_2, a_3\in \CC$, 
the function
\beq\label{int-S}
S({\bf t})=\frac{\bar \tau ({\bf t})}{\hat \tau ({\bf t})}
\eeq
solves the functional equation
\beq\label{int-S1}
\frac{(S^{[a_1]} - S^{[a_1 a_2]})(S^{[a_2]} - S^{[a_2 a_3]})
(S^{[a_3]} - S^{[a_1 a_3]})}{(S^{[a_3]} - S^{[a_2 a_3]})
(S^{[a_1]} - S^{[a_1 a_3]})
(S^{[a_2]} - S^{[a_1 a_2]})}=1,
\eeq
where $S^{[a_i]}\equiv S({\bf t}-[a_i^{-1}])$, 
$S^{[a_ia_j]}\equiv S({\bf t}-[a_i^{-1}]-[a_j^{-1}])$. 
In particular, it satisfies
the SchKP equation (\ref{int-schkp}) which is obtained from
(\ref{int-S1}) in the limit $a_1, a_2, a_3\to \infty$.

Conversely, let $S({\bf t})$ be any solution to (\ref{int-S1}),
then there exist two KP tau-functions $\hat \tau$, $\bar \tau$
such that the bilinear equation (\ref{MEq}) holds and the function
$S$ is given by (\ref{int-S}).
\end{theorem}

\noindent
Equation (\ref{int-S1}) is the generating functional equation
of the SchKP hierarchy. Partial differential equations 
of the hierarchy are obtained by expanding its 
left-hand side in the Taylor series in inverse powers 
of $a_1, a_2, a_3$ and equating the coefficients to zero. 
In the leading non-vanishing order one obtains equation
(\ref{int-schkp}). At the same time, equation
(\ref{int-schkp}) can be interpreted as an integrable
discretization of the SchKP equation. Namely, by setting
$S(n_1, n_2, n_3)=S ({\bf t}-n_1[a_1^{-1}]-n_2[a_2^{-1}]
-n_3[a_3^{-1}])$ for any fixed ${\bf t}$, 
(\ref{int-schkp}) converts into a discrete equation for
$S(n_1, n_2, n_3)$.

We also clarify the interrelation between the SchKP and HD 
hierarchies. The latter is nothing else than
a reformulation of 
the former in terms of pseudo-differential operators (i.e., 
an equivalent formulation of the Lax-Sato type). 
This reformulation requires 
a nontrivial change of variables: the new space variable is 
$s=S(x)$, where is $S$ given by (\ref{int-S}), so $\p_x =S_x \p_s$,
and $U(s)$ from (\ref{HD2}) and (\ref{LLL}) is defined
as $U=S_x$.
Besides, some previously known results 
about BD transformations are reformulated in terms of tau-functions.
This allows them to be given a particularly simple and clear form.

Having mentioned BD transformations\footnote{Sometimes they are 
called gauge transformations.} of the integrable hierarchies,
we should say that there is a huge literature devoted to this issue.
(See, for example, \cite{book}-\cite{Z25a} and references therein.)
In turn, we could not but touch upon the issues related to this.
In fact the language of BD transformations is 
especially convenient 
in analysis of the three hierarchies (\ref{int0}) and interrelations
between them. The key notion in the theory of BD transformations
is a wave function\footnote{Some authors
prefer to call them ``eigenfunctions'' but such terminology seems to 
be rather misleading because in general wave functions are {\it not} 
eigenfunctions of Lax operators. The name ``wave function'' that we use
in this paper is motivated by the fact that the linear differential
equation of the form (\ref{wave1}) below for the $t_2$-flow 
formally coincides with the non-stationary Schredinger equation
in imaginary time.} 
for a given Lax operator (the one for KP, mKP or
HD). For example, in the case of KP wave functions
$\phi ({\bf t})$ are defined as
common solutions to the equations
$$
\p_{t_n}\phi ({\bf t})=B_n \phi ({\bf t}),  \quad n\geq 2,
$$
where $B_n$ are differential operators
$B_n =(L^n)_{\geq 0}$. Given a solution to the KP hierarchy
with a tau-function $\tau ({\bf t})$, the function 
$\hat \tau ({\bf t})=\phi ({\bf t}) \tau ({\bf t})$ is then another
solution. This defines an (auto)BD transformation of the KP
hierarchy and, simultaneously, a solution to the mKP herarchy,
which is just given by the pair $(\tau , \hat \tau )$ of KP 
tau-functions. Moreover, any solution to the mKP hierarchy
is defined in this way by a pair $(\tau , \phi )$, 
where
$\tau$ is a KP tau-function and $\phi$ is a wave function associated
with it (more precisely, with the Lax operator corresponding to
this tau-function). This story can be further continued by
considering wave functions for the mKP Lax operator 
$L^{\rm mKP}={\cal L}$
corresponding to the solution $(\tau , \hat \tau )$. Namely,
let $S$ be such wave function, then $\bar \tau =S\hat \tau$
is again a KP tau-function, and the pair $(\hat \tau , \bar \tau )$
solves the bilinear equation (\ref{MEq}). Then $S=\bar \tau /\hat \tau$
solves the SchKP equation (\ref{int-schkp}).
Moreover, any solution to the SchKP hierarchy
is defined in this way by a pair $(\hat \tau , S)$, 
where
$\hat \tau$ is a KP tau-function 
and $S$ is a wave function for the mKP Lax operator associated
with the solution $(\tau , \hat \tau )$.  Equivalently, to 
construct a solution to the SchKP hierarchy, one may start
from two solutions to the mKP hierarchy: 
$(\tau , \hat \tau )$ (defined 
by fixing a wave function $\phi$) and $(\tau , \bar \tau )$
(defined by fixing another wave function $\phi '$)
For clarity, all said above can be represented as the following
diagram:
\beq\label{diag11}
\begin{array}{ccc}
\bar \tau  & \stackrel{\bar \phi}\longrightarrow & \hat{\bar \tau}
\\
\uparrow\lefteqn{{\scriptstyle \phi '} } 
&& \uparrow\lefteqn{{\scriptstyle \hat \phi}}
\\ 
\tau &\stackrel{\phi}\longrightarrow &\hat \tau
\end{array}
\eeq
Note that the vertices correspond to solutions of KP, 
edges correspond to solutions to mKP, and the whole square
(more precisely, the diagonal from $\bar \tau$ to $\hat \tau$)
can be associated to a solution to SchKP (more precisely,
two solutions: $S=\bar \tau/\hat \tau$ and $S^{-1}=
\hat \tau/\bar \tau$).

Lastly, we establish a close connection
with multi-component KP hierarchy \cite{DJKM81a,KL93,TT07,Teo11}. 

\begin{theorem}
Let
$\tau({\bf s},{\bf t})$, where ${\bf s}=\{s_1, s_2, s_3\}$,  
$s_i\in \ZZ$ with the condition
$s_1+s_2+s_3=0$, be the tau-function of the multi-component
KP hierarchy restricted 
to the three-component sector, with all other discrete 
and continuous variables frozen. 
Then the pair of functions 
$(\tau(0,0,0),\tau(0,1,-1))$ solves equation (\ref{MEq}) and thus 
provides a solution to the SchKP and HD hierarchies.
\end{theorem}

The outline of the paper is as follows. 
Section 2 is devoted to the well known KP hierarchy.
We consider it both within the framework 
of bilinear formalism and the Lax-Sato approach, and provide a proof
of their equivalence. All this is a basic preliminary stuff for 
what follows. A similar treatment of the mKP hierarchy
is the subject of Section 3. The main idea here (which is not new),
and also a useful technical tool, 
is realization of the mKP hierarchy as a chain of B\"acklund-Darboux
transformations of solutions to the KP one.
Section 4 is devoted to the SchKP and HD hierarchies.
First, the SchKP hierarchy is obtained from considering 
B\"acklund-Darboux transformations of the mKP hierarchy
(Section 4.1). Second, the bilinear approach to SchKP via a pair
of KP tau-functions is developed. Third, it is shown that the
reformulation in the framework of the Lax-Sato approach 
provides the HD hierarchy. In Section 5 the embedding of 
the SchKP hierarchy into the multi-component hierarchy is
discussed. Section 6 contains some concluding remarks.

\section{The KP hierarchy}

The set of independent variables 
is ${\bf t}=\{t_1, t_2, t_3, \ldots \}$.
The variable $t_1$ can be identified with the space variable $x$,
so hereafter we put $t_1=x$. 
All other variables are interpreted as times associated with
an infinite number of commuting flows.

In the bilinear formalism, 
the universal dependent variable is tau-function
$\tau ({\bf t})$.
The generating bilinear equation for the KP tau-function
is of the form
\beq\label{kp1}
\oint_{C_{\infty}}e^{\xi ({\bf t}-{\bf t}', z)}
\tau ({\bf t}-[z^{-1}])\tau ({\bf t}'+[z^{-1}])\, dz =0,
\eeq
where we the standard notation (\ref{not}) is used.
The contour $C_{\infty}$ is a big circle of radius $R\to \infty$
around $\infty$.
Equation (\ref{kp1}) is valid for all ${\bf t}$ and ${\bf t}'$.
The simplest possible solution (the trivial one) is $\tau ({\bf t})=1$.
Note also that if $\tau ({\bf t})$ is a solution of (\ref{kp1}), then
$e^{\ell ({\bf t})}\tau ({\bf t})$, where $\ell ({\bf t})$ is 
an arbitrary linear function of the times, is a solution, too.
Two tau-functions that differ by such a factor 
are said to be equivalent.

Let us mention one of the most important corollaries of
equation (\ref{kp1}).
Setting 
$$
{\bf t}-{\bf t'}=[a_1^{-1}]+[a_2^{-1}]+[a_3^{-1}],
$$
where $a_1, a_2, a_3 \in \CC$ are three different points, 
we have
$$
e^{\xi ({\bf t}-{\bf t}', z)}=\frac{a_1 \, a_2 \, 
a_3}{(a_1-z)(a_2-z)(a_3-z)},
$$
and the integral in (\ref{kp1}) can be calculated by an elementary
residue calculus. The result is:
\beq\label{kp1b}
\begin{array}{l}
(a_1-a_2)\tau ({\bf t}-[a_1^{-1}]-[a_2^{-1}])
\tau ({\bf t}-[a_3^{-1}])
\\ \\ 
\phantom{aaaaaaaaaaa}+\, 
(a_2-a_3)\tau ({\bf t}-[a_2^{-1}]-[a_3^{-1}])
\tau ({\bf t}-[a_1^{-1}])
\\ \\ 
\phantom{aaaaaaaaaaaaaaaaaaaaaa}+\,
 (a_3-a_1)\tau ({\bf t}-[a_1^{-1}]-[a_3^{-1}])
\tau ({\bf t}-[a_2^{-1}]) =0.
\end{array}
\eeq
This is the famous equation first obtained by Miwa \cite{Miwa82}.
As is proved in \cite{TT95},
it is in fact equivalent to the entire KP hierarchy defined
by the integral equation (\ref{kp1}) (see also \cite{Shigyo13} for the
direct proof). 

To write down such equations in a compact form, it is convenient
to introduce some short-hand notation. Set
\beq\label{not1}
\begin{array}{l}
\tau ({\bf t}-[a_j^{-1}])=\tau^j, \qquad
\tau ({\bf t}-[a_j^{-1}]-[a_k^{-1}])=\tau^{jk}, \quad \mbox{etc},
\\ \\
\tau ({\bf t}+[a_j^{-1}])=\tau_j, \qquad
\tau ({\bf t}+[a_j^{-1}]+[a_k^{-1}])=\tau_{jk}, \quad \mbox{etc}.
\end{array}
\eeq
Similar notation will be used for other functions of ${\bf t}$ 
introduced below (Baker-Akhiezer functions, wave functions, etc).
This notation allows us to rewrite equation (\ref{kp1b}) in
the following compact form:
\beq\label{kp1c}
a_{12}\tau^{12}\tau^3 +a_{23}\tau^{23}\tau^1 +
a_{31}\tau^{13}\tau^2=0,
\eeq
or, after an overall shift of ${\bf t}$,
\beq\label{kp1d}
a_{12}\tau_{12}\tau_3 +a_{23}\tau_{23}\tau_1 +
a_{31}\tau_{13}\tau_2=0
\eeq
(here and below $a_{ij}\equiv a_i-a_j$). 

More generally, fix an integer $m\geq 2$,
$m+1$ distinct points $a_i\in \CC$, $m-1$ 
distinct points $b_i\in \CC$ and make 
the substitution
\beq\label{miwa1}
{\bf t}-{\bf t}'=\sum_{i=1}^{m+1}[a_i^{-1}]-
\sum_{i=1}^{m-1}[b_i^{-1}].
\eeq
Using residue calculus, we convert
(\ref{kp1}) into the following equation:
\beq\label{g2b}
\begin{array}{l}
\displaystyle{
\sum_{k=1}^{m+1}(-1)^{k-1}
\Delta_m (a_1, \ldots , \hat a_k, \ldots , a_{m+1})
\prod_{i=1}^{m-1}(b_i-a_k)}
\\ \\
\phantom{aaaaaaaaaaaaaaaaaa}
\displaystyle{\times \,
\tau \Bigl ({\bf t}-\sum_{j=1, \neq k}^{m+1}[a_j^{-1}]\Bigr )
\tau \Bigl ({\bf t}-[a^{-1}_k]- 
\sum_{j=1}^{m-1}[b_j^{-1}]\Bigr )=0}
\end{array}
\eeq
(equation ({\ref{kp1b}) is its particular case for $m=2$, $b_1=\infty$).
Here
\beq\label{vandermonde}
\Delta_m (a_1, \ldots , a_m)=\det_{1\leq i,j \leq m}(a_i^{m-j})=
\prod_{i<j}^m (a_i-a_j)
\eeq
is the Vandermonde determinant, and $\hat a_k$ means that 
$a_k$ is omitted. Equations (\ref{g2b}) can be written in the form
which resembles Pl\"ucker relations. To this end, introduce the
notation
\beq\label{pl1}
\tau \bigl [a_1, \ldots , a_m\bigr ]=
\Delta_m (a_1, \ldots , a_{m})\tau \Bigl ({\bf t}-
\sum_{j=1}^m [a_j^{-1}]\Bigr ).
\eeq
Clearly, $\tau \bigl [a_1, \ldots , a_m\bigr ]$ 
is antisymmetric with respect to
any permutation $a_i \leftrightarrow a_k$. In this notation,
equation (\ref{g2b}) acquires the form
\beq\label{g2c}
\sum_{k=1}^{m+1} (-1)^{k-1}
\tau \bigl [a_1, \ldots , \hat a_k, \ldots ,a_{m+1}\bigr ]\,
\tau \bigl [a_k, b_1, \ldots , b_{m-1}\bigr ]=0.
\eeq

Other important objects are 
Baker-Akhiezer\footnote{This name is usually used 
in the context of algebraic-geometrical solutions to the KP hierarchy 
related to Riemann surfaces of finite genus.
Here we extend this name to a broader class of solutions.} (BA) function
$\psi ({\bf t}, z)$ and its adjoint (dual) $\psi^* ({\bf t}, z)$.
For any fixed values of the times, they are functions of
a complex variable $z$ which is called {\it spectral parameter}.
For the sake of brevity, we will sometimes refer to both
$\psi ({\bf t}, z)$ and $\psi^* ({\bf t}, z)$ as BA functions,
or simply $\psi$-functions.
In general, these functions, as functions of $z$,
are assumed to be holomorphic in a sufficiently small
neighborhood ${\sf U}_{\infty}$ of $\infty$ 
(i.e., for $|z|>R$ with some sufficiently
large $R$) except maybe the very point $z=\infty$, where an essential
singularity is allowed.
In this neighborhood they can be
expressed through the tau-function as follows:
\beq\label{kp4}
\begin{array}{l}
\displaystyle{
\psi ({\bf t}, z)=e^{\xi ({\bf t}, z)}\, 
\frac{\tau ({\bf t}-[z^{-1}])}{\tau ({\bf t})},}
\\ \\
\displaystyle{
\psi^* ({\bf t}, z)=e^{-\xi ({\bf t}, z)}\, 
\frac{\tau ({\bf t}+[z^{-1}])}{\tau ({\bf t})}.}
\end{array}
\eeq
These remarkable formulas were obtained in \cite{DJKM83,JM83}.

The expansions of the BA functions as $z\to \infty$ are
\beq\label{kp5}
\begin{array}{l}
\displaystyle{
\psi ({\bf t}, z)=e^{\xi ({\bf t}, z)}\, 
\Bigl (1+\xi_1({\bf t})z^{-1}+ \xi_2({\bf t})z^{-2}+\ldots \, \Bigr ),}
\\ \\
\displaystyle{
\psi^* ({\bf t}, z)=e^{-\xi ({\bf t}, z)}\, 
\Bigl (1+\xi^*_1({\bf t})z^{-1}+ \xi^*_2({\bf t})z^{-2}+\ldots \, 
\Bigr ).}
\end{array}
\eeq
The coefficients $\xi_k({\bf t})$, $\xi^*_k({\bf t})$ can be
expressed through logarithmic derivatives of the tau-function
with respect to the $t_k$'s but we do not need these formulas here.
Their dependence on the times is determined by 
a chosen solution of (\ref{kp1}). It is explicitly specified
in terms of partial differential equations given below.
In terms of the BA functions, the bilinear equation (\ref{kp1})
acquires the form
\beq\label{kp1a}
\oint_{C_{\infty}}
\psi ({\bf t}, z)\psi^* ({\bf t'}, z)\, dz =0.
\eeq
It holds for all ${\bf t}$, ${\bf t'}$.

\begin{remark}
\label{remark:psi}
The latter equation for the $\psi$-functions 
can serve as a starting point of the theory. 
In particular, it implies 
existence of the tau-function, i.e., a function $\tau ({\bf t})$
such that the $\psi$-functions are expressed through it by means
of formulas (\ref{kp4}). The proof is given in \cite{DJKM83}.
\end{remark}

It is worth noting that for some important classes
of solutions the BA functions
can be analytically continued outside ${\sf U}_{\infty}$
as analytic functions on $\CC \setminus {\sf S}$, where
${\sf S}\subset \CC$ is a set of points of nonzero 
co-dimension (which is $2$ for a number of isolated points or
$1$ for lines). (The well known examples are soliton 
and rational solutions.) To avoid unnecessary analytical 
complications in this work, we will have in mind 
the classes of solutions for which such analytic
continuation is possible.

\begin{remark}\label{remark:kp2}
For completeness,
let us say some words about the general case, when
the BA functions can not be analytically continued from
${\sf U}_{\infty}$ to bigger
domains of the complex plane. In this case 
$\p_{\bar z}\psi ({\bf t}, z)$ and $\p_{\bar z}\psi^* ({\bf t}, z)$
may be nonzero in some compact 2D domain ${\sf D}\subset \CC$,
and the BA functions (which are functions of both $z$ and 
$\bar z$ for $z\in {\sf D})$ are characterized by certain integral
or integro-diffrential equations which were suggested in
\cite{ZM85} and are known in the literature as non-local
$\bar \p$-problem. They are:
\beq\label{dbar}
\begin{array}{l}
\displaystyle{
\p_{\bar z}\psi ({\bf t}, z)=\int_{\CCC}
K(z, \zeta )\psi ({\bf t}, \zeta )\, d^2 \zeta,
}
\\ \\
\displaystyle{
\p_{\bar z}\psi^* ({\bf t}, z)=-\int_{\CCC}
\psi^* ({\bf t}, \zeta )
K(\zeta , z)\, d^2 \zeta.
}
\end{array}
\eeq
The kernel $K(z, \zeta )$ is assumed to have a compact support
with respect to the both arguments, i.e., $K(z, \zeta )=0$ if
$z$ or $\zeta$ belongs to ${\sf U}_{\infty}$. This condition allows
one to require the holomorphic in $z$ asymptotics as $z\to \infty$
of the form (\ref{kp5}). Under certain conditions imposed on the
kernel these integro-differential equations are known to have 
nontrivial solutions, and the linear space of solutions is in general
one-dimensional.
\end{remark}

It is instructive to think of the BA function as the result 
of the action of a ``dressing operator'' $W$ to the 
simplest possible 
BA function $e^{\xi ({\bf t}, z)}$ corresponding to 
the trivial solution. The operator $W$ is a pseudo-differential
operator of order 0 of the form
\beq\label{W}
W=1+ \xi_1 \p_x^{-1} + \xi_2 \p_x^{-2} +\ldots ,
\eeq
where $\p_x =\p /\p x$ (recall that
we identify $x=t_1$) and 
the coefficient functions $\xi_k=\xi_k({\bf t})$ 
are the same as in the first formula in (\ref{kp5}). 
Clearly, we can write
\beq\label{kp6}
\psi ({\bf t}, z)=W e^{\xi ({\bf t}, z)},
\eeq
where it is implied that $\p_x^{-1}$ formally acts on
$e^{xz}$ as $\p_x^{-1}e^{xz}=z^{-1}e^{xz}$.

The dressing operator allows one to introduce the
Lax operator $L$ of the KP hierarchy by ``dressing'' of the
$\p_x$:
\beq\label{Lax}
L =W \p_x W^{-1} = \p_x + u_1\p_x^{-1}+u_2\p_x^{-2}+\ldots \, .
\eeq
It is a pseudo-differential operator of order 1.
The coefficient functions $u_i$ can be
expressed through logarithmic derivatives of the tau-function.
Along with the Lax operator, we introduce the purely differential
operators
\beq\label{Bk}
B_k=(L^k)_{\geq 0}, \qquad k\geq 1,
\eeq
where for any subset $\SSS \subset \ZZ$ 
the operation $(\ldots )_{\s}$ applied to any 
pseudo-differential operator $\displaystyle{
P=\sum_{j\in \z}p_j \p_x^j}$ is defined as
$$
(P)_{\s}=\Bigl (\sum_j p_j \p_x^j\Bigr )_{\s}=\sum_{j\in \s}p_j\p_x^j.
$$
For example, $B_1=\p_x$, $B_2=\p_x^2 +2u_1$.

Remarkably, the bilinear relation (\ref{kp1a}) for the 
$\psi$-functions alone 
implicitly defines possible
dynamics in the space of the coefficients $u_i$ of the Lax operator
making them certain functions of the times: $u_i=u_i({\bf t})$.
This dynamics can be described by partial differential equations.

\begin{theorem}(\cite{DJKM83}) \label{theorem:kp1}
Let $\psi ({\bf t}, z)$, $\psi^* ({\bf t}, z)$ be functions of the form (\ref{kp5}). If the bilinear relation (\ref{kp1a}) is satisfied, then
there exists a dressing
operator $W$ of the form (\ref{W}) such that
$\psi ({\bf t}, z)=W e^{\xi ({\bf t}, z)}$ satisfying the evolution
equations
\beq\label{kp7c}
\p_{t_k}W=-(L^k)_{<0}W \;\;
\mbox{for any $k\geq 1$}.
\eeq
\end{theorem}

\noindent
{\it Proof.} The proof given below is borrowed from \cite{DJKM83}.
The following technical lemma is 
useful\footnote{It is a simplified version of Lemma 1.1 from
\cite{DJKM83}.}:

\begin{lemma} \label{lemma:technical}
Let $P$ and $Q$ be monic pseudo-differential operators
of the form
$$
P=1+\sum_{j\geq 1}p_j(x) \p_x^{-j}, \quad
Q=1+\sum_{j\geq 1}q_j(x) (-\p_x)^{-j}.
$$
Then imposing the condition
\beq\label{lem}
\oint_{C_{\infty}}(P_x e^{xz})\, (Q_{x'}e^{-x'z}) \, dz =0
\eeq
for all $x,x'$ implies that
\beq\label{lem1}
\Bigl (P\, Q^{\dag}\Bigr )_{<0}=0.
\eeq 
(The notation
$P_x$, $Q_{x'}$ in (\ref{lem}) means that the operators act on functions 
of $x,x'$, respectively.)
\end{lemma}

\noindent
In (\ref{lem1}), $Q^{\dag}$ is the conjugated 
operator. The conjugation $(\ldots )^{\dag}$ of differential
and pseudo-differential operators is
defined as $(f(x)\circ \p_x )^{\dag}=
-\p_x \circ f(x)$, or, more generally,
if $\displaystyle{P=\sum_{k}p_k \circ \p_x^k}$, then 
$\displaystyle{
P^{\dag}_k=\sum_k(-\p_x)^k \circ p_k .}
$
For the proof of the lemma see \cite{DJKM83} (Lemma 1.1).

Set
$$
\psi ({\bf t}, z)=W e^{\xi ({\bf t}, z)},\quad
\psi^* ({\bf t}, z)=V e^{-\xi ({\bf t}, z)}
$$
where both $W$ and $V$ are of the form (\ref{W}). Plugging this
into (\ref{kp1a}) with $t_k=t'_k$ for $k\geq 2$, we have:
$$
\oint_{C_{\infty}} (W_x e^{xz})\,
(V_{x'} e^{-x'z})\, dz =0.
$$
Now we can apply Lemma 
\ref{lemma:technical} with $P=W$, $Q=V$ and conclude that
$(W\, V^{\dag})_{<0}=0$. Since the both operators are monic 
and contain non-positive powers of $\p_x$ only, this means that
$W\, V^{\dag}=1$, i.e., 
\beq\label{V}
V=(W^{\dag})^{-1}.
\eeq
Next, define the operators $L$ and $B_k$ as in (\ref{Lax}) and
(\ref{Bk}). Then (\ref{kp1a}) implies:
\beq\label{kp11}
\oint_{C_{\infty}}  (\p_{t_n}-B_n  )
\psi ({\bf t}, z)\psi^* ({\bf t'}, z)\, dz =0.
\eeq
We have, using (\ref{kp6}) and (\ref{Lax}):
$$
(\p_{t_n}-B_n  )\psi ({\bf t}, z)=
(\p_{t_n}W +W\p_x^n -B_n W)e^{\xi ({\bf t}, z)}
$$
$$
= (\p_{t_n}W +L^n W  -B_n W)e^{\xi ({\bf t}, z)}
=(\p_{t_n}W +(L^n)_{<0} W)e^{\xi ({\bf t}, z)}.
$$
Again, putting $t_k=t_k'$ for $k\geq 2$, we rewrite (\ref{kp11}) as
$$
\oint_{C_{\infty}}  
(R e^{xz })\, ((W^{\dag})^{-1}e^{-x'z})
\, dz =0,
$$
where $R$ is the operator
$
R=\p_{t_n}W +(L^n)_{<0} W
$.
By Lemma \ref{lemma:technical}, we 
conclude that $(R\, W^{-1})_{<0}=0$.
Clearly, $R$ contains strictly negative powers of $\p_x$ only 
and the same is true for the operator
$R\, W^{-1}$, i.e., $(R\, W^{-1})_{<0}=R\, W^{-1}=0$.
Since $W$ is a monic operator, it follows from this that $R=0$,
which proves (\ref{kp7c}).
\square

\begin{corollary} \label{corollary:kp1}
The function
$\psi ({\bf t}, z)$ is an eigenfunction of the Lax operator $L$
(\ref{Lax})
with eigenvalue $z$,
\beq\label{kp7}
L\psi ({\bf t}, z)=z\psi ({\bf t}, z),
\eeq
and for any $z\in \CC$ it satisfies the differential equations
\beq\label{kp7a}
\p_{t_k}\psi ({\bf t}, z)=B_k \psi ({\bf t}, z) \;\;
\mbox{for any $k\geq 1$},
\eeq
where the $k$-th order differential operators $B_k$ of the form
$\displaystyle{
B_k =\p_x^k +\sum_{j=0}^{k-2}b_k \p_x^j}
$
are given by equation (\ref{Bk}). 
\end{corollary}

\noindent
{\it Proof.} Again, we follow \cite{DJKM83}.
The eigenvalue equation (\ref{kp7}) is a trivial consequence 
of (\ref{kp6}), (\ref{Lax}). Equation (\ref{kp7a}) easily 
follows from (\ref{kp7c}):
$$
\begin{array}{lll}
\p_{t_k}\psi ({\bf t}, z)& =& \p_{t_k}W \, e^{\xi ({\bf t}, z)}+
W z^k e^{\xi ({\bf t}, z)}
\\ && \\
&=&-(L^k)_{<0}W  e^{\xi ({\bf t}, z)}
+W \p_x^k e^{\xi ({\bf t}, z)}
\\ && \\
&=&\Bigl (-(L^k)_{<0}+L^k\Bigr )
W  e^{\xi ({\bf t}, z)}=(L^k)_{\geq 0}\psi ({\bf t}, z).
\end{array}
$$
\square

In what follows, by the BA function we mean a function of the form
(\ref{kp5}) such that it satisfies equations (\ref{kp7}) and
(\ref{kp7a}). Theorem \ref{theorem:kp1} implies that
this definition is equivalent to the first formula in (\ref{kp4}),
where the tau-function satisfies (\ref{kp1}).

There are similar statements for the 
adjoint BA function $\psi^* ({\bf t}, z)$.
From the proof of Theorem \ref{theorem:kp1} we already know
that
\beq\label{kp12}
\psi^*({\bf t}, z)=(W^{\dag})^{-1} e^{-\xi ({\bf t}, z)}.
\eeq
An easy calculation that uses equation (\ref{kp7c}) shows that
the adjoint BA function satisfies linear equations which are
conjugated to those for $\psi ({\bf t}, z)$ 
(see (\ref{kp7}), (\ref{kp7a})).
\begin{corollary}
The function $\psi^* ({\bf t}, z)$ is the adjoint 
BA function, i.e., it 
satisfies the equations
\beq\label{kp8}
L^{\dag}\psi^* ({\bf t}, z)=z\psi^* ({\bf t}, z),
\eeq
\beq\label{kp8a}
-\p_{t_k}\psi^* ({\bf t}, z)=B^{\dag}_k \psi^* ({\bf t}, z) \;\;
\mbox{for any $k\geq 1$}.
\eeq
\end{corollary}
\square

\begin{remark}
\label{remark:kpWtau}
Taking into account equations (\ref{kp4}), we can represent
the dressing operator $W$ as
\beq\label{dres}
W=\normord \frac{\tau ({\bf t}-[\p_x^{-1}])}{\tau ({\bf t})}
\normord , \qquad
(W^{\dag})^{-1}=\normord \frac{\tau ({\bf t}+[\p_x^{-1}])}{\tau ({\bf t})}
\normord ,
\eeq
where the normal ordering means that operator $\p_x$ 
in each term of the expansion in negative powers of $\p_x^{-1}$
is moved
to the right. Note that the second equality is by no means 
an algebraic consequence of the first one. 
It holds only on condition that 
$\tau ({\bf t})$ is a KP tau-function, i.e., satisfies the
bilinear equation (\ref{kp1}).
\end{remark}

The following well known proposition directly links the tau-function
and the Lax operator.

\begin{proposition}
Let $L$ be the Lax operator of the KP hierarchy.
The coefficient at $\p_x^{-1}$ of the pseudo-differential
operator $L^n$ (sometimes called the operator
residue of the pseudo-differential operator $L^n$) is expressed
through the tau-function as follows:
\beq\label{Ltau}
(L^n)_{-1} =\p_{t_n}\p_x \log \tau , \quad n\geq 1.
\eeq
\end{proposition}

\noindent
This relation can be proved by applying the differential operator
$\p_{t_n}$ to the bilinear relation (\ref{kp1}) and putting
${\bf t'}={\bf t}$ afterwords. Details of the proof can be
found e.g. in \cite{Z18}.

As is well known, compatibility of equations (\ref{kp7}),
(\ref{kp7a}) (or (\ref{kp8}), (\ref{kp8a})) implies the 
Lax equations
\beq\label{kp9}
\p_{t_k}L=[B_k, L], \quad k\geq 1,
\eeq
and the Zakharov-Shabat equations
\beq\label{kp10}
\p_{t_n}B_m -\p_{t_m}B_n +[B_m, B_n]=0, \quad n,m \geq 1.
\eeq
The latter are in fact compatibility conditions for the 
Lax equations (\ref{kp9}). They
determine dynamics in the space of dependent variables.
Usually, these equations serve as a definition of the KP hierarchy.

An important role in the theory is played by {\it wave functions}.

\begin{definition}\label{definition:wavekp}
Any solution $\phi ({\bf t})$ to the system of linear 
differential equations (\ref{kp7a}), i.e.,
\beq\label{wave1}
\p_{t_k}\phi ({\bf t}, z)=B_k \phi ({\bf t}, z) \;\;
\mbox{for any $k\geq 1$},
\eeq 
is called a wave function (for the KP hierarchy).
Any solution $\phi^* ({\bf t})$ to the conjugated system of 
equations (\ref{kp8a}), i.e.,
\beq\label{wave2}
-\p_{t_k}\phi^* ({\bf t}, z)=B^{\dag}_k \phi^* ({\bf t}, z) \;\;
\mbox{for any $k\geq 1$},
\eeq
is called an adjoint wave function.
\end{definition}

\noindent
For brevity, we will call both 
$\phi$ and $\phi^*$ wave functions, if this does not lead to
a misunderstanding. 

Note that the BA functions (adjoint BA functions) 
are particular and very special 
solutions to the linear system (\ref{wave1})
(respectively, (\ref{wave2})). Namely, 
they are simultaneously eigenfunctions
of the Lax operator. However, the whole set of wave functions is
much broader. Indeed, since the linear equations (\ref{kp7a})
and (\ref{kp8a})
do not contain the spectral parameter $z$ explicitly, any 
linear combination of BA functions $\psi ({\bf t}, z)$ for
different values of $z$ is a wave function.
So, wave functions
$\phi ({\bf t})$, $\phi^*({\bf t})$ 
can be constructed by integrating 
the BA functions with respect to the
spectral parameter $z$ with arbitrary functions (in general,
distributions)
$\rho (z)$, $\rho^*(z)$ of $z, \bar z$ (we call them density functions):
\beq\label{mkp2a}
\begin{array}{l}
\displaystyle{
\phi ({\bf t})=\int_{\CCC}\psi ({\bf t}; z)\rho (z)d^2 z,}
\\ \\
\displaystyle{
\phi^* ({\bf t})=\int_{\CCC}\psi^* ({\bf t}; z)\rho^* (z)d^2 z,}
\end{array}
\eeq
where $d^2z \equiv dx dy$ is the standard measure 
in the complex plane\footnote{We assume that the density 
functions $\rho (z)$,
$\rho^*(z)$ are such that the integrals converge; for example,
we may require that they are bounded and have a compact support.}.
Below we will deal with equations (\ref{mkp2a}) rewritten in terms
of the tau-function:
\beq\label{mkp2b}
\begin{array}{l}
\displaystyle{
\phi ({\bf t})=\frac{1}{\tau ({\bf t})}
\int_{\CCC}d^2 z \, \rho (z) \, e^{\xi ({\bf t}, z)}
 \tau \Bigl ({\bf t}-[z^{-1}]\Bigr ),}
\\ \\
\displaystyle{
\phi^* ({\bf t})=\frac{1}{\tau ({\bf t})}
\int_{\CCC}d^2 z \, \rho^* (z) \, e^{-\xi ({\bf t}, z)}
\tau \Bigl ({\bf t}+[z^{-1}]\Bigr ).}
\end{array}
\eeq
The following proposition proved in \cite{ANP98} states that
formulas (\ref{mkp2a}) provide in fact general solution to the linear
equations (\ref{wave1}), (\ref{wave2}).

\begin{proposition}\cite{ANP98}
Any KP wave function, i.e., a common solution to the linear equations
(\ref{wave1}) (or (\ref{wave2}) for adjoint wave functions) 
have integral representation of the form (\ref{mkp2a}) with 
some density functions (more generally, distributions) 
$\rho$, $\rho^*$.
\end{proposition}

\noindent
Moreover, in \cite{ANP98} some explicit formulas which
allow one to restore density functions from known wave
functions were given. We will not present them here since 
we do not need them in what follows.

There are the following main possibilities
for the density function:
\begin{itemize}
\item[--]
The support of $\rho(z)$ is a compact domain ${\sf D}\in \CC$:
$
\rho (z)=\mu (z)\Theta_{\sf D}(z),
$
where $\mu (z)$ is a bounded function of $z, \bar z$
and $\Theta_{\sf D}(z)$ is the characteristic function
of the domain ${\sf D}$: 
$\Theta_{\sf D}(z)=1$ if $z\in {\sf D}$ and $0$ otherwise.
In this case we encounter the non-local 
$\bar \p$-problem, as explained
in Remark \ref{remark:kp2}. We do not consider this case in this
paper and hope to address it elsewhere.
\item[--]
The support of $\rho(z)$ is a contour $\Gamma \subset \CC$, 
and $\rho (z)$ is a distribution represented 
in terms of the delta-function $\delta_{\Gamma}(z)$
with the support on $\Gamma$
as follows\footnote{The delta-function $\delta_{\Gamma}(z)$
is defined by the rule  
$\displaystyle{
\int_{\CCC}f(z)\delta_{\Gamma}(z)d^2 z=\int_{\Gamma}f(z)|dz|}
$
for any integrable function $f$.}:
$
\rho (z)=\nu (z)\delta_{\Gamma}(z),
$
where $\nu (z)$ is some integrable function.
\item[--]
The support of $\rho(z)$ is a finite number of points $p_i\in \CC$,
and $\rho (z)$ is a distribution represented as a linear combination
of 2D delta-functions:
$\displaystyle{
\rho (z)=\sum_i c_i \delta^{(2)}(z-p_i)}.
$
\end{itemize}

\noindent
Various combinations of these types of density functions are 
also possible.
The same is assumed for $\rho^* (z)$.

\section{Modified KP hierarchy}

\subsection{Modified KP from Toda lattice}
\label{section:mKP-Toda}

To introduce the modified KP (mKP) hierarchy, 
it is instructive to begin with
the more general 2D Toda lattice (2DTL) hierarchy. The independent
variables of the latter are two infinite sets of times
$$
{\bf t}=\{t_1, t_2, t_3, \ldots \}, \quad
\bar {\bf t}=\{\bar t_1, \bar t_2, \bar t_3, \ldots \}
$$
(the bar here does not mean complex conjugation!) and a discrete
variable $n\in \ZZ$. The universal dependent variable is the
2D Toda tau-function $\tau (n, {\bf t}, \bar {\bf t})$ which
satisfies the bilinear equation

\beq\label{mkp1}
\begin{array}{l}
\displaystyle{
\oint_{C_{\infty}} dz \, z^{n-n'-1}e^{\xi ({\bf t}-{\bf t}', z)}
\tau (n, {\bf t}-[z^{-1}], \, \bar {\bf t})
\tau (n'+1, {\bf t}'+[z^{-1}], \, \bar {\bf t'}) 
}
\\ \\
\hspace{-4mm}=\, \displaystyle{
\oint_{C_{\infty}} dz \, z^{n'-n-1}e^{\xi (\bar {\bf t}-\bar {\bf t}', z)}
\tau (n+1, {\bf t}, \, \bar {\bf t}-[z^{-1}])
\tau (n', {\bf t}', \, \bar {\bf t'}+[z^{-1}]). 
}
\end{array}
\eeq
It is valid for all ${\bf t},\, {\bf t'}$, 
$\bar {\bf t},\, \bar {\bf t'}$
and $n, \, n' \in \ZZ$.

For the choice $\bar {\bf t'}=\bar {\bf t}$, $n'=n-1$ the right-hand
side vanishes and equation (\ref{mkp1}) for 
$\tau (n, {\bf t}, \bar {\bf t})$, regarded as a function
of ${\bf t}$ with fixed $\bar {\bf t}$ and $n$, 
reduces to the bilinear
equation (\ref{kp1}) for the KP hierarchy.

For the choice $\bar {\bf t'}=\bar {\bf t}$, $n-n'=k+1$ with
$k\geq 1$ the right-hand
side vanishes again and the equation becomes
\beq\label{mkp2}
\begin{array}{l}
\displaystyle{
\oint_{C_{\infty}} dz \, z^{k}e^{\xi ({\bf t}-{\bf t}', z)}
\tau (n+k, {\bf t}-[z^{-1}])
\tau (n, {\bf t}'+[z^{-1}]) =0, \quad k\geq 1,
}
\end{array}
\eeq
where we omitted dependence on the bar-times which enter 
implicitly as parameters of solutions. The hierarchy defined by
equation (\ref{mkp2}) is sometimes referred to as the $k$-modified KP
hierarchy ($k$-mKP). In what follows, we will deal with
the most important case $k=1$. This is just 
the mKP hierarchy in the Kupershmidt and Kiso sense. 
At $k=0$ the KP hierarchy
is reproduced.

Note that for the choice  
$\bar {\bf t'}=\bar {\bf t}$, $n-n'=0$ in (\ref{mkp1})
(corresponding to $k=-1$)
the right-hand side of (\ref{mkp1}) does not vanish but
can be easily found explicitly as given by the residue at $\infty$:
\beq\label{mkp3}
\begin{array}{l}
\displaystyle{
\frac{1}{2\pi i}
\oint_{C_{\infty}} \frac{dz}{z} \,e^{\xi ({\bf t}-{\bf t}', z)}
\tau (n, {\bf t}-[z^{-1}])
\tau (n+1, {\bf t}'+[z^{-1}]) =
\tau (n+1, {\bf t})\, \tau (n, {\bf t'}).
}
\end{array}
\eeq
This equation (or rather functional relation) as well as
(\ref{mkp2}) for $k=1$, connects two neighboring tau-functions
on the $n$-lattice and is in fact 
equivalent to it defining the same mKP hierarchy\footnote{This
can be proved using the method developed in \cite{Shigyo13}.}. 

Since in what follows we will not be interested in dependence
on $n$, it is convenient to fix them and simplify the notation 
as follows:
\beq\label{mkp4}
\tau (n, {\bf t})=\tau ({\bf t}), \qquad
\tau (n+1, {\bf t})=\hat \tau ({\bf t}).
\eeq
From what was said above it is clear that
both functions $\tau ({\bf t})$ and $\hat \tau ({\bf t})$ are
KP tau-functions and, as such, satisfy equation (\ref{kp1}).
The mKP hierarchy contains the extra equation (\ref{mkp3}) that
connects them:
\beq\label{mkp5}
\begin{array}{l}
\displaystyle{
\frac{1}{2\pi i}
\oint_{C_{\infty}} \frac{dz}{z} \,e^{\xi ({\bf t}-{\bf t}', z)}
\tau ({\bf t}-[z^{-1}])
\hat \tau ({\bf t}'+[z^{-1}]) =
\hat \tau ({\bf t})\, \tau ({\bf t'}).
}
\end{array}
\eeq

Let us mention some important corollaries of the general
equation (\ref{mkp5}). Setting 
$$
{\bf t}-{\bf t'}=[a_1^{-1}]+[a_2^{-1}],
$$
we have
$$
e^{\xi ({\bf t}-{\bf t}', z)}=\frac{a_1 \, a_2}{(a_1-z)(a_2-z)},
$$
and the integral in (\ref{mkp5}) can be calculated by an elementary
residue calculus. The result is:
\beq\label{mkp5b}
\begin{array}{l}
a_1 \tau ({\bf t}-[a_2^{-1}])\hat \tau ({\bf t}-[a_1^{-1}])-
a_2 \tau ({\bf t}-[a_1^{-1}])\hat \tau ({\bf t}-[a_2^{-1}])
\\ \\
\phantom{aaaaaaaaaaaaaaaaaaaaaaaaa}=
(a_1-a_2) \tau ({\bf t}-[a_1^{-1}]-[a_2^{-1}])\hat \tau ({\bf t}).
\end{array}
\eeq
In the short-hand notation introduced in (\ref{not1}) this
equation can be written in a much more compact form:
\beq\label{mkp5c}
a_1\tau^2 \hat \tau^1 -a_2\tau^1 \hat \tau^2=
a_{12} \tau^{12}\hat \tau , 
\eeq
or, after an overall shift of ${\bf t}$,
\beq\label{mkp5d}
a_1\tau_1 \hat \tau_2 -a_2\tau_2 \hat \tau_1=
a_{12} \tau \hat \tau_{12}.  
\eeq
As is proved in
\cite{Shigyo13}, equation (\ref{mkp5c}) (or (\ref{mkp5d}))
is in fact equivalent to the whole mKP hierarchy defined by the
integral equation (\ref{mkp5}).
There is a slightly more generally looking version of equation
(\ref{mkp5c}) including three rather than two points, which
is sometimes useful. It 
is obtained from (\ref{mkp5}) in the same manner after
the substitution ${\bf t}-{\bf t'}=[a_1^{-1}]+[a_2^{-1}]+[a_3^{-1}]$:
\beq\label{mkp5e}
(a_1^{-1}-a_2^{-1})\tau^{12}\hat \tau^3 +
(a_2^{-1}-a_3^{-1})\tau^{23}\hat \tau^1 +
(a_3^{-1}-a_1^{-1})\tau^{13}\hat \tau^2=0.
\eeq
It is reduced to (\ref{mkp5c}) in the limit $a_3\to \infty$.
(However, as it was already said, equations (\ref{mkp5c}) and
(\ref{mkp5e}) are in fact equivalent to each other and to the entire
hierarchy.)

\begin{proposition}
Let $\tau$, $\hat \tau$ be two functions satisfying 
equation (\ref{mkp5c}) (or (\ref{mkp5d})). Then both of them
are KP tau-functions, i.e., they satisfy equation (\ref{kp1c}).
\end{proposition}

\noindent
{\it Proof.}
Let us write equation (\ref{mkp5c}) in the form
$$
a_{12}\frac{\tau^{12} \hat \tau}{\tau^1 \tau^2}=
a_1 \frac{\hat \tau^1}{\tau^1} -a_2 \frac{\hat \tau^2}{\tau^2}
$$
and sum the three such equations written for the three cyclic
permutations of $(123)$. Then the right-hand side vanishes,
and we obtain (\ref{kp1c}). The argument for $\hat \tau$ is
similar.
\square

\begin{remark}
Comparing equations (\ref{kp1c}) and (\ref{mkp5c}), one can notice
that the latter is formally obtained from the former in the ``limit'' 
$a_3\to 0$. More precisely, the following formal identification
works:
\beq\label{formal}
\tau  (n, {\bf t}-[z^{-1}] ) \Bigr |_{z=0}\rightarrow
\tau  (n+1, {\bf t} ), \quad 
\mbox{or} \quad
\tau  ({\bf t}-[z^{-1}]) \Bigr |_{z=0}\rightarrow
\hat \tau  ({\bf t} ).
\eeq
It is assumed here that the function $\tau  ({\bf t}-[z^{-1}])$, as a
function of $z$, admits an analytic continuation to the point
$z=0$ from the neighborhood of $\infty$, where it was originally
defined. The rule (\ref{formal}) can be most easily understood
in the language of free fermions (see, e.g., \cite{AZ12})).
\end{remark}

\subsection{From KP to modified KP and back}

The transition $\tau ({\bf t})\rightarrow \hat \tau ({\bf t})$ is
referred to as a B\"acklund-Darboux (BD) transformation of solutions to
the KP hierarchy. So, the mKP hierarchy (\ref{mkp3}) describes 
chains of successive BD transformations. In what follows, the 
objects related to the solution with the tau-function
$\hat \tau$ will be denoted by the same 
letters with hat ($\hat W$, $\hat L$,
$\hat \psi ({\bf t}, z)$, $\hat \phi ({\bf t})$, etc), the latter 
is a wave function for the solution with the tau-function
$\hat \tau$.

Theorem \ref{theorem:mkp1} below provides description of a general
solution to the mKP hierarchy in the form (\ref{mkp5}).
In plain words, all solutions to the mKP hierarchy can be
constructed from a solution to the KP hierarchy and a wave
function corresponding to it.

\begin{theorem} \label{theorem:mkp1}
a)
Let $\tau ({\bf t})$ be a tau-function of the KP hierarchy
(i.e., a solution to (\ref{kp1})), and $\phi ({\bf t})$ be a wave
function corresponding to this solution. Then the function
\beq\label{mkp6}
\hat \tau ({\bf t})=\phi ({\bf t}) \tau ({\bf t})
\eeq
is another KP tau-function (i.e., it solves (\ref{kp1}), too) and
the pair $(\tau , \hat \tau )$ satisfies equation (\ref{mkp5}).

b)
Conversely, let a pair $(\tau , \hat \tau )$ of KP tau-functions
be a solution to the mKP bilinear equation (\ref{mkp5}), then
their ratio
\beq\label{mkp6a}
\phi ({\bf t})=\frac{\hat \tau ({\bf t})}{\tau ({\bf t})}
\eeq
is a KP wave function corresponding to the solution 
$\tau ({\bf t})$.
\end{theorem}

\noindent
{\it Proof.} The part a) is well known and is traced back to 
early works on BD transformations of integrable hierarchies
(for a detailed proof see, e.g., the recent review \cite{Z25a}
and references therein). 

Here we give a sketch of proof of part b). We need to prove that
the function $\phi ({\bf t})$ defined by (\ref{mkp6a}) satisfies
the linear equations (\ref{wave1}) with the differential
operators $B_k=(L^k)_{\geq 0}$, where the Lax operator 
$L$ corresponds to the solution with the tau-function
$\tau ({\bf t})$. To this end, we employ 
the method
first suggested in \cite{TT95} and then used in
\cite{Teo11} for the multi-component
KP hierarchy.

Introduce the differential operator
\beq\label{D(z)}
D(z)=\sum_{k\geq 1}\frac{z^{-k}}{k}\, \p_{t_k},
\eeq
then for any function $f({\bf t})$ we
can write $f({\bf t}\pm [z^{-1}])=e^{\pm D(z)}f({\bf t})$.
We have to show that
$$
D(a)\phi ({\bf t})=\sum_{k\geq 1}\frac{a^{-k}}{k}\, B_k \phi ({\bf t}),
$$
with $\phi ({\bf t})$ given by 
(\ref{mkp6a}) and $B_k$ being the differential 
operator in $x$ of $k$-th order
given by (\ref{Bk}). 

To this end, we note that the function $\phi ({\bf t})$ 
defined as the ration of KP tau-functions (\ref{mkp6a})
satisfies the linear equation for the $\phi ({\bf t})$:
\beq\label{mkp6k}
a_1 \phi_2-a_2\phi_1 =a_{12} \frac{\tau \tau_{12}}{\tau_1 \tau_2}\,
\phi_{12}.
\eeq
This immediately follows from equation (\ref{mkp5d}) after the
substitution $\hat \tau =\tau \phi$.
Letting $a_2\to \infty$ (and renaming $a_1=a$), 
we convert it into the following linear differential-difference
equation:
\beq\label{mkp8a}
\Bigl (1-e^{-D(a)}\Bigr )\phi ({\bf t})=
a^{-1}\left (\p_{x}-\p_{x}\log 
\frac{\tau ({\bf t}-[a^{-1}])}{\tau ({\bf t})}\right )\phi ({\bf t}),
\eeq
or
\beq\label{St61}
a\Bigl (1-e^{-D(a)}\Bigr )\phi ({\bf t}) -
\Bigl (\p_x - V({\bf t}, a)\Bigr )
\phi ({\bf t}) =0,
\eeq
where
\beq\label{St71}
V({\bf t}, a)=
\p_{x}\log \frac{\tau ({\bf t}-[a^{-1}])}{\tau ({\bf t})}=
\sum_{k\geq 1}V_k ({\bf t}) a^{-k}.
\eeq
The functions $V_k ({\bf t})$ are expressed through logarithmic
derivatives of the tau-function $\tau$.

The next step is to find
action of $D(a)$ on the BA function. This can be done in the
following way. The starting point
is the well known Hirota-Miwa three-term bilinear relation for the
KP tau-function (\ref{kp1d}) (which is in fact equivalent to the whole
hierarchy, see \cite{Shigyo13} for a direct proof).
Putting $a_3=z$ and
using formula (\ref{kp4}) for the BA function, this relation can be represented in the form a linear  
equation for the BA function:
\beq\label{mkp7}
a_1\psi ({\bf t}+[a_2^{-1}], z)-
a_2\psi ({\bf t}+[a_1^{-1}], z)=(a_1-a_2)
\frac{\tau ({\bf t})
\tau ({\bf t}+[a_1^{-1}]+[a_2^{-1}])}{\tau ({\bf t}+[a_1^{-1}])
\tau ({\bf t}+[a_2^{-1}])}\, \psi ({\bf t}+[a_1^{-1}]+[a_2^{-1}])
\eeq
(see \cite{Z25a} for details).
Letting $a_2\to \infty$ (and renaming $a_1=a$), 
we convert it into the following linear differential-difference
equation for the BA function:
\beq\label{St61a}
a\Bigl (1-e^{-D(a)}\Bigr )\psi ({\bf t},z) -
\Bigl (\p_x - V({\bf t}, a)\Bigr )
\psi ({\bf t},z) =0,
\eeq
which is the same as equation (\ref{St61}) for the $\phi ({\bf t})$.

The operator $e^{-D(a)}-1$ can be expanded as
\beq\label{exp}
e^{-D(a)}-1=\sum_{k\geq 1}p_k(-[\p_x^{-1}] )a^{-k},
\eeq
where 
$p_k({\bf t})$ are elementary Schur polynomials defined by
$$
\exp \Bigl (\sum_{k\geq 1}t_k z^k \Bigr ) =1+ \sum_{k\geq 1}
p_k({\bf t})z^k.
$$
Their general structure is
$
p_k({\bf t})= t_k + P_k(t_{k-1}, \ldots , t_1),
$
where $P_k$ is a polynomial of the variables $t_l$ with $l<k$.
For example,
$$
\begin{array}{l}
p_1({\bf t})=t_1, \quad p_2({\bf t})=t_2 +\frac{1}{2}t_1^2,
\quad 
p_3({\bf t})=t_3 +t_1t_2 +\frac{1}{6}t_1^3,
\end{array}
$$
so
$$
\begin{array}{l}
p_1(-[\p_x^{-1}])=-\p_{t_1}, 
\\ \\
p_2(-[\p_x^{-1}])=
\frac{1}{2}(-\p_{t_2} + \p_{t_1}^2 ),
\\ \\
p_3(-[\p_x^{-1}])=\frac{1}{3}(-\p_{t_3} + 
\frac{3}{2}\p_{t_2}\p_{t_1} -\frac{1}{2}\p_{t_1}^3 ).
\end{array}
$$
Then from equation (\ref{St61}) it follows that the function
$\phi$ satisfies linear differential equations of the form
\beq\label{St71a}
p_{k+1}(-\tilde \p ) \phi - V_k \phi =0, \qquad k\geq 1.
\eeq
To identify them with (\ref{wave1}), we note that, as it follows from
(\ref{St61a}),
the KP BA function $\psi ({\bf t}, z)$ given by (\ref{kp4})
satisfies the same equations
\beq\label{St7ab}
p_{k+1}(-\tilde \p ) \psi - V_k \psi =0, \qquad k\geq 1.
\eeq
The structure of these equations allows one to
express the derivative $\p_{t_k}\psi$ through derivatives
$\p_{t_l}\psi$ with $l<k$:
\beq\label{rhs}
\p_{t_k}\psi ={\sf B}_k(\p_{t_{k-1}}, \ldots , \p_{t_1})\psi ,
\eeq
where ${\sf B}_k$ in the right-hand side is some differential
operator in the variables $t_1, \ldots , t_{k-1}$. Applying it
successively for $k=2,3, \ldots $, it is possible to represent
the right-hand side of (\ref{rhs}) 
as a differential operator containing
$\p_{t_1}=\p_x$ only:
$$
\p_{t_k}\psi =B_k(\p_{x})\psi .
$$
However, we already know from Corollary \ref{corollary:kp1} that
for the BA function the operators 
$B_k(\p_{x})$ coincide with the 
operators ${\cal B}_k$ defined in (\ref{Bk}) (see (\ref{kp7a})).
Since the function $\phi$ satisfies the same linear equations as
the BA function $\psi$, we conclude that from (\ref{St61})
it follows that
$\p_{t_k}\phi=B_k \phi $, i.e., $\phi$ defined 
in (\ref{mkp6a}) is indeed a KP wave function.
\square

%%%%%%%%%%%%%%%%%%%%%%%%%%B\"acklund-Darboux in terms of Lax
%%%%%%%%%%%%%%%%%%%%%%%%%%

Equation (\ref{mkp6}) defines what is called 
{\it forward} BD transformation (of the KP hierarchy). 
A {\it backward} one is
defined in a similar way with the help of any adjoint wave
function $\phi^*({\bf t})$:
\beq\label{mkp6d}
\hat \tau^* ({\bf t})=\phi^* ({\bf t}) \tau ({\bf t}),
\eeq
or, using (\ref{mkp2b}),
\beq\label{mkp6e}
\hat \tau^* ({\bf t})=\int_{\CCC}\rho^*(z)e^{-\xi ({\bf t}, z)}
\tau ({\bf t}+[z^{-1}]) \, dz.
\eeq
More details can be found, e.g., in \cite{Z25a}.
Here we just mention that the Lax operator $L$ for the forward
BD transformation $\tau \to \hat \tau =\phi \tau$ is transformed
as follows:
\beq\label{LBD}
L \longrightarrow \, \phi \p_x \phi^{-1} L  \phi \p_x^{-1}\phi^{-1}
\eeq
(see \cite{OR} for details).

We conclude this section by presenting nonlinear equations 
satisfied by the functions $\phi$ and $\phi^*$.

\begin{theorem}\label{theorem:fromKPtomKP}
The KP wave functions $\phi$ and $\phi^*$ satisfy the 
equations
\beq\label{mkp51}
\frac{a_1 \phi_2 -a_2 \phi_1}{\phi_{12}}+
\frac{a_2 \phi_3 -a_3 \phi_2}{\phi_{23}}+
\frac{a_3 \phi_1 -a_1 \phi_3}{\phi_{13}}=0,
\eeq
\beq\label{mkp51a}
\frac{a_2 \phi^{*}_{12} -a_3 \phi^{*}_{13}}{\phi^{*}_{1}}+
\frac{a_3 \phi^{*}_{23} -a_1 \phi^{*}_{12}}{\phi^{*}_{2}}+
\frac{a_1 \phi^{*}_{13} -a_2 \phi^{*}_{23}}{\phi^{*}_{3}}=0,
\eeq
which are the discrete mKP equations. 
\end{theorem}

\noindent
{\it Proof.} To prove (\ref{mkp51}), we rewrite equation (\ref{mkp6k})
in the form
$$
\frac{a_1 \phi_2 -a_2 \phi_1}{\phi_{12}}=
a_{12}\, \frac{\tau_3 \tau_{12} \tau}{\tau_1 \tau_2 \tau_3}.
$$
Summing three such equations corresponding to cyclic 
permutations of $(123)$ and using (\ref{kp1d}), we obtain
(\ref{mkp51}). Equation (\ref{mkp51a}) can be proved in a
similar way.
\square

\begin{remark}
Note that $\phi^{-1}$ satisfies the same 
equation (\ref{mkp51a}) as $\phi^*$. Therefore,
if $\phi ({\bf t})$ is any wave function corresponding 
to the Lax operator $L$ of the KP hierarchy, then
$\phi^{-1} ({\bf t})$
is an adjoint wave function for the Lax operator $\hat L$
in the sense of Definition \ref{definition:wavekp}.
\end{remark}

\noindent
Let us present some other forms of equations
(\ref{mkp51}), (\ref{mkp51a}). According to the remark above,
it is enough to write the equations for $\phi$.
First of all, after overall shifts 
of the variables, equation (\ref{mkp51}) can be written as
\beq\label{mkp511a}
\frac{a_1 \phi^{31} -a_2 \phi^{32}}{\phi^{3}}+
\frac{a_2 \phi^{12} -a_3 \phi^{13}}{\phi^{1}}+
\frac{a_3 \phi^{23} -a_1 \phi^{12}}{\phi^{2}}=0.
\eeq
A direct calculation shows that this is the same as
\beq\label{mkp51b}
\left | 
\begin{array}{ccc}
1 & a_1 \phi^{1} & \phi^{23} \phi^{1}
\\ && \\
1 & a_2 \phi^{2} & \phi^{13} \phi^{2}
\\ && \\
1 & a_3 \phi^{3} & \phi^{12} \phi^{3}
\end{array}
\right | =0.
\eeq 
To write down another useful form of this equation, we denote
\beq\label{R}
R(\alpha , \beta)=\frac{a_{\alpha}\phi^{\alpha} -
a_{\beta}\phi^{\beta}}{\phi}, \quad \alpha , \beta =1,2,3,
\eeq
then (\ref{mkp511a}) can be written as
$
R^3(1,2) + R^1(2,3) + R^2(3,1) =0. 
$
Moreover, a direct calculation shows that this is equivalent to
\beq\label{R1}
\frac{R^{\gamma}(\alpha , \beta )}{R(\alpha , \beta )}=
\frac{R^{\alpha}(\beta , \gamma )}{R(\beta , \gamma )}, 
\qquad 
(\alpha \beta \gamma)=(123), (231), (312).
\eeq
In the explicit form, this reads:
\beq\label{R3}
\frac{\phi^1 (a_1 \phi^{13} -a_2 \phi^{23})
(a_2 \phi^2 -a_3\phi^3)}{\phi^3 (a_2 \phi^{12} -a_3 \phi^{13})
(a_1 \phi^1 -a_2\phi^2)}=1,
\eeq
which is equivalent to the similar equations obtained from it
by cyclic permutations of the indices $1,2,3$.

So far we were following the way from KP to mKP. The next theorem
(inverse to Theorem \ref{theorem:fromKPtomKP})
shows that there is also a way back: starting from a
solution to equation (\ref{mkp51}) (or (\ref{mkp511a})) it is
possible to restore the corresponding solutions 
$\tau$, $\hat \tau$ to the KP hierarchy. 

\begin{theorem}\label{theorem:inversetoKPtomKP}
Let $\phi ({\bf t})$ be any solution of the mKP equation
(\ref{mkp511a}) (or (\ref{R3})). 
Then there exist functions $\tau ({\bf t})$ and
$\hat \tau ({\bf t})$ such that both of them are KP tau-functions
connected by equation (\ref{mkp5}).
\end{theorem}

\noindent
In the proof, the following technical lemma is used.

\begin{lemma}\label{lemma:technical2}
Let $h({\bf t}, z)$ be a function of ${\bf t}$, $z$ 
satisfying the functional
relation
$$
\frac{h({\bf t}-[a^{-1}], b)}{h({\bf t}, b)}=
\frac{h({\bf t}-[b^{-1}], a)}{h({\bf t}, a)}
$$
for all $a,b \in \CC$. Then there exists a function $H({\bf t})$
such that
$$
h({\bf t}, a)=\frac{H({\bf t}-[a^{-1}])}{H({\bf t})}.
$$
\end{lemma}

\noindent
A detailed proof of this lemma can be found in \cite{DJKM83}.

\noindent
{\it Proof of Theorem \ref{theorem:inversetoKPtomKP}.}
Put
\beq\label{R4}
R({\bf t}, a, b)=\frac{a \phi ({\bf t}
-[a^{-1}])-b \phi ({\bf t}-[b^{-1}])}{\phi ({\bf t})},
\eeq
then equation (\ref{R3}) states that
$$
\frac{R({\bf t}-[a_1^{-1}], a_2, a_3)}{R({\bf t}, a_2, a_3)}=
\frac{R({\bf t}-[a_2^{-1}], a_1, a_3)}{R({\bf t}, a_1, a_3)}.
$$
According to Lemma \ref{lemma:technical2}, there exists a function
$G({\bf t}, z)$ such that
\beq\label{R51}
R({\bf t}, a, b)=\frac{a\, G({\bf t}-[a^{-1}], b)}{G({\bf t}, b)}
\eeq
(the multiplier $a$ in the numerator is restored by comparing
(\ref{R4}) and (\ref{R51}) as $a\to \infty$). Moreover, the
definition of $R({\bf t}, a, b)$ implies that
$R({\bf t}, a, b)=-R({\bf t}, b, a)$, i.e., it holds
$$
\frac{a\, G({\bf t}-[a^{-1}], b)}{G({\bf t}, b)}=-
\frac{b\, G({\bf t}-[b^{-1}], a)}{G({\bf t}, a)}.
$$
In particular, we see from this that
\beq\label{R7}
\lim_{b\to \infty} \frac{a\, 
G({\bf t}-[a^{-1}], b)}{ b \, G({\bf t}, b)}=-1.
\eeq
Modifying the function $G$ as $\tilde G({\bf t}, z)=
e^{-\xi ({\bf t}, z)}G({\bf t}, z)$, we have for the $\tilde G$:
$$
\frac{\tilde G({\bf t}-[a^{-1}], b)}{\tilde G({\bf t}, b)}=
\frac{\tilde G({\bf t}-[b^{-1}], a)}{\tilde G({\bf t}, a)},
$$
and
$
\lim\limits_{a\to \infty} \tilde G({\bf t}, a)=1
$
(this follows from (\ref{R7})).
Applying Lemma \ref{lemma:technical2} again, we conclude that 
there exists a function $F({\bf t})$ such that
$$
\tilde G({\bf t}, z)=\frac{F({\bf t} -[z^{-1}])}{F({\bf t})}.
$$
Therefore, for $R({\bf t}, a,b)$ we have:
\beq\label{R8}
R({\bf t}, a,b)=(a-b)\frac{F({\bf t}-
[a^{-1}]-[b^{-1}])F({\bf t})}{F({\bf t}-[a^{-1}])
F({\bf t}-[b^{-1}])},
\eeq
and equation (\ref{mkp511a}) reads:
$$
R^1(a_2,a_3)+R^2(a_3,a_1)+R^3(a_1,a_2)=
a_{23}\frac{F^{123} F^1}{F^{12} F^{13}} +
a_{31}\frac{F^{123} F^2}{F^{23} F^{21}}+
a_{12}\frac{F^{123} F^3}{F^{13} F^{22}}=0.
$$
The latter equality can be rewritten as
$
a_{12}F^{12}F^3 + a_{23}F^{23}F^1 + a_{31}F^{31}F^2 =0.
$
Thus we can identify $F$ with a KP tau-function $\tau$ 
satisfying the bilinear relation (\ref{kp1c}) (and, therefore,
equation (\ref{kp1})).
Introducing the function $\hat F =\phi F$, we rewrite equation
(\ref{R4}) as
$$
a_1 \hat F^1 F^2 -a_2 \hat F^2 F^1 =a_{12}\hat F F^{12},
$$
which shows that $\hat F$ can be identified with the KP tau-function
$\hat \tau$ satisfying equation (\ref{mkp5c}) 
(and, therefore, (\ref{mkp5})).
\square

The following proposition will be used in Section 
\ref{section:mKPSchKP}.

\begin{proposition}
\label{proposition:phiphi}
Let $\phi$ and $\phi'$ be any two different 
solutions to the mKP hierarchy
(i.e., solutions to (\ref{mkp511a})). Then any linear combination
$\phi'' =a\phi +b\phi'$ with $a,b\in \CC$
is a solution to the mKP hierarchy if and
only if $\phi$ and $\phi'$ are wave functions 
corresponding to the same solution of the KP hierarchy. 
\end{proposition}

\noindent
{\it Proof.}
The ``if'' part is obvious: any linear combination of
the wave functions $\phi$, $\phi'$ corresponding to one and the same
solution to the KP hierarchy satisfies the linear
equations (\ref{wave1}) and thus is a KP wave function, too.
To prove the ``only if'' part, we note that 
Theorem \ref{theorem:inversetoKPtomKP} implies that $\phi$, $\phi'$ and
$\phi''$ can be
represented as ratios of some KP tau-functions:
\beq\label{phi0}
\phi =\frac{\hat \tau}{\tau}, \qquad \phi' =\frac{\bar \tau}{\tau'},
\qquad \phi'' =a\phi +b\phi'=\frac{\tilde \tau}{\tau''}.
\eeq
We have to prove that $\tau'' =\tau' =\tau$. To this end, we
will use equation (\ref{mkp5d}) written in the form of linear
equations for $\phi$, $\phi'$ and $\phi''$:
\beq\label{phi1}
\begin{array}{l}
a_1\phi_2 -a_2 \phi_1 =a_{12}u_{12} \phi_{12},
\\ \\
a_1\phi'_2 -a_2 \phi'_1 =a_{12}u'_{12} \phi'_{12},
\\ \\
a_1\phi''_2 -a_2 \phi''_1 =a_{12}u''_{12} \phi''_{12},
\end{array}
\eeq
where $\phi'' =a\phi +b\phi'$ and
$$
u_{12}=\frac{\tau \tau_{12}}{\tau_1 \tau_2}, \qquad
u'_{12}=\frac{\tau' \tau'_{12}}{\tau'_1 \tau'_2},\qquad
u''_{12}=\frac{\tau'' \tau''_{12}}{\tau''_1 \tau''_2}.
$$
Note that the 3-term bilinear equation for the KP tau-function
implies that the equation
\beq\label{kpu}
a_{12}u_{12} +a_{23}u_{23}+a_{31}u_{31}=0
\eeq
holds together with the same equations for $u'$ and $u''$.
Subtracting the last equation in (\ref{phi1}) 
from the sum of the first two, we get:
\beq\label{phi2}
a (u_{12} -u_{12}'')\phi_{12} +b (u'_{12} -u_{12}'')\phi_{12}'=0
\eeq
which should hold for all $a,b$. 
Note, however, that $u_{12}''$ 
in principle may depend on $a,b$, so we can not yet conclude 
from (\ref{phi2}) 
that $u_{12} =u_{12}'=u_{12}''$. To overcome this difficulty,
we proceed as follows. Equation (\ref{phi2}) can be rewritten as
\beq\label{phi3}
u_{12}''=\frac{au_{12}\phi_{12}+
bu'_{12}\phi'_{12}}{a \phi_{12} +b\phi'_{12}}=
u_{12} +b \frac{(u'_{12}-u_{12})\phi'_{12}}{a\phi_{12} +
b\phi'_{12}}=
u'_{12} -a \frac{(u'_{12}-u_{12})\phi_{12}}{a\phi_{12} +
b\phi'_{12}}.
\eeq
This relation allows us to obtain a system of homogeneous
linear equations for 
$$
x_{\alpha}\equiv u'_{\beta \gamma} -u_{\beta \gamma}, \qquad
(\alpha \beta \gamma )=(123), (231), (312).
$$
Multiplying both sides of the equalities in (\ref{phi3}) by
$a_{12}$ and summing over cyclic permutations of $(123)$,
we get (using equations for $u$, $u'$ and $u''$ of the form (\ref{kpu})):
$$
\begin{array}{l}
\displaystyle{
\frac{a_{23}\phi'_{23}x_1}{a\phi_{23} +b\phi'_{23}}+
\frac{a_{31}\phi'_{31}x_2}{a\phi_{31} +b\phi'_{31}}+
\frac{a_{12}\phi'_{12}x_3}{a\phi_{12} +b\phi'_{12}}=0},
\\ \\
\displaystyle{
\frac{a_{23}\phi_{23}x_1}{a\phi_{23} +b\phi'_{23}}+
\frac{a_{31}\phi_{31}x_2}{a\phi_{31} +b\phi'_{31}}+
\frac{a_{12}\phi_{12}x_3}{a\phi_{12} +b\phi'_{12}}=0}.
\end{array}
$$
It is important to note that $a$ and $b$ in the two equations
are not necessarily the same: each equation should hold 
for arbitrary $a,b$. Putting successively $a=0$, $b=0$ in
the first equation and $a=0$ in the second one, we get the
following system of linear equations:
$$
\left \{
\begin{array}{l}
a_{23}x_1 +a_{31}x_2 +a_{12}x_3=0,
\\ \\
a_{23}r_{23}x_1 +a_{31}r_{31}x_2 +a_{12}r_{12}x_3=0,
\\ \\
a_{23}r^{-1}_{23}x_1 +a_{31}r_{31}^{-1}x_2 +a_{12}r^{-1}_{12}x_3=0,
\end{array}
\right.
$$
where $r_{ij}=\phi'_{ij}/\phi_{ij}$. It is a homogeneous
linear system of the form ${\sf M}\vec x =0$, with the determinant
of the matrix ${\sf M}$ equal to
$$
\det {\sf M}=\frac{a_{12}a_{23}a_{31}}{r_{12}r_{23}r_{31}}\,
(r_{12}-r_{13})(r_{12}-r_{23})(r_{13}-r_{23}).
$$
This determinant is nonzero unless $r_{ij}=r_{ik}$ for some 
$i,j,k=1,2,3$ (i.e., $\phi ({\bf t}+[a^{-1}])=\phi'({\bf t})$
which implies that $\phi ({\bf t})=\phi'({\bf t})$),
but this is impossible since we have assumed that the wave
functions $\phi$ and $\phi'$ are different. Therefore, the system
has only the trivial solution $x_1=x_2=x_3=0$ which means 
that $u=u'$, (and, moreover, $u=u'=u''$, as this 
follows from (\ref{phi2})), i.e.,
$\displaystyle{
\frac{\tau \tau_{12}}{\tau_1 \tau_2}=
\frac{\tau' \tau'_{12}}{\tau'_1 \tau'_2}.}
$
Denoting $f=\tau'/\tau$, we write this as $f_{12}/f_2=f_1/f$, 
or $g({\bf t}+[a^{-1}])=g({\bf t})$ for the function
$g=f_1/f$. It then follows that $g({\bf t})=1$, 
i.e., we have
$f({\bf t}+[a^{-1}])=f({\bf t})$. In the same way, this
implies that $f({\bf t})=c$, where $c$ is a constant, i.e.,
$\tau' =c \tau$. The constant can be put equal to 1 without
any loss of generality because the tau-function $\tau'$ is defined
by (\ref{phi0}) only up to a constant factor. The same argument
works for $\tau''$. Therefore, we have proved that
$\tau'=\tau'' =\tau$ in (\ref{phi0}), so the wave functions
$\phi$, $\phi'$ and $\phi''$ correspond to one and the same
solution of KP with the tau-function $\tau$.
\square

\subsection{Baker-Akhiezer functions for the mKP hierarchy}

Just as it was done for the KP hierarchy, 
one may introduce BA functions for the mKP hierarchy,
so that they would satisfy a bilinear relation similar to
(\ref{kp1a}). Following \cite{Cheng2017,Cheng2018,TakTeo06}, we define them
by the formulas
\beq\label{mkp9}
\begin{array}{l}
\displaystyle{
\varphi ({\bf t}, z)=e^{\xi ({\bf t}, z)}\, 
\frac{\tau ({\bf t}-[z^{-1}])}{\hat \tau ({\bf t})}},
\\ \\
\displaystyle{
\varphi^* ({\bf t}, z)=e^{-\xi ({\bf t}, z)}\, 
\frac{\hat \tau ({\bf t}+[z^{-1}])}{z\, \tau ({\bf t})}}.
\end{array}
\eeq
They are connected with the BA functions for the KP hierarchy
in the following way:
\beq\label{con}
\begin{array}{l}
\varphi ({\bf t}, z)=\phi^{-1}({\bf t}) \psi ({\bf t}, z),
\\ \\
\varphi^* ({\bf t}, z)=\phi ({\bf t}) \hat \psi^* ({\bf t}, z).
\end{array}
\eeq
The expansions as $z\to \infty$ are of the form
\beq\label{mkp10}
\begin{array}{l}
\displaystyle{
\varphi ({\bf t}, z)=e^{\xi ({\bf t}, z)}\, 
\Bigl (w_0({\bf t})+w_1({\bf t})z^{-1}+ w_2({\bf t})z^{-2}+\ldots \, \Bigr ),}
\\ \\
\displaystyle{
\varphi^* ({\bf t}, z)=z^{-1}
e^{-\xi ({\bf t}, z)}\, 
\Bigl (w_0^{-1}({\bf t})+w^*_1({\bf t})z^{-1}+ w^*_2({\bf t})
z^{-2}+\ldots \, 
\Bigr ),}
\end{array}
\eeq
where
\beq\label{w0}
w_0({\bf t})=\frac{\tau ({\bf t})}{\hat \tau ({\bf t})}=
(\phi ({\bf t}))^{-1}\,
\eeq
(see (\ref{mkp6a})). In terms of the $\varphi$-functions,
equation (\ref{mkp5}) acquires the form
\beq\label{mkp5a}
\frac{1}{2\pi i}\oint_{C_{\infty}}\!
\varphi  ({\bf t}, z)\varphi^* ({\bf t'}, z)\, dz =1.
\eeq
Its role in the mKP theory is the same as that of equation (\ref{kp1a}) 
for the KP theory (see Remark \ref{remark:psi}). 

\subsection{Lax-Sato formalism for the mKP hierarchy}
\label{section:Lax-SatomKP}

One can introduce an mKP-analogue of the dressing operator.
It is a pseudo-differential
operator ${\cal W}$ of order 0 of the form
\beq\label{calW}
{\cal W}=w_0 +w_1 \p_x^{-1} +w_2 \p_x^{-2}+ \ldots \, ,
\eeq
where the coefficient functions $w_i=w_i({\bf t})$ are the same 
as in (\ref{mkp10}).
Clearly, we can write
\beq\label{mkp12}
\varphi ({\bf t}, z)={\cal W} e^{\xi ({\bf t}, z)}.
\eeq
The Lax operator for the mKP hierarchy is defined as
\beq\label{Laxmkp}
{\cal L}={\cal W} \p_x {\cal W}^{-1} = \p_x  +v_0+
v_1\p_x^{-1}+
v_2\p_x^{-2}+\ldots \, ,
\eeq
where
\beq\label{v0}
v_0=-\p_x \log w_0.
\eeq
It is a pseudo-differential operator of order 1, 
but, unlike the Lax operator for KP, the coefficient
in front of $(\p_x)^0$ is nonzero.
Along with the Lax operator, we introduce the purely differential
operators
\beq\label{Bkmkp}
{\cal B}_k=({\cal L}^k)_{\geq 1}, \qquad k\geq 1.
\eeq
For example, ${\cal B}_1=\p_x$, ${\cal B}_2=\p_x^2 +v_0\p_x$.

The next theorem is an mKP-analogue of Theorem \ref{theorem:kp1}.

\begin{theorem} \label{theorem:mkp1a}
Let $\varphi ({\bf t}, z)$, $\varphi^* ({\bf t}, z)$ 
be functions of the form (\ref{mkp10}). 
If the bilinear relation (\ref{mkp5a}) is satisfied, then there exists a dressing operator ${\cal W}$ of the form (\ref{calW})
such that
$\varphi ({\bf t}, z)={\cal W} e^{\xi ({\bf t}, z)}$ 
satisfying the evolution
equations
\beq\label{mkp11}
\p_{t_k}{\cal W}=-({\cal L}^k)_{\leq 0}{\cal W} \;\;
\mbox{for any $k\geq 1$}.
\eeq
\end{theorem}

\noindent
{\it Proof.} For the proof we need the following mKP-analogue
of Lemma \ref{lemma:technical}:

\begin{lemma} \label{lemma:technical1}
Let ${\cal P}$ and ${\cal Q}$ be pseudo-differential operators
of the form
$$
{\cal P}=w(x)+\sum_{j\geq 1}p_j(x) \p_x^{-j}, \quad
{\cal Q}=w^{-1}(x)(-\p_x^{-1})+\sum_{j\geq 1}q_j(x) (-\p_x)^{-j-1}.
$$
Then the condition
\beq\label{lem2}
\oint_{C_{\infty}}({\cal P}_x e^{xz})\, ({\cal Q}_{x'}e^{-x'z}) \, 
dz =c
\eeq
for all $x,x'$, where $c$ is some constant, 
implies that
\beq\label{lem3}
\Bigl ({\cal P}\, {\cal Q}^{\dag}\Bigr )_{< 0}=c\p_x^{-1}.
\eeq 
\end{lemma}

Set
$$
\varphi ({\bf t}, z)={\cal W} e^{\xi ({\bf t}, z)},\quad
\varphi^* ({\bf t}, z)={\cal V} e^{-\xi ({\bf t}, z)},
$$
where both ${\cal W}$ and ${\cal V}$ are of the form (\ref{calW})
(the latter one is of order $-1$). 
Plugging this
into (\ref{mkp5a}) with $t_k=t'_k$ for $k\geq 2$, we have:
$$
\oint_{C_{\infty}} ({\cal W}_x e^{xz})\,
({\cal V}_{x'} e^{-x'z})\, dz =1.
$$
Now we can apply Lemma \ref{lemma:technical1} 
with ${\cal P}={\cal W}$, ${\cal Q}={\cal V}$ and conclude that
$({\cal W}\, {\cal V}^{\dag})_{<0}=\p_x^{-1}$. Since the operator
in the left-hand side is of negative order, we have
${\cal W}\, {\cal V}^{\dag}=\p_x^{-1}$, hence
$
{\cal V}=({\cal W}^{-1}\p_x^{-1})^{\dag}
$
and 
$
\varphi^* ({\bf t}, z)= ({\cal W}^{-1}\p_x^{-1})^{\dag}
e^{-\xi ({\bf t}, z)}.
$

Next, define the operators ${\cal L}$ and ${\cal B}_k$ 
as in (\ref{Laxmkp}) and
(\ref{Bkmkp}). Then (\ref{mkp5a}) implies
\beq\label{kp11mkp}
\oint_{C_{\infty}}  (\p_{t_n}-{\cal B}_n  )\,
\varphi ({\bf t}, z)\varphi^* ({\bf t'}, z)\, dz =0
\eeq
for all ${\bf t}$, ${\bf t'}$.
We have, using (\ref{mkp12}) and (\ref{Laxmkp}):
$$
(\p_{t_n}-{\cal B}_n  )\varphi ({\bf t}, z)=
(\p_{t_n}{\cal W} +{\cal W}\p_x^n -{\cal B}_n {\cal W})
e^{\xi ({\bf t}, z)}
$$
$$
= (\p_{t_n}{\cal W} +{\cal L}^n {\cal W}  -{\cal B}_n {\cal W})
e^{\xi ({\bf t}, z)}
=(\p_{t_n}{\cal W} +({\cal L}^n)_{\leq 0} {\cal W})e^{\xi ({\bf t}, z)}.
$$
Again, putting $t_k=t_k'$ for $k\geq 2$, we rewrite (\ref{kp11mkp}) as
$$
\oint_{C_{\infty}}  
({\cal R} e^{xz })\, (({\cal W}^{-1}\p_x^{-1})^{\dag}e^{-x'z})
\, dz =0,
$$
where ${\cal R}$ is the operator
$
{\cal R}=\p_{t_n}{\cal W} +({\cal L}^n)_{\leq 0} {\cal W}
$.
By Lemma \ref{lemma:technical1} (with $c=0$), we 
conclude that $({\cal R}\, {\cal W}^{-1}\p_x^{-1})_{<0}=0$.
Clearly, ${\cal R}$ contains non-positive powers of $\p_x$ only 
while the operator
${\cal R}\, {\cal W}^{-1}\p_x^{-1}$ is of strictly negative
order i.e., $({\cal R}\, {\cal W}^{-1}\p_x^{-1})_{\leq 0}
={\cal R}\, {\cal W}^{-1}\p_x^{-1}=0$.
Since ${\cal W}$ is an invertible operator of order
zero, it follows from this that 
${\cal R}=0$,
which proves (\ref{mkp11}).
\square

\begin{corollary} \label{corollary:mkp1}
The function
$\varphi ({\bf t}, z)$ is an eigenfunction of the Lax operator ${\cal L}$
(\ref{Laxmkp})
with eigenvalue $z$,
\beq\label{kp7mkp}
{\cal L}\varphi ({\bf t}, z)=z\varphi ({\bf t}, z),
\eeq
and for any $z\in \CC$ it satisfies the differential equations
\beq\label{mkp13}
\p_{t_k}\varphi ({\bf t}, z)={\cal B}_k \varphi ({\bf t}, z) \;\;
\mbox{for any $k\geq 1$},
\eeq
where the $k$-th order differential operators ${\cal B}_k$ of the form
$\displaystyle{
{\cal B}_k =\p_x^k +\sum_{j=1}^{k-1}c_k \p_x^j}
$
are given by equation (\ref{Bkmkp}). 
\end{corollary}

\noindent
{\it Proof.}
The eigenvalue equation (\ref{kp7mkp}) is a trivial consequence 
of (\ref{mkp12}), (\ref{Laxmkp}). Equation (\ref{mkp13}) easily 
follows from (\ref{mkp11}):
$$
\begin{array}{lll}
\p_{t_k}\varphi ({\bf t}, z)& =& \p_{t_k}{\cal W} \, 
e^{\xi ({\bf t}, z)}+
{\cal W} z^k e^{\xi ({\bf t}, z)}
\\ && \\
&=&-({\cal L}^k)_{\leq 0}{\cal W}  e^{\xi ({\bf t}, z)}
+{\cal W}\p_x^k e^{\xi ({\bf t}, z)}
\\ && \\
&=&\Bigl (-({\cal L}^k)_{\leq 0}+{\cal L}^k\Bigr )
{\cal W}  e^{\xi ({\bf t}, z)}=({\cal L}^k)_{\geq 1}\varphi ({\bf t}, z).
\end{array}
$$
\square

There are similar statements for the 
adjoint BA function $\varphi^* ({\bf t}, z)$.
From the proof of Theorem \ref{theorem:mkp1a} we already know
that
\beq\label{kp12mkp}
\varphi^*({\bf t}, z)=({\cal W}^{-1}\p_x^{-1})^{\dag} 
e^{-\xi ({\bf t}, z)}.
\eeq
Using (\ref{mkp11}) one can prove the following corollary.
\begin{corollary}
The function $\varphi^* ({\bf t}, z)$ satisfies  
the equations
\beq\label{kp8mkp}
{\cal L}^{*}\varphi^* ({\bf t}, z)=z\varphi^* ({\bf t}, z),
\eeq
\beq\label{kp8amkp}
-\p_{t_k}\varphi^* ({\bf t}, z)=
{\cal B}^* \varphi^* ({\bf t}, z) \;\;\;
\mbox{for any $k\geq 1$},
\eeq
where
\beq\label{mkp14}
{\cal L}^*=\p_x^{-1}{\cal L}^{\dag}_k  \,\p_x, \qquad
{\cal B}^*=\p_x^{-1}{\cal B}^{\dag}_k  \,\p_x.
\eeq
\end{corollary}

\begin{remark}
\label{remark:mkpWtau}
Taking into account equations (\ref{kp4}), we can represent
the dressing operator ${\cal W}$ as
\beq\label{dres1}
{\cal W}=\normord \frac{\tau ({\bf t}-[\p_x^{-1}])}{\hat \tau ({\bf t})}
\normord , \qquad
({\cal W}^{\dag})^{-1}=\normord 
\frac{\hat \tau ({\bf t}+[\p_x^{-1}])}{\tau ({\bf t})}
\normord 
\eeq
where the normal ordering means that operator $\p_x$ 
in each term of the expansion in negative powers of $\p_x^{-1}$
is moved to the right. 
(cf. Remark \ref{remark:kpWtau}).
\end{remark}

Like in the case of the KP hierarchy,
compatibility of equations (\ref{kp7mkp}),
(\ref{mkp13}) implies the 
Lax equations
\beq\label{kp9mkp}
\p_{t_k}{\cal L}=[{\cal B}_k, {\cal L}], \quad k\geq 1,
\eeq
and equations of the Zakharov-Shabat type
\beq\label{kp10mkp}
\p_{t_n}{\cal B}_m -\p_{t_m}{\cal B}_n +[{\cal B}_m, {\cal B}_n]=0, 
\quad n,m \geq 1.
\eeq
Usually, this infinite system of equations 
serves as a definition of the mKP hierarchy.

Wave functions for the mKP hierarchy are introduced in the same way
as in the KP case (see Definition \ref{definition:wavekp}).

\begin{definition}\label{definition:wavemkp}
Any solution $S ({\bf t})$ to the system of linear 
differential equations (\ref{mkp13}), i.e.,
\beq\label{wave1mkp}
\p_{t_k}S ({\bf t})={\cal B}_k S ({\bf t}) \;\;\;\;
\mbox{for any $k\geq 1$},
\eeq 
is called a wave function for the mKP hierarchy.
Any solution $S^* ({\bf t})$ to the system (\ref{kp8amkp}), i.e.,
\beq\label{wave2mkp}
-\p_{t_k}S^* ({\bf t})={\cal B^*}_k S^* ({\bf t}) \;\;\;\;
\mbox{for any $k\geq 1$},
\eeq
is called an adjoint wave function.
\end{definition}

\noindent
For brevity, we will sometimes call both 
$S$ and $S^*$ wave functions.

Like in the KP case, mKP wave functions of general form
can be obtained by integrating 
the mKP BA functions with respect to the
spectral parameter $z$ with arbitrary density functions 
or distributions
$\sigma (z)$, $\sigma^*(z)$:
\beq\label{mkp15}
\begin{array}{l}
\displaystyle{
S({\bf t})=\int_{\CCC}\varphi ({\bf t}; z)\sigma (z)d^2 z,}
\\ \\
\displaystyle{
S^* ({\bf t})=\int_{\CCC}\varphi^* ({\bf t}; z)\sigma^* (z)d^2 z}.
\end{array}
\eeq

\begin{remark}
Since the operators ${\cal B}_k$, $\hat {\cal B}_k$ are of the
form ${\cal A}\, \p_x$, where ${\cal A}$ is a differential operator,
the simplest mKP wave functions are $S({\bf t})=S^*({\bf t})=1$.
\end{remark}

\section{From mKP to SchKP (and HD) and back}

\subsection{BD transformations of the mKP hierarchy}
\label{section:BDmKP}

In the previous section the mKP hierarchy was described in terms of 
a BD transformation of KP. Such a BD transformation of any solution
$\tau ({\bf t})$ of the KP hierarchy to another solution
$\hat \tau ({\bf t})$ is determined by choosing a wave function
$\phi ({\bf t})$ of the KP hierarchy, then $\hat \tau ({\bf t})=
\phi ({\bf t})\tau ({\bf t})$. The pair $(\tau , \hat \tau )$
obtained in this way is a solution of the mKP hierarchy.

A similar description of the SchKP and HD hierarchies
requires that in addition to the already selected 
KP wave function $\phi$, 
one should fix another one, $\phi '({\bf t})$. 
This wave function defines another BD 
transformation $\tau \to \bar \tau =\phi '\tau$ for KP, 
and thus the pair $(\tau , \bar \tau )$ 
is some other solution to the mKP hierarchy.
The transition from $(\tau , \hat \tau )$ to
$(\tau , \bar \tau )$ is nothing else than a BD transformation
for the mKP hierarchy. Alternatively (and equivalently) it can
be described in the same manner as it was done for the KP hierarchy:
given a solution $(\tau , \hat \tau )$, we construct another 
solution by fixing a wave function $S({\bf t})$ of the mKP hierarchy,
then $\bar \tau =S \tau$, with $S=\phi '/\phi$.

As a preliminary step,
we should obtain a bilinear equation which connects the tau-functions
$\tau$, $\hat \tau$ and $\bar \tau$. To this end, we need to
consider composition of the KP BD transformations $\tau \to \hat \tau$
and $\tau \to \bar \tau$ which sends $\tau$ to some 
other KP tau-function which we denote as $\hat{\bar \tau}$.
It is instructive to define it in a slightly more general setup. 
Consider 
two chains of forward BD transformations
defined by fixing two sequences of density functions
$\{\rho_1^{(1)}(z), \rho_2^{(1)}(z), \ldots \}$ and
$\{\rho_1^{(2)}(z), \rho_2^{(2)}(z), \ldots \}$. Starting from
any initial KP tau-function $\tau ({\bf t})$, one can apply 
the first transformation $n_1$ times and then the second one
$n_2$ times. We denote the resulting tau-function as 
$\tau (n_1, n_2, {\bf t})$. Let us denote the first and
the second transformations by ${\bf B}_1$ and
${\bf B}_2$, then we can write:
\beq\label{B1B2}
\begin{array}{l}
\displaystyle{
{\bf B}_1\tau (n_1, n_2, {\bf t})=
\tau (n_1+1, n_2, {\bf t})=
\int_{\CCC}e^{\xi ({\bf t}, z)}z^{n_1+n_2}\rho^{(1)}_{n_1+1}(z)
\tau (n_1, n_2, {\bf t}-[z^{-1}])d^2 z,}
\\ \\
\displaystyle{
{\bf B}_2\tau (n_1, n_2, {\bf t})=
\tau (n_1, n_2+1, {\bf t})=
\int_{\CCC}e^{\xi ({\bf t}, z)}z^{n_1+n_2}\rho^{(2)}_{n_2+1}(z)
\tau (n_1, n_2, {\bf t}-[z^{-1}])d^2 z.}
\end{array}
\eeq
There are two ways to obtain $\tau (n_1+1, n_2+1, {\bf t})$
from $\tau (n_1, n_2, {\bf t})$ (first apply ${\bf B}_1$, then
${\bf B}_2$ or vice versa):
$$
\tau (n_1, n_2, {\bf t})\stackrel{{\bf B}_1}{\longrightarrow}
\tau (n_1+1, n_2, {\bf t})
\stackrel{{\bf B}_2}{\longrightarrow}
\tau (n_1+1, n_2+1, {\bf t}),
$$
$$
\tau (n_1, n_2, {\bf t})\stackrel{{\bf B}_2}{\longrightarrow}
\tau (n_1, n_2+1, {\bf t})
\stackrel{{\bf B}_1}{\longrightarrow}
\tau (n_1+1, n_2+1, {\bf t}).
$$
Writing the two composite transformations explicitly as
double integrals, one can
easily see that the results differ by a sign, i.e., the
anti-commutation relation
\beq\label{b18}
{\bf B}_1 {\bf B}_2=-{\bf B}_2 {\bf B}_1
\eeq
holds true. Therefore, given $\tau (n_1, n_2, {\bf t})$, 
the function $\tau (n_1+1, n_2+1, {\bf t})$ is defined
only up to a sign. For our needs, it is enough to fix this uncertainty
by the convention that the transformation $\tau \to \hat \tau$ (which is
${\bf B}_1$) is applied first. So, we have four KP tau-functions
\beq\label{four}
\tau ({\bf t})=\tau (0,0, {\bf t}), \quad
\hat \tau ({\bf t})=\tau (1,0, {\bf t}),
\quad
\bar \tau ({\bf t})=\tau (0,1, {\bf t}),
\quad
\hat{\bar \tau} ({\bf t})=\tau (1,1, {\bf t}).
\eeq

\begin{proposition}
The four KP tau-functions (\ref{four}) satisfy the bilinear
equation
\beq\label{b1d}
a\Bigl (\hat \tau \bigl ({\bf t}-[a^{-1}]\bigr ) 
\bar \tau \bigl ({\bf t}\bigr ) -
\bar \tau \bigl ({\bf t}-[a^{-1}]\bigr ) 
\hat \tau \bigl ({\bf t}\bigr )
\Bigr )
=\hat{\bar \tau} ({\bf t}\bigr )\, \tau ({\bf t}-[a^{-1}]\bigr ),
\eeq
or, in the short-hand notation,
\beq\label{b1c}
a_1\Bigl (\hat \tau^1 \bar \tau -\bar \tau^1 \hat \tau \Bigr )
=\hat{\bar \tau} \, \tau^1 .
\eeq
\end{proposition}

\noindent
{\it Proof.}
Explicitly, the tau-functions $\hat \tau$, $\bar \tau$
and $\hat{\bar \tau}$ are given by 
$$
\begin{array}{l}
\displaystyle{
\hat \tau ({\bf t})=\int_{\CCC}e^{\xi ({\bf t}, z)}\hat \rho (z)
\tau \Bigl ({\bf t}-[z^{-1}]\Bigr ) d^2 z,}
\\ \\
\displaystyle{
\bar \tau ({\bf t})=\int_{\CCC}e^{\xi ({\bf t}, w)}\bar \rho (w)
\tau \Bigl ({\bf t}-[w^{-1}]\Bigr ) d^2 w,}
\end{array}
$$
$$
\hat{\bar \tau}({\bf t})=
\int_{\CCC}\!\int_{\CCC} e^{\xi ({\bf t}, z)+\xi ({\bf t}, w)}
\hat \rho (z)\bar \rho (w)(w-z) 
\tau \Bigl ({\bf t}-[z^{-1}]-[w^{-1}]\Bigr )d^2 z d^2 w.
$$
Using these formulas, we can write:
\beq\label{b19}
\begin{array}{l}
a_1\Bigl (\hat \tau \bigl ({\bf t}-[a_1^{-1}]\bigr ) 
\bar \tau \bigl ({\bf t}\bigr ) -
\bar \tau \bigl ({\bf t}-[a_1^{-1}]\bigr ) 
\hat \tau \bigl ({\bf t}\bigr )
\Bigr )
-\hat{\bar \tau} ({\bf t}\bigr )\, \tau ({\bf t}-[a^{-1}]\bigr )
\\ \\
\phantom{aaaaaaaaaaaaaaaaaaa}\displaystyle{=
\int_{\CCC}\!\int_{\CCC} e^{\xi ({\bf t}, z)+\xi ({\bf t}, w)}
\hat \rho (z)\bar \rho (w) \, T(z,w, {\bf t}) d^2 z d^2 w,}
\end{array}
\eeq
where
$$
\begin{array}{lll}
T(z,w, {\bf t})&=&(a_1-z)\tau \Bigl ({\bf t}-[z^{-1}]-[a_1^{-1}]\Bigr )
\tau \Bigl ({\bf t}-[w^{-1}]\Bigr ) 
\\ && \\
&& -(a_1-w)\tau \Bigl ({\bf t}-[w^{-1}]-[a_1^{-1}]\Bigr )
\tau \Bigl ({\bf t}-[z^{-1}]\Bigr ) 
\\ && \\
&&-(w-z)\tau \Bigl ({\bf t}-[z^{-1}]-[w^{-1}]\Bigr )
\tau \Bigl ({\bf t}-[a_1^{-1}]\Bigr ).
\end{array}
$$
But since $\tau ({\bf t})$ is a KP 
tau-function, $T(z,w, {\bf t})=0$ by virtue of equation (\ref{kp1b}),
where one should put $a_2=z$, $a_3=w$. 
\square

We will need two corollaries of equation (\ref{b1c}) obtained by
expanding the both sides in inverse powers of $a_1$ as
$a_1\to \infty$. In the leading order we obtain the equation
\beq\label{a1}
\p_{x}\log \frac{\bar \tau}{\hat \tau}=
\frac{\hat{\bar \tau}\, \tau}{\hat \tau \, \bar \tau}
\eeq
(recall that $x=t_1$). The expansion to the next order
gives the equation
\beq\label{a2}
\p_{t_2}\log \frac{\bar \tau}{\hat \tau} -
\p^2_x \log \frac{\bar \tau}{\hat \tau} +
(\p_x \log \hat \tau )^2 -(\p_x \log \bar \tau )^2
+2 \frac{\hat{\bar \tau}\, \tau}{\hat \tau \, \bar \tau}\,
\p_x \log \tau =0.
\eeq
Combining it with (\ref{a1}), one can exclude $\hat{\bar \tau}$
and obtain the equation 
\beq\label{a3}
\p_{t_2}\log \frac{\bar \tau}{\hat \tau} -
\p^2_x \log \frac{\bar \tau}{\hat \tau} -\Bigl (
\p_x \log \frac{\bar \tau}{\hat \tau}\Bigr )^2 -2
\Bigl (\p_x \log \frac{\hat \tau}{\tau}\Bigr )
\Bigl (\p_x  \log \frac{\bar \tau}{\hat \tau}\Bigr )
=0,
\eeq
which will be used below.

It is convenient to depict the four KP tau-functions
(\ref{four}) on the following diagram:
\beq\label{diag1}
\begin{array}{ccc}
\bar \tau  & \stackrel{\bar \phi}\longrightarrow & \hat{\bar \tau}
\\
\uparrow\lefteqn{{\scriptstyle \phi '} } 
&& \uparrow\lefteqn{{\scriptstyle \hat \phi}}
\\ 
\tau &\stackrel{\phi}\longrightarrow &\hat \tau
\end{array}
\eeq

Each arrow in (\ref{diag1}) is equipped with the KP wave
function that defines the corresponding mKP solution (i.e., the
KP BD transformation). They are:
\beq\label{phiphi}
\phi =\frac{\hat \tau}{\tau}, \quad \phi ' =\frac{\bar \tau}{\tau},
\quad \bar \phi =\frac{\hat{\bar \tau}}{\bar \tau}, \quad
\hat \phi =\frac{\hat{\bar \tau}}{\hat \tau}.
\eeq
The first two are associated with the KP Lax operator $L$
corresponding to the initial solution with the tau-function $\tau$,
while the other two, $\bar \phi$ and $\hat \phi$, are associated
with the KP Lax operators $\bar L$ and $\hat L$ corresponding to
the KP solutions with the tau-functions $\bar \tau$ and $\hat \tau$
respectively. The mKP Lax operators for the four mKP solutions
$(\tau , \hat \tau )$, $(\tau , \bar \tau )$, 
$(\hat \tau , \hat{\bar \tau})$, $(\bar \tau , \hat{\bar \tau})$
will be denoted as ${\cal L}$, ${\cal L}'$, $\hat {\cal L}$,
$\bar {\cal L}$ respectively. Each of them is obtained by 
``dressing'' the differential operator $\p_x$ with the dressing 
operators ${\cal W}$, ${\cal W}'$, $\hat {\cal W}$,
$\bar {\cal W}$ as in (\ref{Laxmkp}): 
\beq\label{diag2}
\begin{array}{ccc}
\bar \tau  & \stackrel{\bar{\cal W}}\longrightarrow & \hat{\bar \tau}
\\
\uparrow\lefteqn{{\scriptstyle {\cal W}'} } 
&& \uparrow\lefteqn{{\scriptstyle \hat{\cal W}}}
\\ 
\tau &\stackrel{{\cal W}}\longrightarrow &\hat \tau
\end{array}
\eeq
Accordingly, along with $\varphi ({\bf t}, z)$ 
given by (\ref{mkp9}), we can introduce the BA functions for the
other three solutions:
\beq\label{BA3}
\begin{array}{l}
\displaystyle{ \varphi '({\bf t}, z)=
{\cal W}' e^{\xi ({\bf t}, z)}=
e^{\xi ({\bf t}, z)}
\frac{\tau ({\bf t}-[z^{-1}])}{\bar \tau ({\bf t})}},
\\ \\
\displaystyle{ \hat \varphi ({\bf t}, z)=
\hat {\cal W} e^{\xi ({\bf t}, z)}=
e^{\xi ({\bf t}, z)}
\frac{\hat \tau ({\bf t}-[z^{-1}])}{\hat{\bar \tau} ({\bf t})}},
\\ \\
\displaystyle{ \bar \varphi ({\bf t}, z)=
\bar {\cal W} e^{\xi ({\bf t}, z)}=
e^{\xi ({\bf t}, z)}
\frac{\bar \tau ({\bf t}-[z^{-1}])}{\hat{\bar \tau} ({\bf t})}}.
\end{array}
\eeq

The next theorem states that the mKP dressing
operators ${\cal W}'$, $\hat{\cal W}$, $\bar{\cal W}$
can be expressed through ${\cal W}$ and the ratio
\beq\label{S1}
S({\bf t})=\frac{\bar \tau ({\bf t})}{\hat \tau ({\bf t})},
\eeq
which in what follows will be identified with a wave function
for the mKP Lax operator ${\cal L}$ in the sense
of Definition \ref{definition:wavemkp} (see Theorem \ref{theorem:S}
below).

\begin{theorem}
\label{theorem:WWW}
The dressing operators 
${\cal W}'$, $\hat{\cal W}$, $\bar{\cal W}$ in the diagram
(\ref{diag2}) are expressed through ${\cal W}$ and $S$ given by
(\ref{S1}) as follows:
\beq\label{WWW}
\begin{array}{l}
{\cal W}'=S^{-1} {\cal W},
\\ \\
\hat {\cal W}=S_x^{-1}\p_x {\cal W} \p_x^{-1},
\\ \\
\bar {\cal W}=S_x^{-1}S^2 \p_x S^{-1} {\cal W} \p_x^{-1},
\end{array}
\eeq
where the right-hand sides are to be understood
as compositions of the operators and
$S_x =\p_x S$.
\end{theorem}

\noindent
{\it Proof.}
The first equality in (\ref{WWW}) is evident from the definitions.
To prove the second one, it is enough to check the equality
\beq\label{WWW1}
e^{\xi ({\bf t}, z)}
\frac{\hat \tau ({\bf t}-[z^{-1}])}{\hat{\bar \tau} ({\bf t})}=
z^{-1}S_x^{-1}({\bf t}) \p_x \varphi ({\bf t}, z)
\eeq
which is the result of acting by both sides of
$\hat {\cal W}=S_x^{-1}\p_x {\cal W} \p_x^{-1}$ to the 
exponential function $e^{\xi ({\bf t}, z)}$. Plugging here
$\varphi$ from (\ref{mkp9}), we rewrite (\ref{WWW1}) as
\beq\label{WWW2}
S_x\, \frac{\hat \tau^{[z]}}{\hat{\bar \tau}}=
\frac{\tau^{[z]}}{\hat \tau}\Bigl (
1 +z^{-1}\p_x \log \frac{\tau^{[z]}}{\hat \tau}\Bigr ),
\eeq
where for the sake of brevity we have 
denoted $\tau^{[z]} =\tau ({\bf t}-[z^{-1}])$, 
$\hat \tau^{[z]} =\hat \tau ({\bf t}-[z^{-1}])$.
Equation (\ref{a1}) states that
\beq\label{Sx}
S_x=\frac{\hat{\bar \tau} \tau}{\hat \tau^2},
\eeq
which is to be substituted into the the left-hand side of
(\ref{WWW2}).
To transform the right-hand side, we use equation (\ref{mkp5c}).
Letting $a_2\to \infty$ and $a_1 =z$ in it, we obtain the 
equality
$$
1 +z^{-1}\p_x \log \frac{\tau^{[z]}}{\hat \tau}=
\frac{\hat \tau^{[z]} \tau}{\hat \tau \tau^{[z]}},
$$
which, together with (\ref{Sx}), converts equation (\ref{WWW2})
into identity. The third equation in (\ref{WWW}) can be 
proved in a similar way.
\square

\begin{remark}
The BD transformations of the mKP hierarchy defined by the
diagram (\ref{diag1}) were considered in \cite{OR,Cheng2018} 
(along with transformations of some other types) under
the name of ``gauge transformations''. 
\end{remark}

\subsection{From mKP to SchKP and back} 
\label{section:mKPSchKP}

The pair of
KP tau-functions $(\hat \tau , \bar \tau )$ 
considered in Section \ref{section:BDmKP} can be regarded as 
a solution of the SchKP and, after a change of variables, 
HD hierarchies.
Moreover, any solution of the latter hierarchies can be obtained
in this way.

It turns out that there exist bilinear equations which
connect $\hat \tau$ and $\bar \tau$, without any participation
of $\tau$.

\begin{theorem}\label{theorem:hat-tau-bar-tau}
Let the pairs $(\tau , \hat \tau )$, $(\tau , \bar \tau )$ 
be any two solutions of the mKP hierarchy with the same 
first members of the pairs. Then the KP tau-functions
$\hat \tau$, $\bar \tau$ satisfy the bilinear equation
\beq\label{bil3}
\oint_{C_{\infty}}
e^{\xi ({\bf t}-{\bf t'}, \lambda )}\Bigl [
\hat \tau ({\bf t}-[\lambda^{-1}])\bar \tau ({\bf t'}+[\lambda^{-1}])+
\hat \tau ({\bf t'}+[\lambda^{-1}])\bar \tau ({\bf t}-
[\lambda^{-1}])\Bigr ]\, d\lambda =0
\eeq
valid for all ${\bf t}, \, {\bf t'}$.
\end{theorem}

\noindent
{\it Proof.}
The first step is to make the Miwa substitution (\ref{miwa1})
in (\ref{bil3}), 
which converts the equation to be proved into the following one:
\beq\label{bil4}
\begin{array}{l}
\displaystyle{
\sum_{k=1}^{m+1}(-1)^{k-1}
\Delta_m (a_1, \ldots , \hat a_k, \ldots , a_{m+1})
\prod_{i=1}^{m-1}(b_i-a_k)}
\\ \\
\phantom{aaaaaaa}
\displaystyle{\times \, \left [
\hat \tau \Bigl ({\bf t}-\sum_{j=1, \neq k}^{m+1}[a_j^{-1}]\Bigr )
\bar \tau \Bigl ({\bf t}-[a^{-1}_k]- 
\sum_{j=1}^{m-1}[b_j^{-1}]\Bigr ) \right. }
\\ \\
\phantom{aaaaaaaaaaaaaaa}\displaystyle{+ \, \left.
\bar \tau \Bigl ({\bf t}-\sum_{j=1, \neq k}^{m+1}[a_j^{-1}]\Bigr )
\hat \tau \Bigl ({\bf t}-[a^{-1}_k]- 
\sum_{j=1}^{m-1}[b_j^{-1}]\Bigr )\right ]
=0, \quad m\geq 2.}
\end{array}
\eeq
(The residue calculus is the same as the one leading to (\ref{g2b}).)
In fact, equation (\ref{bil4}) 
is actually equivalent to (\ref{bil3}).
Namely, the argument from \cite{Shigyo13}
(see the proof of Proposition 1 there), which uses properties of
symmetric functions, allows us to prove 
the following proposition.

\begin{proposition}\label{proposition:equiv}
Equations (\ref{bil3}) and (\ref{bil4}) are equivalent
(provided that the latter one 
is valid for all $m\geq 2$ and all sets of points).
\end{proposition}

\noindent
Using the short-hand notation (\ref{not1}), we write:
$$
\hat \tau \Bigl ({\bf t}-\!\! 
\sum_{j=1, \neq k}^{m+1}[a_j^{-1}]\Bigr )=
\hat \tau^{1\, \ldots \, \hat k \ldots \, m\! + \! 1}, \qquad
\bar \tau \Bigl ({\bf t}-\! [a_k^{-1}]-\!\! 
\sum_{l=1}^{m-1}[b_l^{-1}]\Bigr )=
\bar \tau^{k \tilde 1 \, \ldots \, \widetilde{m\! - \! 1}},
$$
where $\tilde l$ is related to $b_l$. This allows us to
represent 
(\ref{bil4}) in a more compact form:
\beq\label{bil6}
\begin{array}{l}
\displaystyle{
\sum_{k=1}^{m+1}(-1)^{k-1}
\Delta_m (a_1, \ldots , \hat a_k, \ldots , a_{m+1})
\Delta_m (a_k, b_1 , \, \ldots , \, , b_{m-1})}
\\ \\
\phantom{aaaaaaaaaaaaaaa}\displaystyle{
\times \, \Bigl (\hat \tau^{1\, \ldots \, \hat k \ldots \, m\! + \! 1}
\, 
\bar \tau^{k \tilde 1 \, \ldots \, \widetilde{m\! - \! 1}} \, +\,
\bar \tau^{1\, \ldots \, \hat k \ldots \, m\! + \! 1} \,
\hat \tau^{k \tilde 1 \, \ldots \, \widetilde{m\! - \! 1}}\Bigr )
=0.
}
\end{array}
\eeq

The next step is to use the first equation in (\ref{mkp2b}). 
Adopted to our case, it yields:
\beq\label{bil5}
\begin{array}{l}
\displaystyle{
\hat \tau \Bigl ({\bf t}-\!\! 
\sum_{j=1, \neq k}^{m+1}[a_j^{-1}]\Bigr )=
\int_{\CCC}d^2 z \hat \rho (z)\, e^{\xi ({\bf t}, z)}
\Bigl (\prod_{j=1, \neq k}^{m+1}\!\! \frac{a_j-z}{a_j}\Bigr )
\tau \Bigl ({\bf t} \! - \! [z^{-1}] -\!\! 
\sum_{j=1, \neq k}^{m+1}[a_j^{-1}]\Bigr )}
\end{array}
\eeq
\beq\label{bil5a}
\begin{array}{l}
\displaystyle{
\bar \tau \Bigl ({\bf t}-\! [a_k^{-1}]-\!\! 
\sum_{l=1}^{m-1}[b_l^{-1}]\Bigr )=
\int_{\CCC}d^2 \zeta \bar \rho (\zeta )\, e^{\xi ({\bf t}, \zeta )}
\frac{a_k-\zeta}{a_k}
\Bigl (\prod_{j=1, \neq k}^{m+1}\!\! \frac{b_l-\zeta }{b_l}\Bigr )}
\phantom{bbbbbbbbbbbbbbb}
\\ \\
\displaystyle{\phantom{aaaaaaaaaaaaaaaaaaaaaaaaaaaaaaaaaa}
\times \, \tau \Bigl ({\bf t} \! - \! [\zeta^{-1}] \! - \! [a_k^{-1}]\!
-\!\! 
\sum_{l=1}^{m-1}[b_l^{-1}]\Bigr )}.
\end{array}
\eeq
Then the left-hand side of (\ref{bil6}), call it $F$, 
becomes, up to an irrelevant
common factor:
$$
F=\int_{\CCC}\! \int_{\CCC} d^2 z d^2 \zeta \hat \rho (z)
\bar \rho (\zeta ) \, e^{\xi ({\bf t}, z) +\xi ({\bf t}, \zeta )}
\left (\sum_{k=1}^{m+1}(-1)^{k-1}\Bigl (R_k (z, \zeta )+R_k (\zeta , z )
\Bigr )
\right ),
$$
where
\beq\label{Rk}
\begin{array}{lll}
R_k(z, \zeta )&=&\Delta_{m+1} (a_1, \ldots , \hat a_k, \ldots , a_{m+1}, z)
\Delta_{m+1} (a_k, b_1 , \, \ldots , \, , b_{m-1}, \zeta )
\\ && \\
&& \times \, \,
\displaystyle{
\tau \Bigl ({\bf t} \! - \! [z^{-1}] -\!\! 
\sum_{j=1, \neq k}^{m+1}[a_j^{-1}]\Bigr )\,
\tau \Bigl ({\bf t} \! - \! [\zeta^{-1}] \! - \! [a_k^{-1}]\!-\!\! 
\sum_{l=1}^{m-1}[b_l^{-1}]\Bigr ).
}
\end{array}
\eeq
To unify the notation for the 
variables, denote $z=a_{m+2}$, $\zeta =b_m$. Then, using the 
notation as in (\ref{g2c}), we have:
$$
R_k(a_{m+2}, b_m)=\tau \bigl [a_1, \, \ldots , a_{m+2}\bigr ]\,
\tau \bigl [a_k, b_1, \, \ldots , b_m\bigr ].
$$
Then equation (\ref{g2c}) with the change $m\to m+1$ yields:
\beq\label{bil7}
\sum_{k=1}^{m+1}(-1)^{k-1}R_k(a_{m+2}, b_m)=(-1)^m\,
\tau \bigl [a_1, \ldots , a_{m+1}\bigr ]\, 
\tau \bigl [a_{m+2}, b_1, \ldots , b_{m}\bigr ].
\eeq
Since the second multiplier is antisymmetric with respect to
the interchange $a_{m+2}\leftrightarrow b_m$, we have
$$
\sum_{k=1}^{m+1}(-1)^{k-1}\Bigl (R_k(z, \zeta )+ 
R_k (\zeta , z)\Bigr )=0,
$$
so $F=0$ and equation (\ref{bil4}) is proved.
By virtue of Proposition \ref{proposition:equiv}, this 
proves equation (\ref{bil3}). 
\square

\begin{corollary}
For all $a,b\in \CC$, the function $\tau ({\bf t}|a,b)=
a \hat \tau ({\bf t}) +
b\bar \tau ({\bf t})$ is a tau-function of the KP hierarchy, i.e., it
satisfies the integral bilinear relation (\ref{kp1}).
More generally, let $n$ pairs $(\tau , \tau^{(\alpha )})$ 
($\alpha =1, \ldots , n$) be 
solutions to the mKP hierarchy, then for all 
$c_1, \ldots , c_n \in \CC$ the function
$$
\tau ({\bf t}|\{c_1, \ldots , c_n\})=
\sum_{\alpha =1}^n c_{\alpha}\tau^{(\alpha )}({\bf t})
$$
is a KP tau-function.
\end{corollary}

\noindent
This statement immediately follows from (\ref{bil3})
and from the fact that all the $\tau^{(\alpha )}$'s are 
KP tau-functions.

\begin{corollary}
Let $(\tau , \hat \tau )$, $(\tau , \bar \tau )$ be any two
solutions of the mKP hierarchy. Then the KP tau-functions
$\hat \tau$, $\bar \tau$ satisfy the bilinear equation
\beq\label{bil1}
a_{12}\Bigl (\hat \tau^{12}\bar \tau^3 + 
\hat \tau^{3}\bar \tau^{12} \Bigr )+
a_{23}\Bigl (\hat \tau^{23}\bar \tau^1 + 
\hat \tau^{1}\bar \tau^{23} \Bigr )+
a_{31}\Bigl (\hat \tau^{13}\bar \tau^2 + 
\hat \tau^{2}\bar \tau^{13} \Bigr )=0,
\eeq
where the short-hand notations (\ref{not1}) are used.
\end{corollary}

\noindent
{\it Proof.} This is the simplest particular case of (\ref{bil4})
obtained at $m=2$, $b_1=\infty$.
\square

\noindent
Let us give another proof of equation (\ref{bil1}) which is not
based on explicit realization of the BD transformations.
The starting point is the following system of equations 
of the type (\ref{mkp5c}) written for the mKP solutions
$(\tau , \hat \tau)$ and $(\tau , \bar \tau)$:
$$
\left \{
\begin{array}{l}
a_1 \hat \tau^{13} \tau^{23} -a_2 \hat \tau^{23} \tau^{13} -
a_{12}\hat \tau^3 \tau^{123}=0,
\\ \\
a_1 \bar \tau^{13} \tau^{23} -a_2 \bar \tau^{23} \tau^{13} -
a_{12}\bar \tau^3 \tau^{123}=0,
\\ \\
a_2 \hat \tau^{21} \tau^{31} -a_3 \hat \tau^{31} \tau^{21} -
a_{23}\hat \tau^1 \tau^{123}=0,
\\ \\
a_2 \bar \tau^{21} \tau^{31} -a_3 \bar \tau^{31} \tau^{21} -
a_{23}\bar \tau^1 \tau^{123}=0,
\\ \\
a_3 \hat \tau^{32} \tau^{12} -a_1 \hat \tau^{12} \tau^{32} -
a_{31}\hat \tau^2 \tau^{123}=0,
\\ \\
a_3 \bar \tau^{21} \tau^{31} -a_1 \bar \tau^{12} \tau^{32} -
a_{31}\bar \tau^2 \tau^{123}=0.
\end{array}
\right.
$$
Let us consider the first four equation. They form a system
of homogeneous linear equations for 
$\tau^{12}, \tau^{23}, \tau^{31}, \tau^{123}$.
Existence of nontrivial solutions implies that the matrix
$$
M_{13} =
\left (
\begin{array}{cccc}
0 & a_1\hat \tau^{13} & -a_2\hat \tau^{23}  &-a_{12}\hat \tau^{3}  
\\ &&& \\
0 & a_1\bar \tau^{13} & -a_2\bar \tau^{23}  &-a_{12}\bar \tau^{3}  
\\ &&& \\
-a_3\hat \tau^{31} & 0 & a_2\hat \tau^{21} & -a_{23}\hat \tau^{1}  
\\ \\
-a_3\bar \tau^{31} & 0 & a_2\bar \tau^{21} & -a_{23}\bar \tau^{1}  
\end{array}
\right ).
$$
is degenerate, i.e., $\det M_{13}=0$. Calculation of the
determinant gives:
$$
\frac{\det M_{13}}{a_1a_2a_3}= \hat \tau^{13}\bar \tau^{13}R -
(\hat \tau^{13})^2 \bar H -(\bar \tau^{13})^2 \hat H,
$$
where
$$
\begin{array}{l}
R= a_{12}\Bigl (\hat \tau^{12}\bar \tau^3 + 
\hat \tau^{3}\bar \tau^{12} \Bigr )+
a_{23}\Bigl (\hat \tau^{23}\bar \tau^1 + 
\hat \tau^{1}\bar \tau^{23} \Bigr )+
a_{31}\Bigl (\hat \tau^{13}\bar \tau^2 + 
\hat \tau^{2}\bar \tau^{13} \Bigr ),
\\ \\
\hat H =a_{12}\hat \tau^{12}\hat \tau^3 +
a_{23}\hat \tau^{23}\hat \tau^1 +
a_{31}\hat \tau^{31}\hat \tau^2,
\\ \\
\bar H =a_{12}\bar \tau^{12}\bar \tau^3 +
a_{23}\bar \tau^{23}\bar \tau^1 +
a_{31}\bar \tau^{31}\bar \tau^2 .
\end{array}
$$
Since both $\hat \tau$ and $\bar \tau$ are KP tau-functions,
$\bar H=\hat H =0$ by virtue of equation
(\ref{kp1c}). Therefore, $\det M_{13}=0$ if and only if
$R=0$.

\begin{remark}
Equation (\ref{bil1}) alone does not imply that 
$\hat \tau$ and $\bar \tau$ are KP tau-functions. 
The simplest counterexample is $\hat \tau ({\bf t})=1$ (satisfies KP),
$\bar \tau ({\bf t})=t_n$ for $n\geq 2$ (does not satisfy KP), but 
this pair solves equation (\ref{bil1}). 
\end{remark}

\begin{corollary}
The function
\beq\label{SS}
S({\bf t})=\frac{\bar \tau ({\bf t})}{\hat \tau ({\bf t})}
\eeq
satisfies the linear equation
\beq\label{SS1}
a_{12}\hat \tau^3 \hat \tau^{12}(S^3 + S^{12})+
a_{23}\hat \tau^1 \hat \tau^{23}(S^1 + S^{23})+
a_{31}\hat \tau^2 \hat \tau^{31}(S^2 + S^{31})=0,
\eeq
or, in a more compact form,
\beq\label{SS2}
\sum_{(\alpha \beta \gamma )=(123)
\atop 
\rm and \; cycles}
a_{\alpha \beta}
\hat \tau^{\alpha \beta}\hat \tau^{\gamma}
(S^{\alpha \beta}+S^{\gamma})=0.
\eeq
\end{corollary}

\noindent
{\it Proof.} This directly follows from equation
(\ref{bil1}) after the substitution (\ref{SS}).
\square

There are more equations for the $S$-function which follow
from (\ref{bil1}) just in an equally simple way as (\ref{SS2}). 
An analog of
(\ref{SS2}) with the $\bar \tau$-function instead of $\hat \tau$ is
\beq\label{SS2a}
\sum_{(\alpha \beta \gamma )=(123)
\atop 
\rm and \; cycles}
a_{\alpha \beta}
\bar \tau^{\alpha \beta}\bar \tau^{\gamma}
\Bigl (\frac{1}{S^{\alpha \beta}} +\frac{1}{S^{\gamma}} \Bigr )=0.
\eeq
From the fact that $\bar \tau$ is a KP tau-function, it follows that
along with (\ref{SS2}) the function $S$ satisfies the bilinear
equation
\beq\label{SS2b}
\sum_{(\alpha \beta \gamma )=(123)
\atop 
\rm and \; cycles}
a_{\alpha \beta}
\hat \tau^{\alpha \beta}\hat \tau^{\gamma}
S^{\alpha \beta}S^{\gamma}=0.
\eeq

\begin{theorem}
The function $S({\bf t})$ defined in (\ref{SS}) satisfies the
equation
\beq\label{SSc}
\frac{S^1 -S^{12}}{S^3 - S^{23}} \cdot
\frac{S^2 -S^{23}}{S^1 - S^{13}} \cdot
\frac{S^3 -S^{13}}{S^2 - S^{12}} =1,
\eeq
which is the discrete SchKP equation.
\end{theorem}

\noindent
{\it Proof.}
Let us write the 
three equations (\ref{SS2}), (\ref{SS2b}) and the 3-term
bilinear KP equation for $\hat \tau$ in the
matrix form:
\beq\label{m1}
\left ( \begin{array}{ccc}
1 & 1 & 1
\\
S^1 +S^{23} & S^2 +S^{13} & S^3 +S^{12}
\\
S^1S^{23} & S^2S^{13} & S^3S^{12}
\end{array} \right )
\left ( \begin{array}{c}
a_{23} \hat \tau^1 \hat \tau^{23}\\
a_{31} \hat \tau^2\hat \tau^{13} \\
a_{12} \hat \tau^3\hat \tau^{12}
\end{array} \right )=
0.
\eeq
In order for this system to have nonzero solutions, 
determinant of the $3\times 3$ matrix must vanish:
$$
\left |
\begin{array}{ccc}
1 & 1 & 1
\\ \\
S^1 +S^{23} & S^2 +S^{13} & S^3 +S^{12}
\\ \\
S^1S^{23} & S^2S^{13} & S^3S^{12}
\end{array} \right |=0.
$$
Calculation of the determinant gives equation (\ref{SSc}).
\square

\noindent
At the same time, equation (\ref{SSc}) generates the whole 
infinite hierarchy of SchKP partial differential
equations for $S({\bf t})$. To see this, recall that
$S^j =S({\bf t} -[a_j^{-1}])$, $S^{jk} =S({\bf t} -[a_j^{-1}]
-[a_k^{-1}])$. Letting successively $a_3\to \infty$, $a_2\to \infty$
and, eventually, $a_1\to \infty$, in (\ref{SSc}), and 
expanding the $S$-function in inverse powers of $a_1, a_2, a_3$,
one arrives at the equation
\begin{equation}\label{contSchKP}
    \left(4\frac{S_{t}}{S_x}-\{S,x\}\right)_x
    =3\left(\frac{S_{y}}{S_x}\right)_y+\frac{3}{2}
    \left(\frac{S_{y}^2}{S_x^2}\right)_x,
\end{equation}
where we put $t_1=x$, $t_2=y$, $t_3=t$ and
\beq\label{Sder}
\{S,x\}=\frac{S_{xxx}}{S_x}-\frac{3}{2}\left(\frac{S_{xx}}{S_x}\right)^2
\eeq
is the Schwarzian derivative.
This is the (2+1)-dimensional 
Schwarzian KP equation. Its solutions that do not depend on $y$
satisfy the simpler (1+1)-dimensional Schwarzian KdV equation
\begin{equation}\label{SchKdV}
    4S_{t}=S_x\{S,x\} .
\end{equation}

Let us present some other forms and corollaries of the linear
equation (\ref{SS1}) for the $S({\bf t})$. 
First of all we note that it can be rewritten
in the form
\beq\label{SS3}
a_1^{-1}(a_2^{-1}-a_3^{-1}) \hat U^{23}\Delta^{23}S +
a_2^{-1}(a_3^{-1}-a_1^{-1}) \hat U^{31}\Delta^{31}S +
a_3^{-1}(a_1^{-1}-a_2^{-1}) \hat U^{12}\Delta^{12}S=0,
\eeq
or, after an overall shift of the variables,
\beq\label{SS3a}
a_1^{-1}(a_2^{-1}-a_3^{-1}) \hat U_{23}\Delta_{23}S +
a_2^{-1}(a_3^{-1}-a_1^{-1}) \hat U_{31}\Delta_{31}S +
a_3^{-1}(a_1^{-1}-a_2^{-1}) \hat U_{12}\Delta_{12}S=0,
\eeq
where
$$
\hat U^{\alpha \beta}=\frac{\hat \tau \, 
\hat \tau^{\alpha \beta}}{\hat \tau^{\alpha}
\hat \tau^{\beta}}, 
\qquad
\hat U_{\alpha \beta}=\frac{\hat \tau \, 
\hat \tau_{\alpha \beta}}{\hat \tau_{\alpha}
\hat \tau_{\beta}}
$$
and
$$
\Delta^{\alpha \beta}S=S^{\alpha \beta}+S -S^{\alpha} -S^{\beta},
\qquad
\Delta_{\alpha \beta}S=S_{\alpha \beta}+S -S_{\alpha} -S_{\beta}.
$$
So, $\Delta_{\alpha \beta}$ is the difference operator
\beq\label{Delta}
\Delta_{\alpha \beta}=1+e^{D(a_{\alpha})+D(a_{\beta})}-
e^{D(a_{\alpha})}-e^{D(a_{\beta})},
\eeq
where the operator $D(a)$ is defined in (\ref{D(z)}), and 
similarly for the $\Delta^{\alpha \beta}$.

It is not difficult to see that in the limit $a_3\to \infty$ 
(\ref{SS3a}) converts into the equation
\beq\label{SS3c}
(a_1-a_2) \hat U_{12} \Delta_{12} S +\p_{t_1} (S_1- S_2)=0.
\eeq
The further limit $a_2\to \infty$ yields the equation
$$
\Bigl (\p_{t_1}\log (\hat \tau_1/\hat \tau)-a_1\Bigr ) 
\p_{t_1}(S_1 -S) +\p_{t_1}^2 S +
\frac{1}{2}(\p_{t_2}+\p_{t_1}^2)(S_1 -S)=0,
$$
or
\beq\label{SS3d}
\p_{t_1}\log (\hat \tau_1/\hat \tau)\Bigl (e^{D(a_1)}\! -\! 1\Bigr )
\p_{t_1}S +\frac{1}{2}\Bigl (e^{D(a_1)}\! -\! 1\Bigr )\Bigl (
\p_{t_2}S +\p_{t_1}^2S -2a_1 \p_{t_1}S\Bigr )+\p_{t_1}^2 S=0.
\eeq

The notation $S({\bf t})$ for the ratio (\ref{SS}) is deliberately
chosen to be the same as the one introduced in Theorem \ref{definition:wavemkp}. 
The next theorem provides identification of the $S({\bf t})$ with
a wave function of the mKP hierarchy.

\begin{theorem}\label{theorem:S}
Let $\phi ({\bf t})$, $\phi '({\bf t})$ be two  
wave functions of the KP hierarchy associated 
with the KP tau-functions $\hat \tau ({\bf t})$ and
$\bar \tau ({\bf t})$ respectively. Then the following holds:

\noindent
a) The mKP Lax operators
${\cal L}$ and $\bar {\cal L}$ corresponding to the solutions
$(\tau , \hat \tau )$, $(\tau , \bar \tau )$
determined by $\phi ({\bf t})$ and $\phi '({\bf t})$ respectively,
are connected by the formula
\beq\label{LbarL}
\bar {\cal L}=S^{-1}{\cal L} S.
\eeq

\noindent
b) The function
\beq\label{S}
S({\bf t})=\frac{\phi '({\bf t})}{\phi ({\bf t})}=
\frac{\bar \tau ({\bf t})}{\hat \tau ({\bf t})}
\eeq
is an mKP wave function in the sense 
of Definition \ref{definition:wavemkp}.
Conversely, any mKP wave function has the form (\ref{S}) 
with some KP wave functions $\phi$, $\phi '$;

\noindent
c) $S^*({\bf t})=S^{-1}({\bf t})$ is an adjoint mKP wave function
in the sense 
of Definition \ref{definition:wavemkp};
\end{theorem}

\noindent
Part a) easily follows from the definitions.
An outline of the proof of part b) is given below.
For the proof we need the following lemma.

\begin{lemma}\label{lemma:St2}
The function $S=S({\bf t})$ defined in (\ref{S}) satisfies the
linear equation
\beq\label{St2}
\p_{t_2}S  =\p_{x}^2 S + 2v_0 \p_x S, 
\qquad v_0=\p_x\log \frac{\hat \tau}{\tau},
\eeq
which is equation (\ref{wave1mkp}) for $k=2$.
\end{lemma}

\noindent
{\it Proof.}
In terms of the function $S$, equation (\ref{a3}) acquires
the form
$$
\p_{t_2}\log S -\p_x^2 \log S -(\p_x \log S)^2 -2v_0\p_x \log S=0,
$$
which is the same as (\ref{St2}).
\square

\noindent
{\it Proof of Theorem \ref{theorem:S}.}
The idea and the main steps of the proof are similar to
those of Theorem \ref{theorem:mkp1}, so we omit some details.

The starting point is equation (\ref{SS3d}). Substituting
$\p_{t_2}S$ from (\ref{St2}), we rewrite it as
\beq\label{St3}
\Bigl (\p_x \log \frac{\hat \tau_1 \tau}{\tau_1 \hat \tau}\Bigr )
e^{D(a_1)}\p_x S +\p_x \log \frac{\hat \tau_1}{\tau}
\Bigl (e^{D(a_1)}\! -\! 1\Bigr ) \p_x S +
e^{D(a_1)}\p_x^2 S -a_1\Bigl (e^{D(a_1)}\! -\! 1\Bigr )\p_x S=0.
\eeq
To transform the left-hand side, we use equation
(\ref{mkp5d}). Letting $a_2\to \infty$ in it, we obtain the
relation 
\beq\label{St4a}
\p_x \log \frac{\hat \tau_1}{\tau} =a_1 - 
a_1 \frac{\hat \tau \tau_1}{\hat \tau_1 \tau}
\eeq
to be substituted into (\ref{St3}). After some simple transformations,
equation (\ref{St3}) can be rewritten in the form
\beq\label{St4}
\p_x \left ( a\Bigl (1-e^{-D(a)}\Bigr )S -
\Bigl (1-a^{-1}\p_x \log \frac{\hat \tau}{\tau^{[a]}}\Bigr )^{-1}
\p_x S \right )=0,
\eeq
where $\tau^{[a]}=\tau ({\bf t}-[a^{-1}])$.
Therefore, we conclude that
$$
a\Bigl (1-e^{-D(a)}\Bigr )S -
\Bigl (1-a^{-1}\p_x \log \frac{\hat \tau}{\tau^{[a]}}\Bigr )^{-1}
\p_x S =f,
$$
where the function $f$ does not depend on $x=t_1$. This function
is actually equal to zero. This follows from the fact that
the same equation (but with zero right-hand side) is satisfied
by the function $S$ defined in (\ref{S}).
Indeed, it is not difficult to check that 
\beq\label{St4e}
\Bigl (a_1 +\p_x \log \frac{\tau^1}{\hat \tau} \Bigr )
\Bigl (\frac{\bar \tau}{\hat \tau}-\frac{\bar \tau^1}{\hat \tau^1}
\Bigr )-\p_x \Bigl (\frac{\bar \tau}{\hat \tau}\Bigr )=0
\eeq
holds. To do this, one should use equations (\ref{b1c}) and
(\ref{St4a}).

Therefore, we have proved that the 
function $S$ satisfies the linear equation
\beq\label{St6}
a\Bigl (1-e^{-D(a)}\Bigr )S({\bf t}) -
W({\bf t}, a)
\p_x S({\bf t}) =0,
\eeq
where
\beq\label{St7}
W({\bf t}, a)=
\Bigl (1-a^{-1}\p_x \log \frac{\hat \tau}{\tau^{[a]}}\Bigr )^{-1}=
1+ \sum_{k\geq 1}W_k ({\bf t}) a^{-k}, 
\qquad W_1({\bf t})=v_0({\bf t}).
\eeq
The functions $W_k ({\bf t})$ are expressed through logarithmic
derivatives of the tau-functions $\tau$ and $\hat \tau$.
The operator $e^{-D(a)}-1$ can be expanded using the 
elementary Schur polynomials as in (\ref{exp}):
$\displaystyle{
e^{-D(a)}-1=\sum_{k\geq 1}p_k(-[\p_x^{-1}] )a^{-k}.}
$
Then from equation (\ref{St6}) it follows that the function
$S$ satisfies linear equations of the form
\beq\label{St7a}
p_{k+1}(-[\p_x^{-1}]) S + W_k \p_x S=0, \qquad k\geq 1.
\eeq
To identify them with (\ref{wave1mkp}), we note that 
the mKP BA function $\varphi ({\bf t}, z)$ given by (\ref{mkp9})
satisfies the same equation (\ref{St6}), i.e.,
\beq\label{St5}
\varphi ({\bf t}, z)-\varphi ({\bf t}-[a^{-1}], z) -
\Bigl (\p_x \log \varphi ({\bf t}, a)\Bigr )^{-1} 
\p_x \varphi ({\bf t}, z)=0.
\eeq
This can be easily checked  by
substituting the expression (\ref{mkp9}) for $\varphi ({\bf t}, z)$ and
using equation (\ref{St4a}), then one can see that (\ref{St5}) is
equivalent to the bilinear equation for the KP tau-function
$\tau ({\bf t})$. Therefore, the BA function 
$\varphi$ satisfies (\ref{St7a}):
\beq\label{St7b}
p_{k+1}(-[\p_x^{-1}]) \varphi + W_k \p_x \varphi =0, \qquad k\geq 1.
\eeq
The structure of these equations allows one to
express the derivative $\p_{t_k}\varphi$ through derivatives
$\p_{t_l}\varphi$ with $l<k$:
$$
\p_{t_k}\varphi ={\sf A}_k(\p_{t_{k-1}}, \ldots , \p_{t_1})\varphi ,
$$
where ${\sf A}_k$ in the right-hand side is some differential
operator in the variables $t_1, \ldots , t_{k-1}$. Applying it
successively for $k=2,3, \ldots $, it is possible to represent
the right-hand side as a differential operator containing
$\p_{t_1}=\p_x$ only:
$$
\p_{t_k}\varphi ={\sf B}_k(\p_{x})\varphi .
$$
However, we already know from Corollary \ref{corollary:mkp1} that
for the BA function the operators 
${\sf B}_k(\p_{x})$ coincide with the 
operators ${\cal B}_k$ defined in (\ref{Bkmkp}) (see (\ref{mkp13})).
Since the function $S$ satisfies the same linear equations as
the BA function $\varphi$, we conclude that from (\ref{St7a})
it follows that
$\p_{t_k}S={\cal B}_k S$, i.e., $S$ defined 
in (\ref{S}) is indeed an mKP wave function.

It remains to prove the inverse statement: that 
any mKP wave function $S$ can be
represented as a ratio of two KP tau-functions $\hat \tau$ and
$\bar \tau$ obtained from some KP tau-function $\tau$ as results
of two forward BD transformations $\tau \to \hat \tau$ and
$\tau \to \bar \tau$. This is already almost evident from the
calculations done above. Indeed, we have checked that the BA
function satisfies equation (\ref{St5}). Integrating it 
over $z$ with some density, we see that the same equation
is satisfied by any wave function $S$ (i.e., equation 
(\ref{St6})). But equation (\ref{St4e}) means that
$S$ can be represented as ratio of $\hat \tau$ and $\bar \tau$.
\square

\begin{remark}
The transformation $\hat \tau \rightarrow \bar \tau$ is 
a BD transformation for KP hierarchy meaning that it sends any its solution
to another solution. However, such BD transformations
are in a sense of a different type
than the forward and backward ones
defined by equations (\ref{mkp6}) and (\ref{mkp6d}) respectively.
To be more precise, they can be realized as a composition of
a forward and backward transformations, and for this reason they
are called {\it binary BD transformations}. 
In this paper, we consider them from a slightly 
different angle than is usually accepted in the literature.
Namely, following a more traditional approach one constructs
a binary BD transformation starting from the pair $(\phi , \phi'^{*})$
of a KP wave function and an adjoint one. In this approach, the role
of the function $S$ (the function which satisfies the SchKP equation)
is played by the Baker-Akhiezer kernel doubly integrated with 
two density functions (for details see, e.g., the recent
review \cite{Z25a}). We, however, prefer to avoid dealing with 
adjoint wave functions and develop the theory
using the couple of wave functions $(\phi , \phi ')$.
\end{remark}

Now, let us prove an SchKP analogue of
Theorem \ref{theorem:inversetoKPtomKP}, i.e., the statement
that is in a sense inverse to the one of Theorem \ref{theorem:S}.

\begin{theorem}
Let $S({\bf t})$ be any solution to the SchKP equation
(\ref{SSc}). Then there exist two functions 
$\hat \tau ({\bf t})$, $\bar \tau ({\bf t})$ such that:
\begin{itemize}
\item[a)] The function $S$ is given by $S({\bf t})=
\bar \tau ({\bf t})/\hat \tau ({\bf t})$;
\item[b)] Both $\hat \tau ({\bf t})$ and 
$\bar \tau ({\bf t})$ are KP tau-functions satisfying the
bilinear equation (\ref{kp1}) and equation (\ref{bil1}) 
that connects them.
\end{itemize}
\end{theorem}

\noindent
{\it Proof.} The proof starts from rewriting equation (\ref{SSc})
in the form
$$
(S^1-S^{12})(S^2-S^{23})(S^3-S^{13})=
(S^3-S^{23})(S^1-S^{13})(S^2-S^{12}).
$$
Note that it is invariant under cyclic permutation of
indices $(123)$. Letting $a_{\gamma}\to \infty$ for $\gamma =1,2,3$,
we have:
\beq\label{RR1}
S_x^{\alpha} (S^{\beta} -S^{\alpha \beta})(S-S^{\beta})=
S_x^{\beta} (S^{\alpha} -S^{\alpha \beta})(S-S^{\alpha})
\eeq
for $(\alpha \beta )=(12), (23), (31)$, where 
$S_x^{\alpha}=\p_x S^{\alpha}$. In terms of the function
$$
R({\bf t}, z)=\frac{S({\bf t})-S({\bf t}-[z^{-1}])}{S_x({\bf t})}
$$
equation (\ref{RR1}) acquires the form
$$
\frac{R({\bf t}-[a_{\alpha}^{-1}], 
a_{\beta})}{R({\bf t}, a_{\beta})}=
\frac{R({\bf t}-[a_{\beta}^{-1}], 
a_{\alpha})}{R({\bf t}, a_{\alpha})},
$$
which allows us to apply Lemma \ref{lemma:technical2}. It then
follows that there exists a function $g({\bf t})$ such that
\beq\label{RR2}
R({\bf t}, a)=\frac{g({\bf t}-[a^{-1}])}{a \,g({\bf t})}.
\eeq
A direct calculation (which uses equation (\ref{SSc})) 
shows that the function
\beq\label{RR3}
\phi ({\bf t})=1/g({\bf t})
\eeq
satisfies equation (\ref{R3}), i.e., is a solution to the mKP hierarchy. 
Theorem \ref{theorem:inversetoKPtomKP} implies 
that there exists a pair of KP tau-functions $\tau , \hat \tau$
connected by the bilinear equation (\ref{mkp5}) such that
$\phi =\hat \tau / \tau$.

Note that if $S$ is a solution to (\ref{SSc}), then 
$S'=S^{-1}$ is another solution. Repeating the arguments above
for this solution, we conclude that the function $\phi ' = S\phi =
S/g$, as well as $\phi$, solves the mKP hierarchy. More generally,
$S'' =(aS+b)/(cS+d)$ for all $a,b,c,d \in \CC$ is a solution, too.
In the same way, this implies that any linear combination
$\phi''=c\phi' +d\phi$ solves the mKP hierarchy.
According to Proposition
\ref{proposition:phiphi}, $\phi$ and $\phi'$ (and $\phi''$) are
wave functions corresponding to one and the same solution to KP
with a tau-function $\tau$. Then $\phi =\hat \tau/\tau$, 
$\phi'=\bar \tau/\tau$, and $(\tau , \hat \tau )$, 
$(\tau , \bar \tau )$ are two solutions to the mKP hierarchy.
Equations (\ref{bil3}) and (\ref{SS1}) that connect the tau-functions
$\hat \tau , \, \bar \tau$ then follow from Theorem
\ref{theorem:hat-tau-bar-tau}.
\square

\begin{remark}
It can be shown (see, e.g. \cite{Z25a}) that the pair of tau-functions
$(\tau , \hat{\bar \tau})$ on the second diagonal in 
the square diagram (\ref{diag1}) satisfy the integral
bilinear relation
\beq\label{mkp2d}
\begin{array}{l}
\displaystyle{
\oint_{C_{\infty}} z^{2}e^{\xi ({\bf t}-{\bf t}', z)}
\hat{\bar \tau }({\bf t}-[z^{-1}])
\tau ({\bf t}'+[z^{-1}]) \, dz =0,
}
\end{array}
\eeq
and, therefore, can be regarded as a solution to the
2-mKP hierarchy in the sense of 
Section \ref{section:mKP-Toda} (see (\ref{mkp2})).
\end{remark}

\subsection{From mKP and SchKP to HD}
\label{section:mKPHD}

We have shown that the bilinear formulation of the SchKP hierarchy
requires two KP tau-functions $\hat \tau$ and $\bar \tau$ such that
they are connected by the bilinear equation (\ref{bil3}).
The dependent variable of the hierarchy is the $S$-function,
which is their ratio. The infinite 
hierarchy itself is obtained by expanding
equation (\ref{SSc}) in a Taylor series in negative powers
of $a_1, a_2, a_3$. 

A natural question arises whether the SchKP hierarchy admits
a formulation of the Lax-Sato type, as it is the case for
KP and mKP.
In this section we show that the answer is yes, and the Lax-Sato
form of the SchKP hierarchy is just what is called 
the (2+1) HD hierarchy. 

First of all, we need some technical preparations. 
Let $S(x)$ and $X(s)$ be two mutually inverse functions:
$S \circ X =X\circ S =\mbox{Id}$, or
\beq\label{tech0}
S(X(s))=s, \qquad X(S(x))=x.
\eeq
Differentiating these formulas, we get
$$
S_x \, X_s =1, \;\; \mbox{where} \;\; S_x\equiv \p_x S(x), \;
X_s\equiv \p_s X(s),
$$
where it is implied that $x$ and $s$ depend on each other
as $s=S(x)$ (or $x =X(s)$).
Given any function $f(x)$ of $x$, the corresponding function of
$s$ is $f(X(s))$. The notation $f^{(x)}(x)$ for
the $f(x)$ and $f^{(s)}(s)$ for $f(X(s))$ are sometimes useful,
then we have the identity $f^{(x)}(x)=f^{(s)}(s)$ which states that
the function $f$ is an invariant object defined regardless 
of the choice of coordinate ($x$ or $s$).
The vector fields $\p_x$ and $\p_s$ are connected by the formula
\beq\label{tech1}
\p_x = S_x \p_s, \quad S_x = S_x (X(s)). 
\eeq

Consider pseudo-differential operators in the variables $x$ and $s$:
\beq\label{AA}
A^{(x)}=\sum_{k} a_k^{(x)}(x) \p_x^{k}, \qquad
A^{(s)}=\sum_{k} a_k^{(s)}(s) (S_x \p_s )^{k},
\eeq
with $A^{(x)}=A^{(s)}$ by the definition.

\begin{lemma}
\label{lemma:A}
For a pseudo-differential operaror $A^{(x)}=A^{(s)}=A$ it holds:
\begin{itemize}
\item[a)] $A^{(s)}_{\geq 1}=A^{(x)}_{\geq 1}$,
\item[b)] $A_{\geq 2}^{(s)}=A^{(x)}_{\geq 1}-(A^{(x)}_{\geq 1}S)
S_{x}^{-1}\p_x$,\\
\\
where $(A^{(x)}_{\geq 1}S)$ means action of the differential
operator $A^{(x)}_{\geq 1}$ to the function $S(x)$.
\end{itemize}
\end{lemma}

\noindent
The both equalities easily follow from the definitions. 
Some details of the proof can be found in \cite{OR}.

The last technical remark is related to the case when 
the functions under consideration depend on a parameter $t$
(actually, we are interested in the case when the set
of parameters ${\sf t}=\{t_2, t_3, t_4, \ldots \}$ is infinite).
Then the rule of change of variables from $(x, t)$ to $(s, t)$
implies that
\beq\label{tech2}
\p_t \Bigr |_{x} =\p_t \Bigr |_{s} +S_t \p_s, \quad 
S_t \equiv \p_t S(x,t),
\eeq
where $\p_t \Bigr |_{x}$ ($\p_t \Bigr |_{s}$) is the partial
derivative at constant $x$ (respectively, $s$).

Now we are ready to reformulate the results of the previous
section using the framework of pseudo-differential operators.
Let $S({\bf t})=S(x, {\sf t})$ (${\sf t}=
\{t_2, t_3, \ldots \}$) be a wave function of the mKP
hierarchy given by (\ref{S}). In what follows it plays the
role of the $S$-function from (\ref{tech0}) depending on 
the $t_k$'s as on parameters. So, we consider the change
of variables $x\to s$, where $s=S(x, {\sf t})$ (with fixed 
${\sf t}$).
Consider the mKP Lax operator ${\cal L}={\cal L}^{(x)}$ 
(\ref{Laxmkp}) acting to functions of $x$ and its $s$-version,
as in (\ref{AA}): 
\beq\label{LL}
{\cal L}^{(x)}=\p_x +\sum_{k\geq 0}v_{k}^{(x)}(x)\p_x^{-k},
\qquad
{\sf L}^{(s)}=S_x \p_s +\sum_{k\geq 0}v_{k}^{(s)}(s)(S_x \p_s )^{-k},
\eeq

\begin{proposition}
Let $S(x, {\sf t})$ be a wave function for the mKP Lax operator
${\cal L}$ and the new independent space variable be
$s=S(x, {\sf t})$. Then the Lax equations (\ref{kp9mkp}) 
for ${\cal L}={\cal L}^{(x)}$ imply the following Lax equations for 
${\sf L}={\sf L}^{(s)}$ given by (\ref{LL}):
\beq\label{sfL}
\p_{t_k}{\sf L}=[({\sf L}^k)_{\geq 2}, \, {\sf L}], \qquad
k\geq 2
\eeq
(the $t_k$-derivative in the left-hand side is taken at constant $s$).
\end{proposition}

\noindent
{\it Proof.}
Using (\ref{tech2}), we can write:
$$
\p_{t_k}{\sf L}=\Bigl [\p_{t_k} \Bigr |_{s}, \, {\sf L}\Bigr ]=
\Bigl [\p_{t_k} \Bigr |_{x}, \, {\cal L}\Bigr ]-
\Bigl [ S_{t_k}\p_s , \, {\sf L}\Bigr ]=
\p_{t_k}{\cal L} -\Bigl [ S_{t_k}\p_s , \, {\sf L}\Bigr ].
$$
Next, using Lemma \ref{lemma:A}, we have:
$$ 
\p_{t_k}{\sf L}-[({\sf L}^k)_{\geq 2}, \, {\sf L}]=
\p_{t_k}{\cal L} -\Bigl [ S_{t_k}\p_s , \, {\sf L}\Bigr ]
-\Bigl [({\cal L}^k )_{\geq 1} - (({\cal L}^k )_{\geq 1}S)\p_s ,\,
{\sf L}\Bigr ]
$$
$$
=\p_{t_k}{\cal L}-\Bigl [({\cal L}^k)_{\geq 1}  , \, {\sf L}\Bigr ]
+\Bigl [(({\cal L}^k)_{\geq 1}S -S_{t_k})\p_s , \, {\sf L}\Bigr ].
$$
If $S$ is a wave function for ${\cal L}$, i.e., 
$S_{t_k}=({\cal L}^k)_{\geq 1}S$, then the last term in the right-hand side
vanishes. The rest terms also vanish by virtue of the Lax equation
for ${\cal L}$, and we obtain zero in the right-hand side. 
Therefore, the Lax equation (\ref{sfL}) is proved.
\square

Comparing ${\sf L}$ with $L^{\rm HD}$ from (\ref{LLL}), 
we see that to identify them one should put
\beq\label{US}
U(s, {\sf t})=S_x (X(s), {\sf t}).
\eeq
In terms of the function $U$ we have:
\beq\label{US1}
\p_x = U \p_s.
\eeq
So, ${\sf L}$ is the Lax operator for the HD hierarchy. Its
dependent variable $U$ is directly connected with $S$, the one subject
to equations of the SchKP hierarchy. Namely, 
take the $x$-derivative $S_x(x, {\sf t}) =
\p_x S(x, {\sf t})$ and consider it 
as a function of $s$ (and of all the times) via the substitution
$x=X(s, {\sf t})$, where $X(s, {\sf t})$ is the function inverse
to $S(x, {\sf t})$, as equation (\ref{US}) suggests.

The dressing operator ${\sf W}$ for the HD hierarchy is 
a pseudo-differential operator in the variable $s$ 
of the form
\beq\label{WW1}
{\sf W}=w_0 +w_1 (U\p_s)^{-1} + w_2 (U\p_s)^{-2} + \ldots 
\eeq
such that
\beq\label{WW}
{\sf L}={\sf W}\, U \p_s {\sf W}^{-1}=
U\p_s + v_0 + v_1 (U\p_s )^{-1} + \ldots \, .
\eeq
In fact it is the mKP dressing operator ${\cal W}$ in which
the change of variables $x\to s$ is done.

By analogy with Definition \ref{definition:wavemkp} it is natural
to introduce wave functions for the HD hierarchy.

\begin{definition}
Any solution $\Phi (s, {\sf t})$ to the system of linear
differential equations
\beq\label{waveHD}
\p_{t_k}\Phi (s, {\sf t}) ={\sf B}_k \Phi (s, {\sf t}), 
\qquad {\sf B}_k=({\sf L}^k)_{\geq 2},
\quad k\geq 2
\eeq
is called wave function for the HD Lax operator ${\sf L}$ (\ref{WW}).
\end{definition}

\noindent
Wave functions for the HD hierarchy admit a simple
description in terms of those for the mKP hierarchy.

\begin{proposition}
\label{proposition:waveHD}
Let $S(x, {\sf t})$ and $\tilde S(x, {\sf t})$ 
be two wave functions of the mKP hierarchy.
Then $\Phi (s, {\sf t})=\tilde S(X(s, {\sf t}), {\sf t})$
is a wave function of the HD hierarchy associated with $S$
(as before, $X(s, {\sf t})$ is the function inverse 
to $S(x, {\sf t})$).
Moreover the opposite is also true: 
if $\Phi(s,{\sf t})$ is a wave function of the HD hierarchy 
associated with $S$, then there exists 
a wave function $\tilde{S}(s,{\sf t})$ of the mKP hierarchy
such that  $\Phi (s, {\sf t})=\tilde S(X(s, {\sf t}), {\sf t})$.
\end{proposition}

\noindent
{\it Proof.}
According to the definitions and Lemma \ref{lemma:A} we can write
the following chain of equalities:
\beq\label{chain}
\begin{array}{lll}
\p_{t_k}\Phi (s, {\sf t})-({\sf L}^k)_{\geq 2}\Phi (s, {\sf t})
&= &\p_{t_k}\tilde S(x, {\sf t}) +\p_{t_k}X(s, {\sf t})
\tilde S_x (x, {\sf t}) -
({\sf L}^k)_{\geq 2}\tilde S(X(s, {\sf t}), {\sf t})
\\ && \\
&=& ({\cal L}^k)_{\geq 1}\tilde S -({\cal L}^k)_{\geq 1}\tilde S
+({\cal L}^k)_{\geq 1}S) S_x^{-1} \tilde S_x +
\p_{t_k} X(s, {\sf t}) \tilde S_x
\\ && \\
&=& \Bigl (\p_{t_k} X(s, {\sf t}) + \p_{t_k}S(x, {\sf t})S_x^{-1}
(x, {\sf t})
\Bigr )\tilde S_x
\end{array}
\eeq
Taking the $t_k$-derivative of the identity
$X(S(x, {\sf t}), {\sf t})=x$, we get:
$$
\p_{t_k} X(s, {\sf t}) + \p_{t_k}S(x, {\sf t})\p_s X(s, {\sf t})=0
$$
But 
$\p_s X(s, {\sf t})=S_x^{-1}(x, {\sf t})$ as they are derivatives
of the mutually inverse functions. Therefore, the expression 
in the last line in
(\ref{chain}) vanishes and we have
$$
\p_{t_k}\Phi (s, {\sf t})-({\sf L}^k)_{\geq 2}\Phi (s, {\sf t})=0,
$$
which just means that 
$\Phi (s, {\sf t})=\tilde S(X(s, {\sf t}), {\sf t})$
is a wave function of the HD hierarchy
if and only if $\tilde{S}$ is a wave function of the mKP hierarchy.
\square

\begin{remark}
Since the differential operators ${\sf B}_k$ are divisible by
$\p_s^2$ from the right, 
it is evident that the simplest HD wave functions are 
$\Phi =1$ and $\Phi =s$. The corresponding second
mKP wave functions $\tilde S$ from Proposition 
\ref{proposition:waveHD} in these two cases are
$\tilde S =1$ and $\tilde S =S$ respectively.
\end{remark}

\subsection{Wave functions and BD transformations}

\begin{figure}[t]
\centering{\includegraphics[scale=1.2]{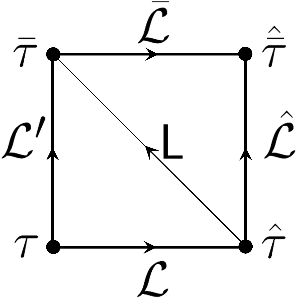}}
\vspace{1cm}
\caption{The square}
\label{figure:square}
\end{figure}

For both KP and mKP hierarchies, their wave functions allow 
one to construct other solutions by means of BD transformations. 
In the KP case, this is part a) of Theorem \ref{theorem:mkp1},
and there is only one such solution corresponding to 
a given wave function.
In the mKP case, given an mKP wave function $S$,
there are three different solutions described by 
Theorem \ref{theorem:WWW}.
All this is summarized by the square 
in Fig.\ref{figure:square}, which is
a pictorial version of diagrams (\ref{diag1}) or (\ref{diag2}).
In the figure, the vertices correspond to KP tau-functions, 
the directed edges correspond to mKP Lax operators 
and the face corresponds to the Lax operator ${\sf L}$ 
of the HD hierarchy.

A similar pattern holds for HD wave functions.
\begin{proposition}
\label{proposition:mKPBack}
Let $\mathcal{L}$ be an mKP Lax operator and $S, F$ be its 
two wave functions, $\mathcal{L}'=S^{-1}\mathcal{L}S$, $\hat{\mathcal{L}}=S^{-1}_x\partial_x\mathcal{L}\partial^{-1}S_{x}$, 
$\bar{\mathcal{L}}=S^{-1}_xS^2\partial_xS^{-1}\mathcal{L}S\partial_x^{-1}S^{-2}S_{x}$, then 
\begin{equation}\label{mKPBack}
\begin{array}{c}
    F'=FS^{-1}, \qquad
    \hat{F}=F_xS_{x}^{-1}, \qquad
    \bar{F}=S^2(F/S)_xS_x^{-1}
\end{array}
\end{equation}
are wave functions of ${\cal L}'$, $\hat {\cal L}$ and
$\bar {\cal L}$ respectively.
\end{proposition}

\noindent
{\it Proof.}
First of all, we note that all the three cases can 
be considered at once since $F'=V(F)$, where 
$V$ is a pseudo-differential operator such that 
$\mathcal{L}'=V\mathcal{L}V^{-1}$. The same holds 
for $\hat{F}$ and $\bar{F}$. 
Moreover, for $V\mathcal{L}V^{-1}$ to be an mKP Lax operator, the 
relation 
$\partial_{t_k}V=((V\mathcal{L}V^{-1})^k )_{>1}V-
V(\mathcal{L}^k)_{>1}$ must hold. 
Calculating $\partial_{t_k}(V(F))$, one gets
\begin{equation}
\partial_{t_k}(V(F))=(((V\mathcal{L}V^{-1})^k )_{>1}V-V(\mathcal{L}^k)_{>1}-V(\mathcal{L}^k)_{>1})(F)=((V\mathcal{L}V^{-1})^{k})_{>1}(V(F)),
\end{equation}
which completes the proof.
\square

\begin{figure}[t]
\centering{\includegraphics[scale=1.2]{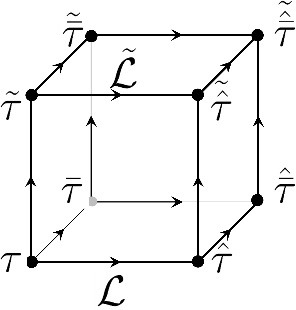}}
\vspace{1cm}
\caption{The cube}
\label{figure:cube}
\end{figure}

To depict the HD Lax operator $\sf L$ with its wave function $\Phi$, 
we need to add an extra dimension to Fig.\ref{figure:square}, see 
the cube in Fig.\ref{figure:cube}. 
Here, too, each vertex corresponds to a solution of the 
KP hierarchy, and each edge equipped with arrow symbolizes
a solution to the mKP hierarchy.
The bottom face corresponds to ${\sf L}={\sf L}(\hat{\tau},\bar{\tau})$;  
the diagonal from $\hat \tau$ to $\tilde \tau$ 
on the front face corresponds to $F$, 
which is another wave function for $\mathcal{L}$, 
connected to $\Phi$ as in Proposition \ref{proposition:waveHD}.
Using Proposition \ref{proposition:mKPBack} we define 
$F'$, $\hat{F}$, $\bar{F}$, from which the HD Lax operators 
corresponding to the walls of the cube can be defined.

Let us define the Lax operator corresponding to the roof 
of the cube. Firstly, we obtain the mKP Lax operator 
corresponding to the upper edge of the front wall: 
$\tilde{\mathcal{L}}=F^{-1}_xF^2\partial_xF^{-1}
\mathcal{L}F\partial_x^{-1}F^{-2}F_{x}$, then, 
according to Proposition \ref{proposition:mKPBack}, 
$F^2(S/F)_xF_x^{-1}$ is its wave function, which allows us
to construct the desired Lax operator.
The results are summarized in the following list
of equations:
\beq\label{list}
\begin{array}{l}
    {\sf L}_{\rm face}={\mathcal L}^{(f)}, 
    \\ \\
    {\sf L}_{\rm left}=
    (S^{-1}{\mathcal L}S)^{(r_1)},
    \\ \\
    {\sf L}_{\rm right}=((S_x)^{-1}\partial_x
    \mathcal{L}\partial_x^{-1}S_x))^{(r_2)},
    \\ \\
    {\sf L}_{\rm back}=[((S^{-1})_x)^{-1}\partial_x 
    S^{-1}\mathcal{L}S\partial_x^{-1}(S^{-1})_{x}]^{(r_3)},
  \\ \\
  {\sf L}_{\rm top}=[((F^{-1})_x)^{-1}\partial_x 
  F^{-1}\mathcal{L}F\partial_x^{-1}(F^{-1})_{x}]^{(r_4)}.
\end{array}
\eeq
Here the notation 
$\mathcal{A}^{(r)}$ means the operator ${\cal A}$ 
acting on functions of the variable $r$, which in the first line
is $f$, in the second line is $r_1$, etc. These variables are defined
as follows: $f={\rm Inv}(F)$, $r_1={\rm Inv}(F/S)$, $r_2={\rm Inv}(F_x/S_x)$, $r_3={\rm Inv}((F/S)_x/(S^{-1})_x)$, 
$r_4={\rm Inv}((S/F)_x/(F^{-1})_x)$, 
where ${\rm Inv}$ means the inverse function:
$H({\rm Inv}(H(x)))=x$.

Moreover, there are expressions for the tau-functions associated 
with vertices of the cube in terms of $\hat{\tau}, \tau,S,F$:
\beq
\begin{array}{lll}
\displaystyle{
    \hat{\bar{\tau}}=\frac{S_x\hat{\tau}}{\tau}, }&\bar{\tau}=S\hat{\tau}, 
    & \displaystyle{\tilde{\bar{\tau}}=(F/S)_x\frac{S^2\hat{\tau}^2}{\tau}},
    \\ && \\
   \displaystyle{\tilde{\hat{\tau}}=\frac{F_x\hat{\tau}}{\tau}}, & 
   \tilde{\tau}=F\hat{\tau}, & \displaystyle{\tilde{\hat{\bar{\tau}}}=
   \left(\frac{F_x}{S_x}\right)_x(S^{-1})_x\frac{\hat{\tau}}{\tau^2}}.
\end{array}
\eeq

\begin{remark}
    As it was already mentioned, there are 
    two distinct choices of HD wave functions $\Phi=s$ and $\Phi=c$, 
    and their linear combinations lead to the $PSL(2)$-action on $S$ 
    and consequently on ${\sf L}$. The $S\rightarrow S+c$ 
    with constant $c$ does not affect ${\sf L}$, however, 
    $S\rightarrow S^{-1}$ leads to the transformation ${\sf L}^{(s)}\rightarrow (s^{-1}{\sf L}s)^{(1/s)}$. This can be also interpreted as the connection of the HD Lax operator to the side of the face, 
    and since each face is two-sided, this transformation corresponds to a flip.
\end{remark}

\begin{remark}
    One can try to construct other solutions of the HD hierarchy 
    (for example, using $\tilde{\mathcal{L}}$) to define Lax operators corresponding to the walls of the cube. However, the result will be the same up to a flip (see the previous remark). This is due 
    to the commutativity of the BD transformations for 
    the KP hierarchy.
\end{remark}

\section{Embedding into multi-component KP hierarchy}

In this section we will show that the SchKP hierarchy
defined by the bilinear equation (\ref{bil3}) admits an
embedding into the multi-component KP hierarchy
\cite{DJKM81a,KL93,TT07,Teo11}. In this hierarchy,
independent variables are $n$ infinite sets of (in general complex) 
``times''
$$
{\bf t}=\{{\bf t}_1, {\bf t}_2, \ldots , {\bf t}_n\}, \qquad
{\bf t}_{\alpha}=\{t_{\alpha , 1}, t_{\alpha , 2}, 
t_{\alpha , 3}, \ldots \, \},
\qquad \alpha = 1, \ldots , n
$$
and $n$ additional variables $s_1, \ldots , s_n\in \ZZ$
such that 
\beq\label{s1}
\sum_{\alpha =1}^n s_{\alpha}=0.
\eeq
By ${\bf s}$ we denote the vector
${\bf s}=\{s_1, \ldots , s_n\}$ and by ${\bf e}_{\alpha}$ the vector
whose $\alpha$'s component is 1 and all other components are 0.
We will also use the following standard
notation:
\beq\label{st1}
\left ({\bf t}\pm [z^{-1}]_{\gamma}\right )_{\alpha j}=t_{\alpha , j}\pm
\delta_{\alpha \gamma} \frac{z^{-j}}{j}, \qquad
\xi ({\bf t}_{\alpha}, z)=\sum_{j\geq 1}t_{\alpha , j}z^j.
\eeq
In the bilinear formalism, 
the dependent variable is the tau-function 
$\tau ({\bf s}, {\bf t})$. 
The hierarchy is defined as an infinite set of bilinear equations
for the tau-function which are encoded in the basic bilinear
relation \cite{DJKM81a,Teo11}
\beq\label{s3}
\begin{array}{l}
\displaystyle{
\sum_{\mu =1}^n \epsilon_{\alpha \mu}({\bf s})
\epsilon_{\beta \mu}({\bf s}')
\oint_{C_{\infty}}\! dz \, 
z^{s_{\mu}-s_{\mu}'+\delta_{\alpha \mu}+\delta_{\beta \mu}-2}
e^{\xi ({\bf t}_{\mu}-{\bf t}_{\mu}', \, z)}}
\\ \\
\displaystyle{\phantom{aaaaaaaaaa}
\times \tau \left ({\bf s}+{\bf e}_{\alpha}-{\bf e}_{\mu}, 
{\bf t}-[z^{-1}]_{\mu}\right )
\tau \left ({\bf s}'+{\bf e}_{\mu}-{\bf e}_{\beta}, 
{\bf t}'+[z^{-1}]_{\mu}\right )=0}
\end{array}
\eeq 
valid for all ${\bf t}$, ${\bf t}'$, ${\bf s}$, ${\bf s}'$
subject to the constraint (\ref{s1})). 
In (\ref{s3}) $epsilon_{\alpha \mu}({\bf s})$ is a sign
factor defined as
\beq\label{s3a}
\epsilon_{\alpha \mu}({\bf s})=\left \{
\begin{array}{cl} 
\hspace{0.3cm}(-1)^{s_{\alpha +1}+\ldots +s_{\mu}}, 
&\quad \alpha < \mu
\\ 
\hspace{-1.5cm}1, &\quad \alpha =\mu
\\ 
-(-1)^{s_{\mu +1}+\ldots +s_{\alpha}}
&\quad \alpha > \mu .
\end{array}
\right.
\eeq
Different bilinear relations 
for the tau-function of the Hirota-Miwa type which follow from
(\ref{s3}) for special choices of ${\bf s}-{\bf s}'$ and ${\bf t}-{\bf t}'$
are given in \cite{Teo11}.

The following well known statement directly follows from the
general equation (\ref{s3}).

\begin{proposition}
\label{proposition:1comp}
The tau-function $\tau ({\bf s}, {\bf t})$, regarded as a function
of ${\bf t}_{\alpha}$, with any ${\bf s}$ and ${\bf t}_{\mu}$ with
$\mu \neq \alpha$ fixed, is a tau-function of the KP hierarchy, i.e.,
it satisfies the bilinear equation (\ref{kp1}).
\end{proposition}
\square

Let us fix three distinct indices 
$\alpha , \beta , \gamma \in \{1, 2, \ldots , n\}$ such that
$\alpha <\beta <\gamma$ and
consider the case when $s_{\mu}=s'_{\mu}$ unless
$\mu =\alpha , \beta , \gamma$ and 
${\bf t}_{\mu}={\bf t'}_{\mu}$ unless $\mu =\alpha$
(in the latter case we put ${\bf t}_{\alpha}={\bf t}$,
${\bf t'}_{\alpha}={\bf t'}$ for simplicity of the notation).
It is easy to see that in this case the sum in (\ref{s3})
contains only three non-zero terms:
\beq\label{s3b}
\begin{array}{c}
\displaystyle{
\oint_{C_{\infty}}\! dz \, 
z^{s_{\alpha}-s_{\alpha}'}
e^{\xi ({\bf t}-{\bf t'}, \, z)}
\tau \left ({\bf s}, 
{\bf t}-[z^{-1}]\right )
\tau \left ({\bf s}', 
{\bf t}'+[z^{-1}]\right )}
\\ \\
\displaystyle{
+\, \epsilon_{\alpha \beta}({\bf s})\epsilon_{\alpha \beta}({\bf s'})
\oint_{C_{\infty}}\! dz \, z^{s_{\beta}-s'_{\beta}-2}
\tau \left ({\bf s}+{\bf e}_{\alpha}-{\bf e}_{\beta}, 
{\bf t}-[z^{-1}]_{\beta}\right )
\tau \left ({\bf s'}+{\bf e}_{\beta}-{\bf e}_{\beta}, 
{\bf t}'+[z^{-1}]_{\beta}\right )}
\\ \\
\displaystyle{
+\, \epsilon_{\alpha \gamma}({\bf s})\epsilon_{\alpha \gamma}({\bf s'})
\oint_{C_{\infty}}\! dz \, z^{s_{\gamma}-s'_{\gamma}-2}
\tau \left ({\bf s}\! +\! {\bf e}_{\alpha}\! -\! {\bf e}_{\gamma}, 
{\bf t}\! -\! [z^{-1}]_{\gamma}\right )
\tau \left ({\bf s'}\! +\! {\bf e}_{\beta}\! -\! {\bf e}_{\gamma}, 
{\bf t}'\! +\! [z^{-1}]_{\gamma}\right )}=0.
\end{array}
\eeq
Putting here 
${\bf s}={\bf e}_{\beta}-{\bf e}_{\gamma}$, ${\bf s'}={\bf 0}$
or
${\bf s}={\bf 0}$, ${\bf s'}={\bf e}_{\beta}-{\bf e}_{\gamma}$,
we see that the last two integrals are given by residues at infinity.
This leads to the following two equations:
\beq\label{s4}
\begin{array}{l}
\displaystyle{
\frac{1}{2\pi i}\oint_{C_{\infty}}
e^{\xi ({\bf t}-{\bf t'}, \, z)}
\tau \left ({\bf e}_{\beta}-{\bf e}_{\gamma}, 
{\bf t}-[z^{-1}]\right )
\tau \left ({\bf 0}, 
{\bf t}'+[z^{-1}]\right )\, dz}
\\ \\
\phantom{aaaaaaaaaaaaaa}
+\displaystyle{\epsilon_{\alpha \beta}({\bf e}_{\beta}-{\bf e}_{\gamma})
\epsilon_{\alpha \beta}({\bf 0})
\tau \left ({\bf e}_{\alpha}-{\bf e}_{\gamma}, {\bf t}\right )
\tau \left ({\bf e}_{\beta}-{\bf e}_{\alpha}, {\bf t}\right )=0},
\end{array}
\eeq

\beq\label{s4a}
\begin{array}{l}
\displaystyle{
\frac{1}{2\pi i}\oint_{C_{\infty}}
e^{\xi ({\bf t}-{\bf t'}, \, z)}
\tau \left ({\bf 0}, 
{\bf t}-[z^{-1}]\right )
\tau \left ({\bf e}_{\beta}-{\bf e}_{\gamma}, 
{\bf t}'+[z^{-1}]\right )\, dz}
\\ \\
\phantom{aaaaaaaaaaaaaa}
+\displaystyle{\epsilon_{\alpha \gamma}({\bf e}_{\beta}-{\bf e}_{\gamma})
\epsilon_{\alpha \gamma}({\bf 0})
\tau \left ({\bf e}_{\alpha}-{\bf e}_{\gamma}, {\bf t}\right )
\tau \left ({\bf e}_{\beta}-{\bf e}_{\alpha}, {\bf t}\right )=0}.
\end{array}
\eeq
It is easy to see that
$$
\epsilon_{\alpha \beta}({\bf e}_{\beta}-{\bf e}_{\gamma})
\epsilon_{\alpha \beta}({\bf 0})+
\epsilon_{\alpha \gamma}({\bf e}_{\beta}-{\bf e}_{\gamma})
\epsilon_{\alpha \gamma}({\bf 0})=0.
$$
Therefore, summing the two equations (\ref{s4}), (\ref{s4a}),
we obtain the equation
\beq\label{s5}
\begin{array}{l}
\displaystyle{
\oint_{C_{\infty}}
e^{\xi ({\bf t}-{\bf t'}, \, z)}\Bigl [
\tau \left ({\bf 0}, 
{\bf t}-[z^{-1}]\right )
\tau \left ({\bf e}_{\beta}\! -\! {\bf e}_{\gamma}, 
{\bf t}'+[z^{-1}]\right )}
\\ \\
\phantom{aaaaaaaaaaaaa}\displaystyle{
 +\, 
\tau \left ({\bf e}_{\beta}\! -\! {\bf e}_{\gamma}, 
{\bf t}-[z^{-1}]\right )
\tau \left ({\bf 0}, 
{\bf t'}+[z^{-1}]\right )\Bigr ]dz =0}
\end{array}
\eeq
which coincides with (\ref{bil3}) after the identification
\beq\label{s6}
\tau \left ({\bf 0}, 
{\bf t}\right )=\hat \tau ({\bf t}), \qquad
\tau \left ({\bf e}_{\beta}\! -\! {\bf e}_{\gamma}, 
{\bf t}\right )=\bar \tau ({\bf t}).
\eeq
It also follows from Proposition \ref{proposition:1comp} that 
both $\tau \left ({\bf 0}, 
{\bf t}\right )$ and $\tau \left ({\bf e}_{\beta}\! -\! {\bf e}_{\gamma}, 
{\bf t}\right )$ are KP tau-functions.

Summarizing, in this section we have proved the following theorem.

\begin{theorem}
The SchKP hierarchy is
embedded into the $n$-component KP hierarchy with $n\geq 3$
in the following way. The independent variables corresponding to 
all but three components are to be frozen and put equal to zero 
(for definiteness, let the three unfrozen components be numbered
by indexes 1, 2, 3 with the tau-function
$\tau ({\bf s}, {\bf t})=
\tau (s_1, s_2, s_3; {\bf t}_1,  {\bf t}_2,  {\bf t}_3)$). 
Then
$$
\hat \tau ({\bf t})=\tau (0, 0, 0; {\bf t},  {\bf 0},  {\bf 0}), 
\quad
\bar \tau ({\bf t})=
\tau (0, 1, -1; {\bf t},  {\bf 0},  {\bf 0})
$$
is a pair of tau-functions for the SchKP hierarchy, i.e., they
satisfy equation (\ref{bil3}).
\end{theorem}
\square

\section{Concluding remarks}

The main result of this paper is the established 
equivalence between the
Harry Dym (HD) hierarchy of partial differential equations and the 
bilinear integral equation (\ref{bil3}) for pairs $(\hat \tau ,
\bar \tau )$ of functions such that each of them is a tau-function
of the KP hierarchy (i.e. satisfies (\ref{kp1})). More precisely,
equation (\ref{bil3}) defines the Schwarzian KP
(SchKP) hierarchy within the framework of the 
bilinear formalism while the HD hierarchy
arises as its equivalent reformulation in terms of the Lax-Sato
approach. In contrast to what happens for the KP and mKP hierarchies,
such reformulation requires switching to a new space variable
$s=S(x)$, with $S$ being a wave function for the Lax operator 
of the mKP hierarchy,
so $\p_x =S_x \p_s$, and $U(s)$ from (\ref{HD2}) 
and (\ref{LLL}) (the dependent variable in the HD equation) is defined
as $U=S_x$.

Another result is an embedding of the SchKP hierarchy into 
the multi-component KP hierarchy, which, to the best of our knowledge,
was not discussed in the literature before.
Also, some known results related to B\"acklund-Darboux transformations
of integrable hierarchies have been reformulated in this paper 
in terms of tau-functions. 

An interesting open question is whether the bilinear equation
(\ref{bil1}) (which was obtained as a particular corollary of the 
more general integral equation (\ref{bil3}) is actually 
equivalent to the whole hierarchy, like this is the case for the
KP and mKP hierarchies. The problem is that the direct method 
of proving this equivalence based on some determinant or pfaffian
identities (suggested in \cite{Shigyo13} for KP, 
mKP and BKP) seems to be not applicable to (\ref{bil1})
because of a more complicated structure of the equation. So, one should
probably try to apply some indirect methods (similar to those
developed in \cite{TT95}). Another open question is whether
any solution to SchKP can be realized as a solution to the
multi-component KP hierarchy.

At last, we note that the connection with B\"acklund-Darboux
transformations and integrable discretizations of soliton 
equations deserves further study. In particular, this seems
to be an especially promising direction of further 
investigation, taking into 
account that BD transformations of integrable hierarchies play
an important role in applications to quantum integrable
model solvable by Bethe ansatz (see, e.g. the recent review
\cite{Z25}).

\section*{Acknowledgments}

\addcontentsline{toc}{section}{Acknowledgments}

We are grateful to A. Orlov and T. Takebe for 
illuminating discussions.
The work of V.P. (Sections 3.2, 4.2, 5)
and A.Z (Sections 2, 3.1, 4.1) 
was implemented in the framework 
of the Basic Research Program at HSE University (HSE-BR-2025-84).
%This work is an output of the research project 
%``Symmetry. Information. Chaos''
%implemented as a part of the Basic Research Program at 
%National Research University Higher School 
%of Economics (HSE University).


\begin{thebibliography}{99}

\addcontentsline{toc}{section}{References}

\bibitem{KP}
B. Kadomtsev, V. Petviashvili, {\it On the stability of solitary waves in weakly dispersive media}, Sov. Phys. Dokl. {\bf 15} (1970) 539–-541.

\bibitem{book1} V. Zakharov, S. Manakov, S. Novikov, L. Pitaevski,
{\it Theory of solitons. The inverse problem method}, Nauka, Moscow,
1980, Plenum (1984).

\bibitem{DJKM83} E. Date, M. Jimbo, M. Kashiwara and T. Miwa,
{\it Transformation groups for soliton equations: 
Nonlinear integrable systems --
classical theory and quantum theory} (Kyoto, 1981), 
Singapore: World Scientific,
1983, 39--119.

\bibitem{JM83} M. Jimbo and T. Miwa, {\it Solitons and 
infinite dimensional Lie
algebras}, Publ. Res. Inst. Math. Sci. Kyoto {\bf 19} (1983) 943--1001.

\bibitem{Hirota-book}
R. Hirota, {\it Direct methods in soliton theory},
Cambridge Tracts in Mathematics, vol. 155, 2004.


\bibitem{HarnadBalogh} 
J. Harnad and F. Balogh, {\it Tau functions and their applications},
Cambridge Monographs on Mathematical Physics, Cambridge University
Press, 2021.

\bibitem{K85} B. Kupershmidt, {\it Mathematics of dispersive 
water waves}, Comm. Math. Phys. {\bf 99} (1985) 51--73.

\bibitem{K89} B. Kupershmidt, {\it On the integrability of
modified Lax equations}, J. Phys. A: Math. Gen. {\bf 22} (1989)
L993--L998.

\bibitem{Kiso}
K. Kiso, {\it A remark on the commuting flows defined by Lax
equations}, Prog. Theor. Phys. {\bf 83} (1990) 1108--1125.

\bibitem{KO93}
B. Konopelchenko, W. Oevel, {\it An $r$-matrix approach to 
nonstandard classes of integrable equations}, Publ. RIMS,
Kyoto Univ. {\bf 29} (1993) 581--666.

\bibitem{OR} W. Oevel, C. Rogers, {\it Gauge transformations
and reciprocal links in 2+1 dimensions}, Rev. Math. Phys.
{\bf 05} (1993) 299--330.

\bibitem{Dickey99}
L. Dickey, {\it Modified KP and discrete KP}, Lett. Math. Phys.
{\bf 48} (1999) 277--289.

\bibitem{TakTeo06}
T. Takebe, L. Teo, {\it Coupled modified KP hierarchy 
and its dispersionless limit}, SIGMA {\bf 2} (2006) 072.

\bibitem{UT84}
K. Ueno and K. Takasaki, {\it Toda lattice hierarchy},
Adv. Studies in Pure Math. {\bf 4} (1984) 1--95.



\bibitem{Weiss83}
J. Weiss, {\it The Painlev\'e property 
for partial differential equations: II. B\"acklund 
transformations, Lax pairs
and Schwarzian derivative}, J. Math. Phys. {\bf 24} (1983) 1405–-1413.

\bibitem{BK98}
L. Bogdanov, B. Konopelchenko, 
{\it Analytic-bilinear approach to 
integrable hierarchies: 
I.Generalized KP hierarchy}, J. Math. Phys. 
{\bf 39} (1998) 4683–-4700.

\bibitem{BK99}
L. Bogdanov, B. Konopelchenko, 
{\it M\"obius invariant integrable lattice 
equations associated with KP
and 2DTL hierarchies}, Phys. Lett. A {\bf 256} (1999) 39–-46.

\bibitem{KS02}
B. Konopelchenko, W. Schief, {\it Menelaus' theorem, Clifford
configurations and inversive geometry of the Schwarzian KP 
hierarchy}, J. Phys. A: Math. Gen. {\bf 35} (2002) 6125--6144.

\bibitem{Schief03}
W. Schief, {\it Lattice geometry of the discrete Darboux, KP,
BKP and CKP equations. Menelaus’ and
Carnot’s theorems}, 
J. Nonlinear Math. Phys. {\bf 10} (Supplement 2) (2003) 
194-–208.

\bibitem{Kr75}
M. Kruskal, {\it Nonlinear wave equations}, in: Dynamical systems,
theory and applications: Battelle Seattle 1974 Recontres, ed. J. Moser,
Berlin, Heidelberg: Springer, 1975 (pp. 310--354).

\bibitem{Kad90}
L. Kadanoff, {\it Exact solutions for the Saffman-Taylor 
problem with surface tension}, Phys. Rev. Lett. {\bf 65} (1990)
2986--2988.


\bibitem{KD84}
B. Konopelchenko, V. Dubrovsky, {\it Some new integrable 
nonlinear evolution equations in 2+1 dimensions}, Phys. Lett. A
{\bf 102} (1984) 15--17.







\bibitem{AZ12}
A. Alexandrov, A. Zabrodin, {\it Free fermions and tau-functions},
Journal of Geometry and Physics {\bf 67} (2013) 37--80.

\bibitem{book}
V.B. Matveev, M.A. Salle, 
{\it Darboux transformations and solitons},
Berlin: Springer, 1991. 

\bibitem{RS02}
C. Rogers, W. Schieff, {\it B\"acklund and Darboux transformations},
Cambridge Texts in Applied Mathematics {\bf 30}, Cambridge
University Press, 2002.

\bibitem{Nimmo97}
J.J.C. Nimmo,
{\it Darboux transformations and the discrete KP equation},
J. Math. Phys. A: Math. Gen. {\bf 30} (1997) 8693–-8704.

\bibitem{NW97}
J.J.C. Nimmo, R. Willox,
{\it Darboux transformations for the two-dimensional Toda system},
Proc. R. Soc. Lond. A {\bf 453} (1997) 2497--2525.

\bibitem{HSS84a}
J. Harnad, Y. Saint-Aubin, S. Shnider,
{\it B\"acklund transformations for nonlinear sigma models
with values in Riemannian symmetric spaces},
Commun. Math. Phys. {\bf 92} (1984) 329--367.

\bibitem{ANP98}
H. Aratyn, E. Nissimov, S. Pacheva, {\it Method of squared
eigenfunction potentials in integrable hierarchies},
Comm. Math. Phys. {\bf 193} (1998) 493--525.

\bibitem{Cheng2017}
H. Chen, L. Geng, N. Li, J. Cheng,
{\it Solving the constrained modified KP hierarchy by gauge
transformations}, J. Nonlin. Math.
Phys. {\bf 26} (2019) 54--68.

\bibitem{Cheng2018}
J.P. Cheng, {\it The gauge transformation of 
modified KP hierarchy}, J. Nonlin. Math. Phys. {\bf 25} (2018) 
66–-85.

\bibitem{Z25a}
A. Zabrodin, {\it Revisiting B\"acklund-Darboux transformations 
for KP and BKP integrable hierarchies},
arXiv:2506.07208.

\bibitem{Miwa82} T. Miwa, {\it On Hirota's difference equations},
Proc. Japan Acad. {\bf 58} Ser. A (1982) 9--12.

\bibitem{TT95}
K. Takasaki, T. Takebe, 
{\it Integrable hierarchies and dispersionless limit}, 
Rev. Math. Phys. {\bf 7} (1995) 743--808,
arXiv:hep-th/9405096.

\bibitem{Shigyo13}
Y. Shigyo, {\it On addition formulae of KP, mKP
and BKP hierarchies}, SIGMA {\bf 9} (2013) 035.

\bibitem{ZM85}
V.E. Zakharov,  S.V. Manakov,
{\it Construction of higher-dimensional 
nonlinear integrable systems and of their solutions},
Funct. Anal. and Its Appl. {\bf 19} (1985) 89–-101.

\bibitem{Z18}
A. Zabrodin, {\it Lectures on nonlinear integrable equations
and their solutions}, arXiv:1812.11813.



\bibitem{DJKM81a} E. Date, M. Jimbo, M. Kashiwara and 
T. Miwa, {\it Transformation groups
for soliton equations III}, J. Phys. Soc. Japan {\bf 50} (1981) 3806--3812.

\bibitem{KL93} V. Kac and J. van de Leur, {\it The 
$n$-component KP hierarchy and 
representation theory}, in: A.S. Fokas, V.E. Zakharov (Eds.), Important Developments
in Soliton Theory, Springer-Verlag, Berlin, Heidelberg, 1993.


\bibitem{TT07} K. Takasaki and T. Takebe, {\it Universal
Whitham hierarchy, dispersionless Hirota equations
and multicomponent KP hierarchy}, Physica D {\bf 235} (2007) 109--125.

\bibitem{Teo11} L.-P. Teo, {\it The multicomponent KP hierarchy: 
differential Fay identities and Lax
equations}, J. Phys. A: Math. Theor. {\bf 44} (2011) 225201.


%%%%%%%%%%%%%%%%%%%%%%%%%%%%%%%%%%

\bibitem{Z25}
A. Zabrodin,
{\it Classical facets of quantum integrability}, to be published in
Proceedings of Beijing Summer Workshop in 
Mathematics and Mathematical Physics, 2024,
arXiv:2501.18557. 




%%%%%%%%%%%%%%%%%%%%%%%%%%%%%%%%%%%%%%%%%
%%%%%%%%%%%%%%%%%%%%%%%%%%%%%%%%%%%%%%%%%

%\bibitem{HBC89}
%implicit solution of the Harry Dym eqiation and its connections with 
%the Korteveg-de Vries equation}, J. Phys. A: Math. Gen. {\bf 22}
%(1989) 241--255.

%\bibitem{HJN16}
%J. Hietarinta, N. Joshi, F.W. Nijhoff, {\it Discrete Systems and
%Integrability}, Cambridge Texts in Applied Mathematics, Cambridge
%University Press, 2016.


%\bibitem{DJM82}  E. Date, M. Jimbo and T. Miwa, {\it Method for
%generating discrete soliton equations I, II}, Journ. Phys. Soc. Japan
%{\bf 51} (1982) 4116--4131.


%\bibitem{OHTI93}
%Y. Ohta, R. Hirota, S. Tsujimoto, T. Imai, {\it Casorati and
%discrete Gram type determinant representations of solutions to
%the discrete KP hierarchy}, J. Phys. Soc. Japan {\bf 62} (1993)
%1872--1886.





\end{thebibliography}
\end{document}